\documentclass{LMCS}

\usepackage{enumerate}
\usepackage{graphicx}
\usepackage{latexsym}
\usepackage{amsmath,amssymb,amstext}
\usepackage{comment}
\usepackage{pifont}
\usepackage{color}
\usepackage{wrapfig}
\usepackage{mathtools}
\usepackage{stmaryrd}

\usepackage{tikz}
\usetikzlibrary{arrows,automata,backgrounds,calc,chains,fadings,folding,decorations.fractals,decorations.pathreplacing,fit,patterns,positioning,mindmap,shadows,shapes.geometric,shapes.symbols,through,trees,plotmarks}

\input{main.def}

\usepackage{enumerate}
\usepackage{hyperref}

\def\doi{6 (2:??) 2010}
\lmcsheading%
{\doi}
{1--39}
{}
{}
{Oct.~29, 2009}
{Jun.~29, 2010}
{}   

\begin{document}

\title{Transforming Outermost into Context-Sensitive Rewriting\rsuper *}

\author[J.~Endrullis]{J\"{o}rg Endrullis}
\address{
    Vrije Universiteit Amsterdam,
    De Boelelaan 1081a,
    1081 HV Amsterdam,
    The Netherlands
}
\email{joerg@few.vu.nl, diem@cs.vu.nl}

\author[D.~Hendriks]{Dimitri Hendriks}
\address{\vskip-6 pt}
%\address{
%    Vrije Universiteit Amsterdam,
%    De Boelelaan 1081a,
%    1081 HV Amsterdam,
%    The Netherlands
%}
%\email{diem@cs.vu.nl}

\titlecomment{{\lsuper *}%
  This is a modified and extended version of~\cite{endr:hend:09}
  which appeared in the proceedings of RTA~2009.
}

\begin{abstract}
  We define two transformations from term rewriting systems (TRSs) 
  to context-sensitive TRSs in such a way that 
  termination of the target system 
  implies outermost termination of the original system.
  In the transformation based on `context extension', 
  each outermost rewrite step is modeled
  by exactly one step in the transformed system.
  This transformation turns out to be complete for the class of left-linear TRSs.
  The second transformation is called `dynamic labeling'
  and results in smaller sized context-sensitive TRSs.
  Here each modeled step is adjoined with
  a small number of auxiliary steps.
  
  As a result state-of-the-art termination methods  %%and automated termination provers 
  for context-sensitive rewriting become available 
  for proving termination of outermost rewriting.
  Both transformations have been implemented in \jambox,
  making it the most successful tool 
  in the category of outermost rewriting 
  of the annual termination competition of 2008. %~\cite{term:comp:08}.
\end{abstract}

\maketitle

\section{Introduction}

\newcommand{\ffxnr}{0}%
\newcommand{\trsffx}{\iatrs{\ffxnr}}%
\newcommand{\setffx}{\iaset{\ffxnr}}%
\newcommand{\algffx}{\iaalg{\ffxnr}}%

Termination is a key aspect of program correctness,
and therefore a widely studied subject in term rewriting and program verification.
While termination is undecidable in general, 
various automated techniques have been developed for proving termination.
One of the most powerful techniques is the method of dependency pairs~\cite{arts:gies:00}.
In~\cite{alar:guti:luca:06} dependency pairs for context-sensitive rewriting have been introduced,
and in~\cite{alar:guti:luca:08} the dependency pair framework~\cite{gies:thie:schn:04,thie:07,guti:luca:10} 
has been extended to context-sensitive rewriting,
thereby extending the class of context-sensitive TRSs
for which termination can be shown automatically.
Context-sensitive rewriting~\cite{luca:98} is a restriction on term rewriting 
where rewriting in some fixed arguments of function symbols is disallowed.
It offers a flexible paradigm
to analyze properties of rewrite strategies,
in particular of (lazy) evaluation strategies employed by functional programming languages.

In this paper context-sensitive rewriting is the target formalism 
for a transformational approach to the problem of outermost termination,
that is, termination with respect to outermost rewriting.
Outermost rewriting is a rewriting strategy where a redex may be contracted as long as it
is not a proper subterm of another redex occurrence.
The main reason for studying outermost termination 
%(termination with respect to outermost rewriting)
is its practical relevance:
lazy functional programming languages like Miranda~\cite{miranda}, Haskell~\cite{haskell} or Clean~\cite{clean},
are based on outermost rewriting as an evaluation strategy,
and in implementations of rewrite logic such as Maude~\cite{maude} and CafeOBJ~\cite{cafeobj}, 
outermost rewriting is an optional strategy.
%can be specified.

To illustrate outermost rewriting, and the transformations we propose, 
we consider the term rewriting system~$\trsffx$ consisting of the following rules:
\begin{align}
  \ta &\to \f{\ta} & \f{\f{x}} &\to \tb 
  \tag{$\trsffx$}
\end{align}
Clearly, this system is not terminating, as witnessed by the infinite rewrite sequence:
\[
  \ta \to \f{\ta} \to \f{\f{\ta}} \to \f{\f{\f{\ta}}} \to \ldots
  %\punc,
\]
However, $\trsffx$ is outermost terminating. 
Indeed, the third step in the rewrite sequence above is not an outermost step,
since the contraction takes place inside another redex.
The only (maximal) outermost rewrite sequence the term $\ta$ admits is:
\begin{align}
  \ta \outred \f{\ta} \outred \f{\f{\ta}} \outred \tb
  \label{ex:ffx:outermost:sequence}
\end{align}

The contribution of the present paper consists of two transformations of arbitrary TRSs
into context-sensitive TRSs (henceforth also called `\csTRS{s}')
in such a way that rewriting in the \csTRS{} 
corresponds to outermost rewriting in the original TRS.
As a result, advanced termination techniques for \csTRS{s}
become available for proving outermost termination.
Automated termination provers for \csTRS{s} can directly (without modification, only preprocessing) 
be used for proving outermost termination.
One of the transformations turns out to be complete for the class of quasi-left-linear TRSs,
a generalized form of left-linear TRSs, see~\cite{raff:zant:09}.
In other words, termination of the resulting \csTRS{}
is equivalent to outermost termination of the original system.
%The other transformation is complete with respect to all correctly labeled terms.

The transformations are comprised of a variant of semantic labeling~\cite{zant:95}. 
In semantic labeling the function symbols in a term are labeled 
by the interpretation of their arguments (or a label depending on these values) 
according to some given semantics.
We employ semantic labeling to mark symbols at redex positions,
and then obtain a \csTRS{} by defining a replacement map that disallows rewriting inside arguments
of marked symbols.

We illustrate our use of semantic labeling by means of the TRS~$\trsffx$ given above. 
We choose an algebra with values $0$ and $1$,
indicating the presence of the symbol~$\tf$:
\begin{align}
  \algffx = \pair{\{0,1\}}{\sinterpret}
  &&
  \interpret{\ta} = \interpret{\tb} = 0
  &&
  \funap{\interpret{\tf}}{x} = 1 \quad \text{for $x\in\{0,1\}$}
  \tag{$\algffx$}
\end{align}
We write $\rmark{\tf}$, and say that `$\tf$ is marked', if the value of its argument is $1$,
and just $\tf$ if the value is $0$. 
The symbol $\ta$ is a redex, and hence it is always marked, while $\tb$ never is.
If $\tf$ is marked it corresponds to 
a redex position with respect to the rule $\f{\f{x}} \to \tb$.
For example the term $\f{\f{\f{\ta}}}$
is labeled as
$\funap{\rmark{\tf}}{\funap{\rmark{\tf}}{\f{\rmark{\ta}}}}$.
We obtain a \csTRS{} by forbidding rewriting inside the argument of the symbol~$\rmark{\tf}$.
Since $\rmark{\ta}$ is a constant, there is nothing to be forbidden.
Hence for correctly labeled terms, rewriting inside redex occurrences is disallowed,
and this corresponds to the strategy of outermost rewriting.

In order to rewrite labeled terms we have to label the rules of a TRS.
Simply labeling both sides of a rule does not always work.
For example, when we label the rules of~$\trsffx$ using the algebra~$\algffx$,
we obtain the following~\csTRS: %, which we denote by~\ref{ex:ffx:lab}:
\begin{align}
  \rmark{\ta} &\to \f{\rmark{\ta}} &
  \funap{\rmark{\tf}}{\f{x}} &\to \tb
  &
  \funap{\rmark{\tf}}{\funap{\rmark{\tf}}{x}} &\to \tb
  \tag{$\overline{\atrs}_0$}
  %\punc.
  \label{ex:ffx:lab}
\end{align}
This system has two instances of the second rule, one for each possible value assigned to the variable $x$. 
Now, despite of the fact that the original TRS is outermost terminating,
the labeled system~\ref{ex:ffx:lab} still admits an infinite rewrite sequence:
\begin{align}
  \rmark{\ta} 
  \mured \f{\rmark{\ta}} 
  \mured \f{\f{\rmark{\ta}}} 
  \mured \f{\f{\f{\rmark{\ta}}}} \mured \ldots
  \label{ex:ffx:infred}
\end{align}
The reason is that the term $\f{\f{\rmark{\ta}}}$ 
is not correctly labeled, as its root symbol $\tf$ should have been marked. %with $\star$.
In~\cite{zant:95} this problem is avoided by allowing labeling only with \emph{models}.
Roughly speaking, an algebra is a model for a TRS~$\atrs$ 
if left and right-hand sides of all rewrite rules of $\atrs$ 
have equal interpretations.
However, this requirement is too strict for the purpose of marking redexes, 
because contraction of a redex at a position $p$ may create a redex above $p$ 
in the term tree, as exemplified by~\eqref{ex:ffx:infred}.
In fact, for~$\trsffx$ there exists \emph{no} model  
which is able to distinguish between redex and non-redex positions.
Let us explain.
The rewrite step $\f{\ta} \to \f{\f{\ta}}$
creates a redex at the top.
The term $\f{\ta}$ is not a redex, 
and therefore its root symbol $\tf$ should not be marked.
On the other hand $\f{\f{\ta}}$ is a redex 
and so the outermost $\tf$ has to be marked.
The change of the labeling of a context (here $\f{\contexthole}$) 
implies that the interpretation of its arguments $\ta$ and $\f{\ta}$ cannot be the same.
Therefore we cannot require the rule $\ta \to \f{\ta}$ to preserve the interpretation.

To that end, we generalize this notion of model 
and relax the condition $\interpret{\ell} = \interpret{r}$ to:
\begin{align}
  \myex{n}{\interpret{\contextfill{\acontext}{\ell}} = \interpret{\contextfill{\acontext}{r}}}%\punc,
  \quad
  \text{for all contexts $\acontext$ of depth~$\geq n$}
  \label{eq:cmodel:req}
\end{align}
Thus rules are allowed to change the interpretation as long as 
the effect is limited to contexts of a bounded depth.
We call this depth \emph{the \cdepth{} of $\ell \to r$} and denote it by $\cdpth{\aalg}{\ell \to r}$.
As it turns out, algebras satisfying this weaker requirement~\eqref{eq:cmodel:req}, 
are strong enough to recognize redex positions. 
Such algebras we will call \emph{\cmodel{s}}.

The algebra $\aalg_0$ given above is a \cmodel{} for the TRS~$\trsffx$.
As opposed to models, for \cmodel{s} it is no longer sufficient to simply label the rules.
This is demonstrated by the rewrite sequence~\eqref{ex:ffx:infred} in the \csTRS~\ref{ex:ffx:lab}:
an application of the rule $\rmark{\ta} \to \f{\rmark{\ta}}$ in
the term $\f{\rmark{\ta}}$ creates the incorrectly labeled term 
$\f{\f{\rmark{\ta}}}$.

Therefore, in order to preserve correct labeling, 
labels in the context of the original rewrite step sometimes have to be updated.
We present two solutions to this problem:
the transformation of \emph{context extension} 
and the transformation of \emph{dynamic labeling}.

In the transformation based on context extension~\cite{endr:hend:09},
worked out in Section~\ref{sec:cxtext},
the update of semantic labels is established
by prefixing appropriate contexts to both sides of a rewrite rule. 
The depth of these contexts is bounded by the \cdepth{} of the rule.
Thus, the update of the labels is coded within the context, 
and no additional rewrite steps are needed.
As a result of that, every outermost rewrite step in the original system 
is modeled by exactly one rewrite step in the transformed \csTRS.
A disadvantage of the transformation, however, 
is that the resulting \csTRS{} can have a large number of rules
arising from the prepending of contexts
in combination with semantic labeling.

An alternative solution (and new with respect to~\cite{endr:hend:09})
is dynamic labeling, described in Section~\ref{sec:dynlab}:
instead of extending rules with contexts we now use rewriting 
to propagate the changed information upward in the term tree.
With respect to context extension, 
this approach results in a smaller number of rules of the transformed system.
On the other hand, the property of the context extension of a one--to--one correspondence of the rewrite steps,
is now weakened to a one--to--$m$ correspondence where $m \le 1 + \cdpth{\aalg}{\ell \to r}$.
This means that an outermost rewrite step is modeled by one step in the transformed system
plus a number of auxiliary steps necessary for updating the labels,
and this number is bounded by the \cdepth{} of the corresponding rule.
In most practical cases this value is typically small ($\leq 2$).
This is shown in Section~\ref{sec:evaluation} where we evaluate the implemented transformations.

We illustrate the two transformations by means of our running example,
the TRS~$\trsffx$ together with the algebra~$\aalg_0$ 
which forms a \cmodel{} for~$\trsffx$. %; both were introduced above.
The algorithm based on context extension
transforms $\trsffx$ into the following \csTRS~\ref{ex:ffx:dynamic}, 
which truthfully simulates outermost rewriting in $\trsffx$:
\begin{gather}
  \begin{aligned}
  \funap{\svoodoolabel{\tf}{}}{\rmark{\ta}} &\to \funap{\rmark{\tf}}{\funap{\svoodoolabel{\tf}{}}{\rmark{\ta}}}
  &&&&&&&
  \funap{\svoodoolabel{\topsymb}{}}{\funap{\rmark{\tf}}{\funap{\svoodoolabel{\tf}{}}{x}}} 
  &\to \funap{\svoodoolabel{\topsymb}{}}{\tb} \\
  \funap{\svoodoolabel{\topsymb}{}}{\rmark{\ta}} &\to \funap{\svoodoolabel{\topsymb}{}}{\funap{\svoodoolabel{\tf}{}}{\rmark{\ta}}} 
  &&&&&&&
  \funap{\svoodoolabel{\topsymb}{}}{\funap{\rmark{\tf}}{\funap{\rmark{\tf}}{x}}} 
  &\to \funap{\svoodoolabel{\topsymb}{}}{\tb}
  %\punc.
  \end{aligned}
  \tag{$\cxtext{\aalg_0}{\semlab}{\atrs_0}$}
  \label{ex:ffx:dynamic}
\end{gather}
The rule $\f{\rmark{\ta}} \to \funap{\rmark{\tf}}{\f{\rmark{\ta}}}$
is obtained from prepending the context $\f{\contexthole}$ to $\ta \to \f{\ta}$. 
This enables correct updating of the labeling of the context during rewriting.
Because we still have to allow rewrite steps $\ta \to \f{\ta}$ of the original TRS 
at the top of a term, %.
we extend the signature with a unary function symbol $\topsymb$
which represents the top of a term.
Thus when prepending contexts we include $\funap{\topsymb}{\contexthole}$, giving rise to the rule
$\funap{\topsymb}{\rmark{\ta}} \to \funap{\topsymb}{\f{\rmark{\ta}}}$.
The necessity of the symbol $\topsymb$ becomes apparent especially
when we consider the rule $\f{\f{x}} \to \tb$.
Here prepending the context $\f{\contexthole}$ is not even an option
since $\f{\f{\f{x}}} \to \f{\tb}$
is not an outermost rewrite step;
this rule can only be applied at the top of a term.
Hence we get the two rules displayed on the right, 
one for each possible interpretation of the variable~$x$.

The second algorithm we define, that of dynamic labeling,
transforms $\trsffx$ (using $\algffx$)
into the following \csTRS~, which we denote by $\dynlab{\aalg_0}{\semlab}{\text{$\trsffx$}}$\,:
\begin{align}
  \rmark{\ta} &\to \relabel{0}{1}{\f{\rmark{\ta}}}
  &&&
  \f{\relabel{0}{1}{x}} &\to \funap{\rmark{\tf}}{x} 
  \notag \\
  \funap{\rmark{\tf}}{\f{x}} &\to \relabel{1}{0}{\tb}
  &&&
  \funap{\topsymb}{\relabel{0}{1}{x}} &\to \funap{\topsymb}{x} 
  %\funap{\svoodoolabel{\topsymb}{0}}{\relabel{0}{1}{x}} &\to \funap{\svoodoolabel{\topsymb}{1}}{x} 
  \tag{$\dynlab{\aalg_0}{\semlab}{\text{$\trsffx$}}$} \\
  \funap{\rmark{\tf}}{\funap{\rmark{\tf}}{x}} &\to \relabel{1}{0}{\tb}
  &&&
  \funap{\topsymb}{\relabel{1}{0}{x}} &\to \funap{\topsymb}{x}
  %\funap{\svoodoolabel{\topsymb}{1}}{\relabel{1}{0}{x}} &\to \funap{\svoodoolabel{\topsymb}{0}}{x}
  \notag 
\end{align}
where rewriting beneath the redex symbol $\rmark{\tf}$ 
and the symbols $\srelabel{0}{1}$ and $\srelabel{1}{0}$ is disallowed.
Displayed on the left, we recognize the original rules.
Since left and right-hand side of the original rule $\ta \to \f{\ta}$ 
have distinct interpretations ($0$ and $1$), 
in $\dynlab{\aalg_0}{\semlab}{\text{$\trsffx$}}$
the right-hand side is prefixed with the symbol $\srelabel{0}{1}$.
By application of the relabeling rules (displayed on the right), 
this symbol moves upward to take care of the update of labels in the context 
of the original rule application.
Likewise, the rule $\f{\f{x}} \to \tb$ 
(of which there are two versions in $\dynlab{\aalg_0}{\semlab}{\text{$\trsffx$}}$, 
one for each value $x$ can be assigned to)
means a change of interpretation, and relabeling the context is necessary.
In this example, each original step is accompanied by exactly one relabel step.
The $\ssrelabel$ symbols dissolve after one such step.
The only (maximal) rewrite sequence from the term 
$\funap{\topsymb}{\rmark{\ta}}$ in $\dynlab{\aalg_0}{\semlab}{\text{$\trsffx$}}$ is:
%displayed below.
\begin{align*}
  \funap{\topsymb}{\rmark{\ta}} 
  & 
  \mured \funap{\topsymb}{\relabel{0}{1}{\f{\rmark{\ta}}}}
  \mured \funap{\topsymb}{\f{\rmark{\ta}}}
  \mured \funap{\topsymb}{\f{\relabel{0}{1}{\f{\rmark{\ta}}}}}
  \\
  &
  \mured \funap{\topsymb}{\funap{\rmark{\tf}}{\f{\rmark{\ta}}}}
  \mured \funap{\topsymb}{\relabel{1}{0}{\tb}}
  \mured \funap{\topsymb}{\tb}
\end{align*}
Notice the correspondence with the outermost rewrite sequence~\eqref{ex:ffx:outermost:sequence}.
%from the term~$\ta$ in $\trsffx$.

Clearly, semantic labeling increases the number of rules and the number of symbols of a TRS.
This results in a larger search space for finding termination proofs,
and hence may lead to exhaustion of time or memory resources.
On the other hand, one can say that semantic labeling 
does not complicate termination proofs, in the sense that % termination 
proofs for the unlabeled system carry over to the labeled one:
whenever $\atrs'$ is a labeling of a TRS $\atrs$
and $\aalg = \quadruple{\aset}{\interpret{\cdot}}{{\alggt}}{{\algge}}$
is a monotone $\asig$-algebra~\cite{endr:wald:zant:08}
which proves termination of $\atrs$, 
then the extension of $\sinterpret$ to the labeled signature $\asig'$
by defining $\interpret{\svoodoolabel{\tf}{\lambda}} \defdby \interpret{\tf}$
for every $\tf \in \asig$ and label~$\lambda$,
yields a monotone $\asig'$-algebra witnessing termination of $\atrs'$.
Apart from this, % termination proofs carrying over to the labeled system,
the labeled systems often allow for simpler proofs, 
because the enriched signature provides for more freedom in the choice of interpretations, see~\cite{zant:95}.
As the transformations presented here are based on a variant of semantic labeling,
they inherit these properties from semantic labeling.

The two transformations have been implemented by the first author 
in the termination prover \jambox~\cite{jambox}.
Notwithstanding the increased number of rules by semantic labeling and context extension,
\jambox{} performs efficiently on the set of examples from 
the Termination Problem Database (TPDB~\cite{term:comp:08}),
and was best in proving termination 
in the category of outermost rewriting of the termination competition
of 2008~\cite{term:comp:08}, see Section~\ref{sec:evaluation}. %Table~\ref{table:score}.

\subsection*{Related work.}
The first tool for proving 
outermost termination was {\cariboo}~\cite{{fiss:gnae:kirch:02b},{gnae:kirc:09}}. % using induction.
\cariboo{} is a stand-alone tool, and its method is based on induction.

For the idea of a transformational approach to outermost termination
in order to make use of the power of termination provers 
we were inspired by~\cite{raff:zant:09},
which in turn is based on ideas in~\cite{gies:midd:04}. 
In~\cite{raff:zant:09} the signature is enriched with unary symbols 
$\msf{top}$, $\msf{up}$, and $\msf{down}$ % and $\msf{block}$
and the TRS is extended with `anti-matching' rules 
such that $\funap{\msf{down}}{t}$ is a redex
if and only if $t$ is not a redex with respect to the original TRS.
The idea is that the symbol $\msf{down}$ is moved down in the term tree 
as long as no redex is encountered.
Once a redex is encountered, a rewrite step is performed,
and the symbol $\msf{down}$ is replaced by $\msf{up}$,
which then moves upwards again to the top of the term, 
marked by $\msf{top}$.
This transformation is implemented in the tool~\trafo~\cite{raff:zant:09},
participant in the termination competion of 2008~\cite{term:comp:08}.

Based on a similarly elegant idea, Thiemann~\cite{thie:09} defines a complete transformation 
from outermost to innermost rewriting, which is implemented in \aprove. 
For traversal to the redex positions, rules of the form
$\funap{\msf{down}}{\funap{\msf{isRedex}}{\f{\ldots}}}
 \to \f{\ldots,\funap{\msf{down}}{\funap{\msf{isRedex}}{\ldots}},\ldots}$
are used.
In order to simulate outermost rewriting and to prevent from moving inside redexes, 
rules $\funap{\msf{isRedex}}{\ell} \to \funap{\msf{up}}{r}$ 
are added for every rule $\ell \to r$ of the original TRS.
Then, by the innermost rewriting strategy, 
the latter rules have priority over the traversal rules, 
whenever an original redex is encountered.

The simplicity of both approaches is attractive, 
but the yo-yoing effect in the resulting TRSs 
makes that the original outermost rewrite steps 
are `hidden' among a vast amount of auxiliary steps. 
This increases derivational complexity,
and makes it hard for automated termination provers
to find proofs for the transformed systems.

The present paper is a modified and extended version of~\cite{endr:hend:09}.
In particular, we introduce a novel approach for proving outermost termination:
dynamic labeling (Section~\ref{sec:dynlab}).
We stress that the number of extra relabeling steps introduced 
in the dynamic labeling of a system is typically small 
and bounded by the \cdepth{}
of the applied rewrite rule. 

%\subsection*{Overview.}
%In Section~\ref{sec:prelims}\dots\
%In Section~\ref{sec:label}\dots\
%In Section~\ref{sec:static}\dots\
%In Section~\ref{sec:cxtext}\dots\
%In Section~\ref{sec:dynlab}\dots\
%In Section~\ref{sec:constructing}\dots\
%In Section~\ref{sec:minimizing}\dots\
%In Section~\ref{sec:minmax}\dots\
%In Section~\ref{sec:evaluation}\dots\
%In Section~\ref{sec:discussion}\dots\

\section{Preliminaries}\label{sec:prelims}
For a general introduction to term rewriting 
and to context-sensitive rewriting,
we refer to~\cite{terese:03} and~\cite{luca:98}, respectively. 
Here we repeat some of the main definitions, for the sake of completeness, and to fix notations. 

A \emph{signature} $\asig$ is a non-empty set of symbols each having a fixed \emph{arity},
given by a mapping $\sarity \funin \asig \to \nat$. 
We write $\arity{\tf}$ for the arity of $\tf\in\asig$,
and we define $\asig_n \defdby \{ \tf\in\asig \where \arity{\tf} = n \}$.
%% cut 1
Given $\asig$ and a set $\avars$ of variables, 
the \emph{set $\term{\asig}{\avars}$ of terms over $\asig$}
is the smallest set satisfying:
$\avars \subseteq \term{\asig}{\avars}$, and
$\tf(t_1,\dots,t_n) \in \term{\asig}{\avars}$ if $\tf \in \asig$ of arity $n$ 
and $t_i \in \term{\asig}{\avars}$ for all $1 \leq i \leq n$.
%% For \emph{constants} $\fun{c}\in\asig_0$ we simply write $\fun{c}$ for the term $\funap{\fun{c}}{}$.
We use $x,y,z,\ldots$ to range over variables, 
and write $\vars{t}$ for the set of variables occurring in a term $t$. 

The set of \emph{positions} $\pos{t} \subseteq \nat^*$ of a term $t \in \term{\asig}{\avars}$
is defined as follows:
$\pos{x} = \{\lstemp\}$ for variables $x \in \avars$
and
$\pos{\f{t_1,\ldots,t_n}} = 
 \{\lstemp\} \cup \{\lstconcat{i}{\apos} \where 1 \le i \le n,\, \apos \in \pos{t_i}\}$
 for symbols $\fun{f}\in\asig_n$.
We write $\rootsymb{t}$ to denote the root symbol (or variable) of $t$,
and $\subtrmat{t}{p}$ for the subterm of $t$ rooted at position $p$.
Then $\symbat{t}{p}$ is the symbol at position $p$ in $t$.
% We write $\symbat{t}{p}$ to denote the root symbol of $\subtrmat{t}{p}$, 
% the subterm of $t$ rooted at $p$, and we write $\rootsymb{t}$ for the root symbol of $t$, 
% that is $\rootsymb{t} = \symbat{t}{\posemp}$.

A \emph{substitution} $\asubst$ is a map $\asubst \funin \avars \to \term{\asig}{\avars}$ from variables to terms.
For terms $t \in \term{\asig}{\avars}$ and substitutions $\asubst$,
$\subst{\asubst}{t}$ is inductively defined by
$\subst{\asubst}{x} = \funap{\asubst}{x}$ for $x \in \avars$,
and 
\(
  \subst{\asubst}{\f{t_1,\ldots,t_n}} 
  = \f{\subst{\asubst}{t_1},\ldots,\subst{\asubst}{t_n}}
\) for $\fun{f}\in\asig_n$, and $t_1,\ldots,t_n\in\ter{\asig}{\avars}$.

Let $\contexthole$ be a fresh symbol, i.e., $\contexthole \not\in \asig\join\avars$.
A \emph{context} $\acontext$ is a term from $\term{\asig}{\avars\join\{\contexthole\}}$
containing precisely one occurrence of $\contexthole$.
By $\contextfill{\acontext}{s}$ we denote the term $\subst{\asubst}{\acontext}$
where $\funap{\asubst}{\contexthole} = s$ and $\funap{\asubst}{x} = x$ for all $x \in \avars$.
We use $\bfunap{\contexts}{\asig}{\avars}$ to denote the set of contexts over $\asig$ and $\avars$. 
We write $\vars{\acxt}$ with $\acxt \in \bfunap{\contexts}{\asig}{\avars}$
to denote the set of variables of $\acxt$ excluding~$\cxthole$.
The \emph{depth} of a context $C$ is defined as the length $\lstlength{p}$ of the position $p$ 
at which $\contexthole$ resides, that is, the position $p$ such that 
$\symbat{\acontext}{\apos} = \contexthole$.

A \emph{term rewriting system (TRS)} over $\asig$ %, $\avars$ 
is a finite set of pairs $\pair{\ell}{r} \in \term{\asig}{\avars}\times\term{\asig}{\avars}$,
called \emph{rewrite rules} and written as $\ell \to r$,
for which the \emph{left-hand side}~$\ell$ is not a variable ($\ell \not\in \avars$)
and all variables in the \emph{right-hand side}~$r$ occur in $\ell$: 
$\vars{r} \subseteq \vars{\ell}$. 
For a TRS~$\atrs$ we define $\to_{\atrs}$, the \emph{rewrite relation}
induced by $\atrs$ as follows:
For terms $s, t \in \term{\asig}{\avars}$ we write $s \to_{\atrs} t$,
or just $s \to t$ if $\atrs$ is clear from the context,
if there exists a rule $\ell \to r \in \atrs$, a substitution $\asubst$
and a context $\acontext \in \bfunap{\contexts}{\asig}{\avars}$
such that
$s = \contextfill{\acontext}{\subst{\asubst}{\ell}}$
and
$t = \contextfill{\acontext}{\subst{\asubst}{r}}$;
we sometimes write $s \redat{\atrs,\apos} r$ 
to explicitly indicate the rewrite position $p$,
i.e., when $\symbat{\acontext}{\apos} = \contexthole$.
A term of the form $\subst{\asubst}{\ell}$, 
for some rule $\ell \to r \in \atrs$, and a substitution $\asubst$,
is called a \emph{redex}.

For terms $s$ and $t$, we say that \emph{$s$ outermost rewrites to $t$} at a position $p\in\pos{s}$,
denoted by $s \outred_{\atrs,p} t$,
if $s \to_{\atrs,p} t$ and for all positions $p'$ strictly above $p$ 
(i.e., $p'$ a proper prefix of $p$) we have that $\subtrmat{s}{p'}$ 
is not a redex with respect to $\atrs$.  

A binary relation ${\alggt} \subseteq \aset \times \aset$ over a set $\aset$
is called \emph{well-founded} 
if no infinite decreasing sequence $a_1 \alggt a_2 \alggt a_3 \alggt \ldots $ exists.
A TRS $\atrs$ is called \emph{terminating} or \emph{strongly normalizing},
denoted by $\SNr{\atrs}$, if $\to_{\atrs}$ is well-founded.

A mapping $\samumap \funin \asig \to \powerset{\nat}$
is called a \emph{replacement map (for $\asig$)} if for all symbols~$\tf \in \asig$ 
we have $\amumap{\tf} \subseteq \{1,\ldots,\arity{\tf}\}$. 
% We use $\mumaps{\asig}$ to denote the set of all replacement maps for $\asig$.
When we define a replacement map $\samumap$, the case for constants $\ta\in\asig$ is left implicit,
as we always have $\amumap{\ta} = \setemp$.
A \emph{context-sensitive term rewriting system (\csTRS)} is a pair $\pair{\atrs}{\samumap}$
consisting of a TRS $\atrs$ and a replacement map $\samumap$. % \in \mumaps{\asig}$.
The set of \emph{$\samumap$-replacing positions $\mupos{t}$} of a term $t \in \term{\asig}{\avars}$
is defined by $\mupos{x} = \{\lstemp\}$ for $x \in \avars$ and
$\mupos{\f{t_1,\ldots,t_n}} = 
 \{\lstemp\} \cup \{\lstconcat{i}{\apos} \where i \in \amumap{\tf},\, \apos \in \mupos{t_i}\}$
 for $\fun{f}\in\asig_n$ and $t_1,\ldots,t_n\in\ter{\asig}{\avars}$.

In context-sensitive term rewriting
only redexes at \mbox{$\samumap$-replacing} positions are contracted:
we say \emph{$s$ $\samumap$-rewrites to $t$}, and denote it by $s \muredr{\atrs} t$ %or $s \to_{\pair{\atrs}{\samumap}} t$,
whenever $s \redat{\atrs,\apos} t$ with $\apos \in \mupos{s}$.
For instance, consider the system~$\atrs$ consisting of the single rule:
\begin{align*}
  \ta \to \bfunap{\fun{cons}}{\tb}{\ta}
\end{align*}
and let $\samumap$ be given by $\amumap{\fun{cons}} = \{1\}$.
%% and $\amumap{\ta} = \amumap{\tb} = \setemp$.
Then, obviously $\atrs$\, is non-terminating. %(it produces $b$'s indefinitely)
On the other hand, the context-sensitive TRS~$\pair{\atrs}{\samumap}$ is terminating,
because the replacement map of the symbol $\fun{cons}$ allows rewriting only in its first argument.

We conclude this section by defining some non-standard notions.
\begin{definition}\label{def:thin:contexts}
  A \emph{thin context} is a context that has at every depth
  at most one symbol from \mbox{$\asig \cup \{\cxthole\}$}; 
  all other symbols are variables.
\end{definition}

\begin{example}
  $\f{\f{x,\g{\cxthole}},y}$ is a thin context,
  whereas $\f{\g{\cxthole},\h{x}}$ is not,
  since $\tg$ and $\sh$ are at the same depth.
\end{example}

\begin{definition}[\cite{ster:midd:08}]\label{def:flat:contexts}
  A \emph{flat context}
  is a context $\acxt \in \bfunap{\contexts}{\asig}{\avars}$
  of the form:
  \[
    \acxt = \f{x_1,\ldots,x_{j-1},\cxthole,x_{j+1},\ldots,x_n}
  \]
  where $\tf\in\asig_n$ with $n > 0$ and $x_1,x_2,\ldots\in\avars$ are pairwise distinct variables.
  % We denote the set of flat contexts by $\bfunap{\contexts^\flat}{\asig}{\avars}$.
  For a  term $t\in\ter{\asig}{\avars}$ we say that %a flat context 
  $\acxt$ is \emph{fresh for~$t$} if $\vars{\acxt} \cap \vars{t} = \setemp$.
  We denote the set of flat contexts fresh for $t$ by $\bfunap{\contexts^\flat_t}{\asig}{\avars}$.
  %  \begin{equation}
  %  \contexts_t^\flat
  %  = 
  %  \big\{\, 
  %    \f{x_1,\ldots,x_{j-1},\cxthole,x_{j+1},\ldots,x_{\arity{\tf}}}
  %    \where f \in \asig_{\topsymb} ,\; j \in \{1,\ldots,{\arity{\tf}}\}
  %  \,\big\}
  %  \label{eq:flat:context}
  %  \end{equation}
  %  where $x_1,x_2,\ldots\in\avars$ such that $x_i\not\in\vars{t}$.
\end{definition}

\begin{definition}[\cite{raff:zant:09}]
  A TRS $\atrs$ is called \emph{quasi-left-linear}
  if every non-linear left-hand side of a rule in $\atrs$
  is an instance of a linear left-hand side from $\atrs$.
\end{definition}

\begin{example}
  The following TRS is quasi-left-linear (and outermost terminating):
  \begin{align*}
    \gb{\f{x}}{x} &\to \gb{\gb{x}{x}}{x} &
    \gb{x}{y} &\to y
  \end{align*}
\end{example}

\section{Generalizing Models to \cmodel{s}}\label{sec:label}

In outermost rewriting the only redexes which are allowed to be
rewritten are those which are not nested within any other redex occurrence.
We model this strategy by context-sensitive rewriting
with the use of semantic labeling: we mark the symbols which are the root of a redex
in order to disallow rewriting within that redex.
We first recall the definition of semantic labeling and of models from~\cite{zant:95},
and then generalize these to fit our purpose.

\begin{definition}
  A \emph{$\asig$-algebra} $\aalg = \pair{\aset}{\sinterpret}$ 
  consists of a non-empty set $\aset$, called the \emph{domain of $\aalg$},
  and for each $n$-ary symbol $\tf \in \asig$ a function 
  $\interpret{\tf} \funin \aset^n \to \aset$, 
  called the \emph{interpretation of $\tf$}.
  Given an \emph{assignment} $\alpha \funin \avars \to \aset$ of the variables to $\aset$, 
  the \emph{interpretation of a term $t \in \term{\asig}{\avars}$ with respect to $\alpha$} 
  is denoted by $\interpreta{t}{\alpha}$ and inductively defined by:
  \begin{align*}
    \interpreta{x}{\alpha} &= \funap{\alpha}{x}
    &
    \interpreta{\f{t_1,\ldots,t_n}}{\alpha} 
    &= \funap{\interpret{\tf}}{\interpreta{t_1}{\alpha},\ldots,\interpreta{t_n}{\alpha}}
  \end{align*}
  where $x\in\avars$, $\fun{f}\in\asig_n$, and $t_i\in\ter{\asig}{\avars}$ for $1 \leq i \leq n$. 
  For substitutions $\asubst \funin \avars \to \term{\asig}{\avars}$,
  we write $\interpreta{\asubst}{\alpha}$ for the function 
  $\mylam{x}{\interpreta{\funap{\asubst}{x}}{\alpha}}$.
  For ground terms $t \in \term{\asig}{\setemp}$ 
  and ground substitutions $\asubst\funin\avars\to\term{\asig}{\setemp}$
  we write $\interpret{t}$ and $\interpret{\asubst}$ for short. %, respectively.
  We usually write $\aalg$\, for both the algebra and its domain,
  and we use $\sinterpret$ to denote the interpretation function of $\aalg$.
\end{definition}

\begin{lemma}\label{lem:xyz}
  Let $\aalg$ be a $\asig$-algebra,
  $\alpha\funin{\avars\to\aalg}$ an assignment,
  and $\asubst\funin{\avars\to\term{\asig}{\avars}}$ a substitution.
  Then, for all terms $t\in\term{\asig}{\avars}$:
  \( %\begin{align*}
    \interpreta{\subst{\asubst}{t}}{\alpha} 
    % & 
    = \interpreta{t}{\interpreta{\asubst}{\alpha}}
    %\punc.
  \)\,. \qed
  %\end{align*}
\end{lemma}
%\begin{proof}
%\end{proof}

For completeness of the transformation of context extension (Theorem~\ref{thm:abstract:maxcomplete}), 
it is important that there exist no `junk' elements in the $\asig$-algebra,
that is, elements $a$ for which there are no ground terms $t$ such that $\interpret{t} = a$.
See Example~\ref{ex:whycore}.
For that reason we restrict $\asig$-algebras to `core' $\asig$-algebras
whose domain equals the set of all interpretations of ground terms over $\asig$.
%!
\begin{definition}\label{def:core}
  The \emph{core of a $\asig$-algebra $\aalg$} 
  is the $\asig$-algebra $\reachable{\aalg} = \pair{\aset_c}{\sinterpret_c}$
  where $\reachable{\aset}$ is the least set such that
  $\funap{\interpret{\tf}}{a_1,\ldots,a_n} \in \reachable{\aset}$ 
  whenever $\tf \in \asig_n$ and $a_1,\ldots,a_n \in \reachable{\aset}$,
  and where $\reachable{\sinterpret{}}$ is the restriction of $\sinterpret{}$ to $\reachable{\aset}$.
  We say that \emph{$\aalg$ is core} whenever $\aalg = \reachable{\aalg}$.
\end{definition}
By construction of the core of a $\asig$-algebra we then obtain:
\begin{lemma}\label{lem:core}
  For every element $a \in \reachable{\aalg}$ of the core of a $\asig$-algebra $\aalg$
  there exists a ground term $t \in \term{\asig}{\setemp}$ with $\interpret{t} = a$.
  \qed
\end{lemma}

\begin{definition}[\cite{zant:95}]\label{def:labeling}
  %Let $\asig$ be a signature. 
  A \emph{semantic labeling $\pair{\aalg}{\asemlab}$} 
  consists of a $\asig$-algebra $\aalg$\,
  and a family $\asemlab = \afam{\slabelf{\tf}}{\tf \in \asig}$ of 
  labeling functions 
  $\slabelf{\tf} \funin \aalg^{\arity{\tf}} \to \labels{\tf}$ 
  where $\labels{\tf}$ is a finite and non-empty set of labels for each symbol $\tf \in \asig$. 
  %
  %\noindent
  For a term $t \in \term{\asig}{\avars}$ and an assignment $\alpha \funin \avars \to \aalg$, 
  we define $\dolabel{t}{\alpha}$, 
  \emph{the labeling of\, $t$ with respect to $\alpha$},
  inductively as follows:
  \begin{align*}
     \dolabel{x}{\alpha} &= x 
     \\
     \dolabel{\f{t_1,\ldots,t_n}}{\alpha} 
     &= 
     \funap{\svoodoolabel{\tf}{\alab}}{
       \dolabel{t_1}{\alpha},\ldots,\dolabel{t_n}{\alpha}
     }
     %\punc.
  \end{align*}
  where $x\in\avars$, $\fun{f}\in\asig_n$, $t_1,\ldots,t_n\in\ter{\asig}{\avars}$ 
  and $\alab = \labelf{\tf}{\interpreta{t_1}{\alpha},\ldots,\interpreta{t_n}{\alpha}}$. 
  
  Let $\atrs$ be a TRS over $\asig$.
  The \emph{semantic labeling of $\atrs$} 
  is the TRS $\labeltrs{\atrs}$ over the labeled signature 
  $\labelsig{\asig} = \{\svoodoolabel{\tf}{\alab} \where \tf \in \asig,\; \alab \in \labels{\tf}\}$,
  defined by:
  \[
    \labeltrs{\atrs} 
    = 
    \{ 
      \dolabel{\ell}{\alpha} \to \dolabel{r}{\alpha} 
      \where 
      {\ell\to r}\in R \,,\; \alpha \funin {\vars{\ell} \to \aalg} 
    \}
    %\punc.
  \]
  %
  %\noindent
  For a substitution $\asubst\funin{\avars\to\term{\asig}{\avars}}$, 
  and an assignment $\alpha \funin \avars \to \aalg$,
  we write $\dolabel{\asubst}{\alpha}$ for the function 
  $\mylam{x}\dolabel{\funap{\asubst}{x}}{\alpha}$.
  For ground terms $t \in \term{\asig}{\setemp}$ 
  and ground substitutions $\asubst\funin\avars\to\term{\asig}{\setemp}$
  we write $\dolabelg{t}$ and $\dolabelg{\asubst}$ for short.
\end{definition}

Term labeling satisfies the following useful property: 
\begin{lemma}[\cite{zant:95}]\label{lem:zzz}
  Let $\aalg$ be a $\asig$-algebra, 
  let $\alpha\funin{\avars\to\aalg}$ be an assignment, 
  and let $\asubst\funin{\avars\to\term{\asig}{\avars}}$ be a substitution. 
  Then, for all terms $t\in\term{\asig}{\avars}$ it holds that:
  \begin{align*}
    \dolabel{\subst{\asubst}{t}}{\alpha} 
    & 
    = \subst{\,\dolabel{\asubst}{\alpha}}{\dolabel{t}{\interpreta{\asubst}{\alpha}}}
  \end{align*}
  %\qed
\end{lemma}
\begin{proof}
  Direct by induction on the term structure, and Lemma~\ref{lem:xyz}.
  %Case $t=x$ is trivial.
  %If $t=\funap{\fun{f}}{t_1,\ldots,t_n}$ for some $\fun{f}\in\asig_n$ and $t_1,\ldots,t_n\in\ter{\asig}{\avars}$,
  %then 
  %\begin{align*}
  %  \dolabel{\subst{\asubst}{t}}{\alpha}
  %  & = \funap{\svoodoolabel{\tf}{\alab}}{\dolabel{t_1}{\alpha},\ldots,\dolabel{t_n}{\alpha}}
  %  \quad
  %  \text{where $\alab = \labelf{\tf}{\interpreta{t_1}{\alpha},\ldots,\interpreta{t_n}{\alpha}}$, and}
  %  \\
  %  \subst{\,\dolabel{\asubst}{\alpha}}{\dolabel{t}{\interpreta{\asubst}{\alpha}}}
  %  & = \funap{\svoodoolabel{\tf}{\alab'}}{\dolabel{t_1}{\alpha},\ldots,\dolabel{t_n}{\alpha}}
  %  \quad 
  %  \text{where $\alab' = \labelf{\tf}{\interpreta{t_1}{\alpha},\ldots,\interpreta{t_n}{\alpha}}$, and}
  %\end{align*}
\end{proof}

The $\asig$-algebra of a semantic labeling has to satisfy certain constraints 
in order to obtain that a TRS is terminating if and only if its labeled version is.
In~\cite{zant:95} the algebra has to be a `model': % (or `quasi model'):
\begin{definition}\label{def:model}
  A $\asig$-algebra $\aalg$ 
  is called a \emph{model} for a TRS $\atrs$\,
  if for all rules $\ell \to r \in \atrs$ and assignments of variables in the left-hand side 
  $\alpha \funin \vars{\ell} \to \aalg$ 
  we have that $\interpreta{\ell}{\alpha} = \interpreta{r}{\alpha}$.
\end{definition}
In the introduction we argued why this notion of model is too restrictive for the purpose at hand.
In order to be able to distinguish between redex and non-redex positions
we introduce \cmodel{s}, a generalization of models.
%
%\begin{definition} %% moved to preliminaries
%  The \emph{depth} of a context $C$ is defined as the length $\lstlength{p}$ of the position $p$ 
%  at which $\contexthole$ resides, that is, the position $p$ such that 
%  $\symbat{\acontext}{\apos} = \contexthole$.
%\end{definition}
%
\begin{definition}\label{def:cmodel}
  A \emph{\cmodel} for a TRS $\atrs$ over $\asig$ 
  is a $\asig$-algebra $\aalg$ 
  where for each rule $\ell \to r \in \atrs$ there exists an $n \in \nat$
  such that for each context $\acontext$ of depth~$n$
  and assignment $\alpha \funin \avars \to \aalg$ 
  we have 
  $\interpreta{\contextfill{\acontext}{\ell}}{\alpha} = \interpreta{\contextfill{\acontext}{r}}{\alpha}$.
  When $n\in\nat$ is minimal for a rule $\ell\to r$ with respect to this property, 
  we call $n$ the \emph{\cdepth{} for $\ell\to r$},
  and denote it by $\cdpth{\aalg}{\ell\to r}$.
  The \emph{\cdepth\ for $\atrs$ with respect to $\aalg$}, denoted by $\cdpth{\aalg}{\atrs}$,
  is defined as the maximal \cdepth{} of the rules of $\atrs$:
  $\cdpth{\aalg}{\atrs} = \max{\{\cdpth{\aalg}{\ell\to r} \where {\ell\to r}\in\atrs\}}$.
\end{definition}

\newcommand{\fffxnr}{1}
\newcommand{\trsfffx}{\atrs_\fffxnr}
\newcommand{\algfffx}{\aalg_\fffxnr}

\begin{example}\label{ex:fffx:cmodel}
  Let $\trsfffx$ be the following TRS over 
  $\asig = \{\fun{c},\tf,\tg\}$ (where $\fun{c}$ is a constant):
  \begin{align}
    \f{\g{x}} &\to \f{\f{\g{x}}} 
    & \f{\f{\f{x}}} &\to x
    \tag{$\trsfffx$}
    %\label{ex:fffx}
  \end{align}
  The algebra $\algfffx = \{\bot,f,\mit{ff},g\}$ 
  with the interpretation function defined, for all $x\in\algfffx$, by:
  \begin{align}
    \interpret{\fun{c}} = \bot 
    && 
    \funap{\interpret{\tf}}{\bot} = \funap{\interpret{\tf}}{g} = f
    && 
    \funap{\interpret{\tf}}{f} = \funap{\interpret{\tf}}{\mit{ff}} = \mit{ff}
    &&
    \funap{\interpret{\tg}}{x} = g
    \tag{$\algfffx$}
  \end{align}
  forms a \cmodel{} for $\trsfffx$.
  The \cdepth{} of the rule $\f{\g{x}} \to \f{\f{\g{x}}}$ is $1$; 
  both contexts $\funap{\fun{f}}{\cxthole}$ 
  and $\funap{\fun{g}}{\cxthole}$ make that left and right-hand side of the rule have equal interpretations,
  respectively $\mit{ff}$ and $g$, 
  regardless of the value we assign to the variable $x$.
  For the other rule the \cdepth{} is~$2$:
  If the interpretation of $x$ is $\bot$ or $g$ then $\f{\cxthole}$ 
  is not yet enough to interpret both sides by the same element of the algebra;
  an additional context $\f{\cxthole}$ or $\g{\cxthole}$ has to be wrapped around.
  Thus the \cdepth{} of the TRS is $\cdpth{\algfffx}{\trsfffx} = 2$.
\end{example}

We now define semantic labelings based on \cmodel{s}, to which we refer as `\clabeling{s}'.
We define a \clabeling{} over an extended signature $\asig_{\topsymb} = \asig \cup \{\topsymb\}$.
The symbol~$\topsymb$ represents the top of a term, 
and we assume $\topsymb$ to be fresh for $\asig$, i.e., $\topsymb \not\in \asig$.
Moreover, a \clabeling{} includes a set $\asigred\subseteq\labelsig{\asig}$ of `redex symbols',
the set of symbols below which rewriting should be forbidden.
For example, for a sound transformation from outermost to context-sensitive rewriting 
it has to be guaranteed that in a well-labeled term the symbols from $\asigred$ occur at redex positions only.

\begin{definition}\label{def:clabeling}
  Let $\atrs$ be a TRS over $\asig$.
  A \emph{\clabeling\ for $\atrs$} is a tuple 
  $\triple{\aalg}{\asemlab}{\asigred}$
  where $\aalg$ is a \cmodel\ for $\atrs$,
  $\pair{\aalg}{\asemlab}$ is a semantic labeling 
  over the signature $\asig_{\topsymb} = \asig \cup \{\topsymb\}$,
  and $\asigred \subseteq \labelsig{\asig}$ is a subset of the labeled signature.  
  We fix the interpretation $\interpret{\topsymb}$ of $\topsymb$ to be an
  arbitrary constant function $\mylam{x}{a}$ for some $a \in \aalg$. 
  
  A \clabeling{} $\triple{\aalg}{\asemlab}{\asigred}$ for $\atrs$ is called:
  \begin{enumerate}[(i)]
  \item 
    \emph{sound} if $\rootsymb{\dolabelg{t}} \in \asigred$ 
    implies that $t$ is a redex with respect to~$\atrs$,
    for all ground terms $t \in \term{\asig}{\setemp}$;
  
  \item
    \emph{complete} if $\rootsymb{\dolabelg{t}} \in \asigred$
    whenever $t$ is a redex with respect to~$\atrs$,
    for all ground terms $t \in \term{\asig}{\setemp}$;
  
  \item
    \emph{maximal} if 
    $\labelf{\tf}{a_1,\ldots,a_n} = \triple{a_1}{\ldots}{a_n}$,
    for all symbols $\tf \in \asig_n$ and all values $a_1,\ldots,a_n \in \aalg$\,;
  
  \item
    \emph{core} if the $\asig$-algebra $\aalg$ is core.

  \end{enumerate}
\end{definition}

\begin{remark}
  From Definition~\ref {def:clabeling} it follows that a \clabeling{} $\triple{\aalg}{\asemlab}{\asigred}$ for $\atrs$ is:
  \begin{enumerate}[$-$]
    \item
    sound if and only if $\symbat{\dolabelg{t}}{\apos} \in \asigred$ 
    implies that $\subtrmat{t}{\apos}$ is a redex with respect to~$\atrs$,
    for all ground terms $t \in \term{\asig}{\setemp}$ and all
    positions $\apos \in \pos{t}$, and
    \item
    complete if and only if $\symbat{\dolabelg{t}}{\apos} \in \asigred$
    whenever $\subtrmat{t}{\apos}$ is a redex with respect to~$\atrs$,
    for all ground terms $t \in \term{\asig}{\setemp}$ 
    and all positions $\apos \in \pos{t}$.
  \end{enumerate}
\end{remark}

\begin{example}\label{ex:fffx:clabeling}
  We continue with Example~\ref{ex:fffx:cmodel} where we defined a \cmodel{}
  $\algfffx$ for the TRS~$\trsfffx$.
  We let $\pair{\algfffx}{\semlab}$ be the semantic labeling 
  where $\semlab$ labels each symbol with the interpretation of its arguments.
  The set of redex symbols is defined by 
  $\sigredex{\asig} = \{ \svoodoolabel{\tf}{g} , \svoodoolabel{\tf}{\mit{ff}} \}$.
  These symbols correspond to redex positions
  with respect to the first and the second rule of~$\trsfffx$.
  Then $\triple{\algfffx}{\asemlab}{\sigredex{\asig}}$
  forms a sound, complete, maximal, core \clabeling{} for  $\trsfffx$.
  %(Definition~\ref{def:clabeling}).
\end{example}

We explain why we %allow an arbitrary labeling function $\slabelf{\topsymb}$ for $\topsymb$, and
fix the interpretation $\interpret{\topsymb}$ to be a constant function.
First note that if $\interpret{\topsymb}$ is a constant function, 
then the extension of the signature with $\topsymb$ does not interfere with the property of 
$\aalg$ being a \cmodel\ for $\atrs$.
Second, the transformation given in Section~\ref{sec:cxtext} extends the rules
with contexts until the interpretations of the left and right-hand side are equal.
If $\interpret{\topsymb}$ is constant, then the extension halts at the symbol $\topsymb$,
that is, no further contexts are prefixed to $\topsymb$.
This corresponds to the intuition of $\topsymb$ representing the top of the term.
Moreover, it is not important which constant function $\mylam{x}{a}$ we choose for $\interpret{\topsymb}$:
as no symbols will be prefixed to $\topsymb$, no symbol will be labeled with its value.

\begin{remark}
  The transformations defined in Sections~\ref{sec:cxtext} and~\ref{sec:dynlab} are sound
  whenever we use a sound \clabeling.
  This means that termination of the target system  
  implies outermost termination of the original system.
  On the other hand, using a complete \clabeling{} does not guarantee completeness of either transformations. 
  More precisely, using a complete \clabeling{} does not imply
  that the transformed system is terminating
  whenever the original system is outermost terminating.
  A complete \clabeling{} guarantees that in a correctly labeled term all redex positions are marked,
  so that only outermost steps are possible.
  But, for the transformation to be complete we need two more properties.
  First, rewriting needs to preserve correct labeling of terms.
  Secondly, for the labeled system, global termination of all terms (including the not correctly labeled ones)
  should be equivalent with local termination~\cite{endr:vrij:wald:09} of the well-labeled terms.
  This point is of practical importance because 
  the state of the art of automated analysis for global termination 
  is far more advanced than for local termination.
  Both properties do in general not hold for complete \clabeling{s}.
\end{remark}

\begin{remark}
  Maximal \clabeling{s} can be defined in a more general fashion by requiring
  $\labelf{\tf}{\vec{a}} \neq \labelf{\tf}{\vec{b}}$
  for all $\tf \in \asig_n$ and all $\vec{a}, \vec{b} \in \aalg^n$ with $\vec{a} \neq \vec{b}$.
  The important point is that the value of all arguments can be inferred from the label.
  For the sake of a simple presentation we stick to the definition where
  labels are tuples of argument values.
\end{remark}

\section{Static Context Extension}\label{sec:static}

In this section we describe a naive approach for semantic labeling with \cmodel{s}.
This serves both as an introduction and as a motivation for the transformations 
that we present in Sections~\ref{sec:cxtext} and~\ref{sec:dynlab}.
Input for these transformations is a TRS together with a \cmodel{} for this TRS.
In Sections~\ref{sec:constructing}--\ref{sec:minmax} we explain how
\cmodel{s} are constructed.

As can be inferred from Definition~\ref{def:cmodel}, 
it is possible to transform a TRS~$\atrs$ 
by prepending contexts to its rules in such a way that 
its \cmodel~$\aalg$\, becomes a model for the transformed system $\tilde{\atrs}$,
and then apply the usual semantic labeling to $\tilde{\atrs}$.
We call this transformation `static context extension',
as opposed to the transformation of `dynamic context extension' presented in the next section.
In the dynamic version, contexts are prepended only when needed 
and dependent on the values assigned to the variables in the rules.

Since every context %of a certain depth 
is an instance of a thin context (Definition~\ref{def:thin:contexts}) with the same depth, 
rules are prefixed by thin contexts, in both versions of context extension.

%%We first give a definition of the transformation based on static context extension,
%%and then illustrate the transformation by means of an example.

\begin{definition}\label{def:cxtext:static}
  Let $\atrs$ be a TRS and $\triple{\aalg}{\asemlab}{\asigred}$ a \clabeling{} for $\atrs$.
  The \emph{static context extension of $\atrs$ with respect to $\triple{\aalg}{\asemlab}{\asigred}$}, 
  is the \csTRS{} $\pair{\staticcxtext{\aalg}{\semlab}{\atrs}}{\samumap}$ where
  $\staticcxtext{\aalg}{\semlab}{\atrs}$ is the TRS resulting from the steps listed below
  and where
  the replacement map $\samumap$
  is defined by $\amumap{\tf} = \setemp$ if $\tf \in \asigred$,
  and $\amumap{\tf} = \{1,\ldots,\arity{\tf}\}$ otherwise,
  for all $\tf\in\labelsig{\asig_{\topsymb}}$.
  %
  %$\staticcxtext{\aalg}{\semlab}{\atrs}$,
  %is the result of three steps:
  %prepending contexts, semantic labeling, and removal of thus created rules 
  %which have a redex symbol in the prepended context: 
  %
  \begin{enumerate}[(i)]

  \item \label{item:static:cxtext:prepend}
  Replace each rule $\arule$ with \cdepth~$n$
  by the set of rules:
  \begin{enumerate}
    \item\label{static:i} 
    $\cxtap{\acxt}{\ell} \to \cxtap{\acxt}{r}$ for each thin context~$\acxt$ of depth~$n$, 
    and
    \item\label{static:ii} 
    $\funap{\topsymb}{\cxtap{\acxt}{\ell}} \to \funap{\topsymb}{\cxtap{\acxt}{r}}$ 
    for each thin context~$\acxt$ of depth~$<n$.
  \end{enumerate}
  We let $\tilde{\atrs}$ denote the union of these sets.  

  \item \label{item:static:cxtext:label}
    Apply semantic labeling to %each of the rules from 
    $\tilde{\atrs}$ %step~\ref{item:static:cxtext:prepend} 
    using the \clabeling~$\pair{\aalg}{\asemlab}$. 
    We obtain $\labeltrs{\tilde{\atrs}}$. 
  
  \item \label{item:static:cxtext:remove}
  Remove from $\labeltrs{\tilde{\atrs}}$ %
  %the TRS resulting from step~\ref{item:static:cxtext:label} 
  all labeled rules that contain a redex symbol from $\asigred$
  in the prepended context.
  The TRS thus obtained is denoted by 
  $\staticcxtext{\aalg}{\semlab}{\atrs}$.

  \end{enumerate}  
  %
  %The result of these three steps is called 
  %the \emph{static context extension of~$\atrs$}.
\end{definition}

Note that the rules from item~\ref{static:ii} of Definition~\ref{def:cxtext:static}
model the application of an original rule at a depth 
smaller than its \cdepth, that is, `near' the top of the term.
The created rules are hence wrapped into $\funap{\topsymb}{\cxthole}$.

\begin{example}\label{ex:fffx:static:cxtext}
  We illustrate the transformation by static context extension
  on the TRS~$\trsfffx$:
  \begin{align}
    \f{\g{x}} &\to \f{\f{\g{x}}} 
    & \f{\f{\f{x}}} &\to x
    \tag{$\trsfffx$}
    %\label{ex:fffx}
  \end{align}
  together with the \cmodel{} $\algfffx = \{\bot,f,\mit{ff},g\}$
  introduced in Example~\ref{ex:fffx:cmodel}
  and the \clabeling{} 
  $\triple{\algfffx}{\asemlab}{\sigredex{\asig}}$
  defined in Example~\ref{ex:fffx:clabeling}.

  The first step of the transformation yields $\tilde{\atrs}_{\fffxnr}$ 
  which consists of the following $10$ rules:
  \begin{align*}
    \cxtfill{\acxt}{\f{\g{x}}} &\to \contextfill{\acontext}{\f{\f{\g{x}}}}
    &\cxtfill{\bcxt}{\f{\f{\f{x}}}} &\to \cxtfill{\bcxt}{x}
    \tag{$\tilde{\atrs}_{\fffxnr}$}
  \end{align*}
  where 
  $\acontext \in \{\funap{\topsymb}{\cxthole},\f{\cxthole},\g{\cxthole}\}$,
  and 
  $ \bcxt \in 
  \{ \funap{\topsymb}{\cxthole},
     \funap{\topsymb}{\f{\cxthole}},
     \funap{\topsymb}{\g{\cxthole}},
     \f{\f{\cxthole}},\f{\g{\cxthole}},\linebreak
     \g{\f{\cxthole}},\g{\g{\cxthole}}
  \}
  $.
  Note that the algebra~$\algfffx$ extended with the symbol $\topsymb$ and 
  interpretation $\funap{\interpret{\topsymb}}{x} = \bot$ for all $x\in\algfffx$, 
  is a model for the TRS $\tilde{\atrs}_{\fffxnr}$
  (in fact, any value for the interpretation of $\topsymb$ will do).
  
  The second step is to label the TRS $\tilde{\atrs}_{\fffxnr}$.
  This yields $\labeltrs{\tilde{\atrs}_{\fffxnr}}$ consisting of $4 \cdot (3+7) = 40$ rules, 
  four instances for each of the ten rules, % of the static context extension,
  one for each value that can be assigned to $x$.
  
  The final step is to remove from $\labeltrs{\tilde{\atrs}_{\fffxnr}}$
  each rule which contains a redex symbol within 
  the context that was prepended in the construction of $\tilde{\atrs}_{\fffxnr}$.
  Such a rule would enable a rewrite step which is not outermost.
  Of the 40 rules of above, 12 have to be thrown out. 
  %% since they contain a redex symbol in the prepended context.
  This concerns the labelings of the second rule of~$\trsfffx$
  where $\f{\f{\cxthole}}$, $\funap{\topsymb}{\f{\cxthole}}$ 
  or $\g{\f{\cxthole}}$ has been prepended.
  They contain the redex symbol $\svoodoolabel{\tf}{\mit{ff}}$ in the prepended context.
  For instance, if we prepend the thin context $\f{\f{\cxthole}}$ 
  to the rewrite rule $\f{\g{x}} \to \f{\f{\g{x}}}$ 
  and then perform semantic labeling, we obtain the following rule in $\labeltrs{\tilde{\atrs}_{\fffxnr}}$:
  \begin{align}
    \funap{\tf^{\mit{ff}}}{\funap{\tf^{\mit{f}}}{\funap{\tf^{\mit{g}}}{\funap{\tg^{\bot}}{x}}}} 
    &\to 
    \funap{\tf^{\mit{ff}}}{\funap{\tf^{\mit{ff}}}{\funap{\tf^{\mit{f}}}{\funap{\tf^{\mit{g}}}{\funap{\tg^{\bot}}{x}}}}}
    \label{illegalrule}
  \end{align}
  This rule (with the redex symbol $\tf^\mit{ff}$ in the prepended context) has to be discarded 
  as it would admit the following infinite rewrite sequence, altough $\trsfffx$ is outermost terminating:
  \begin{align*}
    \funap{\topsymb^{\mit{f}}}{\funap{\tf^{\mit{g}}}{\funap{\tg^{\bot}}{\fun{c}}}}
    &\mured
    \funap{\topsymb^{\mit{ff}}}{\funap{\tf^{\mit{f}}}{\funap{\tf^{\mit{g}}}{\funap{\tg^{\bot}}{\fun{c}}}}}
    \mured \funap{\topsymb^{\mit{ff}}}{\funap{\tf^{\mit{ff}}}{\funap{\tf^{\mit{f}}}{\funap{\tf^{\mit{g}}}{\funap{\tg^{\bot}}{\fun{c}}}}}}
    \\
    &\stackrel{\eqref{illegalrule}}{\mured}
    \funap{\topsymb^{\mit{ff}}}{\funap{\tf^{\mit{ff}}}{\funap{\tf^{\mit{ff}}}{\funap{\tf^{\mit{f}}}{\funap{\tf^{\mit{g}}}{\funap{\tg^{\bot}}{\fun{c}}}}}}}
    \mured
    \funap{\topsymb^{\mit{f}}}{\funap{\tf^{\mit{g}}}{\funap{\tg^{\bot}}{\fun{c}}}}
    \mured \ldots
  \end{align*}
  
  We note in advance that the \emph{dynamic} context extension $\cxtext{\algfffx}{\asemlab}{\trsfffx}$
  %, described in the next section, 
  of $\trsfffx$, worked out in Example~\ref {ex:fffx:dyn:cxtext}, consists of $19$ rewrite rules, 
  whereas the static version $\staticcxtext{\algfffx}{\asemlab}{\trsfffx}$ above consists of $28$ rules.
\end{example}

In general, the presence of a redex symbol may depend on the interpretation of the variables.
This is better demonstrated by the following example.
\begin{example}
  Consider the TRS over the signature $\{\fun{c},\tf,\tg\}$ (with $\fun{c}$ a constant):
  \begin{align}
    \g{\f{\g{x}}} &\to \f{\g{\g{\f{x}}}}
    & \f{x} &\to x
    \tag{$\atrs_2$}
    \label{ex:fx}
  \end{align}  
  We use the \cmodel{} $\iaalg{2} = \{\bot,g,\mit{fg}\}$
  with the interpretation of the symbols defined by:
  \begin{align}
    \interpret{\fun{c}} = \bot
    &&
    \funap{\interpret{\tg}}{x} = g
    &&
    \funap{\interpret{\tf}}{g} = \mit{fg}
    &&
    \funap{\interpret{\tf}}{\bot} = \funap{\interpret{\tf}}{\mit{fg}} = \bot
    \tag{$\iaalg{2}$}
  \end{align}
  for all $x\in\iaalg{2}$.
  Again we use maximal labeling 
  so that the symbols $\svoodoolabel{\tg}{\mit{fg}}$, $\svoodoolabel{\tf}{\bot}$, 
  $\svoodoolabel{\tf}{g}$ and $\svoodoolabel{\tf}{\mit{fg}}$
  correspond to redex positions.
  The \cdepth{} of the rule $\f{x} \to x$ is $2$ and
  its static context extension contains the rule 
  $\g{\g{\f{x}}} \to \g{\g{x}}$.
  From this we obtain three labeled rules:
  \begin{align*}
    \funap{\svoodoolabel{\tg}{g}}{\funap{\svoodoolabel{\tg}{\bot}}{\funap{\svoodoolabel{\tf}{\bot}}{x}}} 
    &\to \funap{\svoodoolabel{\tg}{g}}{\funap{\svoodoolabel{\tg}{\bot}}{x}}
    &&\text{for $\funap{\alpha}{x} = \bot$}\\
    \funap{\svoodoolabel{\tg}{g}}{\funap{\svoodoolabel{\tg}{\mit{fg}}}{\funap{\svoodoolabel{\tf}{g}}{x}}} 
    &\to \funap{\svoodoolabel{\tg}{g}}{\funap{\svoodoolabel{\tg}{g}}{x}}
    &&\text{for $\funap{\alpha}{x} = g$}\\
    \funap{\svoodoolabel{\tg}{g}}{\funap{\svoodoolabel{\tg}{\bot}}{\funap{\svoodoolabel{\tf}{\mit{fg}}}{x}}} 
    &\to \funap{\svoodoolabel{\tg}{g}}{\funap{\svoodoolabel{\tg}{\mit{fg}}}{x}}
    &&\text{for $\funap{\alpha}{x} = \mit{fg}$}
  \end{align*}
  The second rule should not be allowed, as it would enable a rewrite step that is not outermost.
  This is witnessed by the symbol $\svoodoolabel{\tg}{\mit{fg}}$ in the prepended context.
\end{example}

\section{Dynamic Context Extension}\label{sec:cxtext}

We present an approach for semantic labeling with \cmodel{s}, 
called `dynamic context extension',%
  \footnote{%
    In~\cite{endr:hend:09} we used the term `dynamic labeling' for what we here call `dynamic context extension'.
    The term `dynamic labeling' is now reserved for the transformation that we define in Section~\ref{sec:dynlab}.
  }
where we stepwise extend rules by contexts, only when needed and 
dependent on the variable interpretation used for the semantic labeling.
For different interpretations of the variables usually different
context depths are necessary for achieving equal interpretations of left and right-hand side.
In each extension step we check whether a candidate symbol is a redex symbol,
and, if it is, this symbol is excluded from prepending.
Here, by a redex symbol we mean a labeled symbol which
indicates the presence of a redex in the original system.
%Moreover we can check in each extension step
%whether the prepended symbol is a redex symbol and then immediately exclude it from prepending.
Dynamic context extension is more efficient in the sense that 
both the number and the size of the rules of the resulting \csTRS{} 
are smaller than in the static version defined in the previous version.

The transformation starts with constructing pairs $\pair{\arule}{\alpha}$ of rules and variable assignments.
Then these rules are extended with flat contexts until the interpretations of left and right-hand side are equal.
Finally, each obtained rule is labeled using the corresponding interpretation.
More precisely, we implement this process as follows.

We iteratively construct sets $\icxtextpairs{0}, \icxtextpairs{1}, \ldots$,
until $\icxtextpairs{i+1} = \icxtextpairs{i}$ for some $i$. % \in \nat$.
The initial set $\icxtextpairs{0}$ consists of pairs 
$\pair{\ell\to r}{\alpha}$ for each rule $\ell\to r$, 
and each interpretation $\alpha\funin\vars{\ell}\to\aalg$ of the variables.
Then, in each step, $\icxtextpairs{i+1}$ is obtained from $\icxtextpairs{i}$ by replacing every pair 
$\pair{\ell\to r}{\alpha}$ of $\icxtextpairs{i}$
for which the interpretation of the left-hand side differs from the right-hand side 
($\interpreta{\ell}{\alpha} \neq \interpreta{r}{\alpha}$),
by the pairs $\pair{\contextfill{\acontext}{\ell}\to \contextfill{\acontext}{r}}{\alpha'}$
for every flat context $\acontext$ (Definition~\ref{def:flat:contexts}) 
and every extension $\alpha'\funin\vars{\contextfill{\acontext}{\ell}}\to\aalg$ of $\alpha$,
such that the root of the labeled, extended left-hand side 
$\dolabel{\contextfill{\acontext}{\ell}}{\alpha'}$ is not a redex symbol.
Among the flat contexts to be prepended we include $\funap{\topsymb}{\contexthole}$
to cater for the case that the rule is applied at the top of the term.

\begin{example}\label{ex:fffx:dyn:cxtext}
  We reconsider from Examples~\ref{ex:fffx:cmodel} and~\ref{ex:fffx:clabeling} the term rewriting system~$\trsfffx$:
    \begin{align}
    \f{\g{x}} &\to \f{\f{\g{x}}} 
    & \f{\f{\f{x}}} &\to x
    \tag{$\trsfffx$}
    %\label{ex:fffx}
  \end{align}
  together with the \clabeling{} $\triple{\algfffx}{\asemlab}{\sigredex{\asig}}$,
  where $\algfffx = \{\bot,f,\mit{ff},g\}$, $\asemlab$ labels symbols with their arguments 
  and $\sigredex{\asig} = \{ \svoodoolabel{\tf}{g} , \svoodoolabel{\tf}{\mit{ff}} \}$.
  The initial set $\icxtextpairs{0}$ of pairs $\pair{\text{rule}}{\text{assignment}}$ is:
  \begin{align*}
  \icxtextpairs{0}
  =
  \big\{
    \pair{\f{\g{x}} \to \f{\f{\g{x}}}}{\mylam{x}{a}}\,,
    \pair{\f{\f{\f{x}}} \to x}{\mylam{x}{a}}
    \where
    a\in\algfffx
  \big\}
  \end{align*}
  The only element $\pair{\ell\to r}{\alpha}$ of $\icxtextpairs{0}$ such that 
  $\interpreta{\ell}{\alpha} = \interpreta{r}{\alpha}$
  is $\pair{\f{\f{\f{x}}} \to x}{\mylam{x}{\mit{ff}}}$.
  For this pair no context needs to be prepended. 
  The other pairs have to be replaced by their context extensions
  and thus $\icxtextpairs{1}$ consists of the following ($4\cdot3+1+3\cdot2 = 19$) pairs:
  \begin{align*}
    \pair{\contextfill{\acontext}{\f{\g{x}}} 
    \to \contextfill{\acontext}{\f{\f{\g{x}}}}}{\mylam{x}{a}}
    &&& \text{for all $a\in\algfffx$, %and 
    %$\acontext\in\{\topsymb,f,g\}$}
    $\acontext\in\{\funap{\topsymb}{\cxthole},\f{\cxthole},\g{\cxthole}\}$}
    \\
    \pair{\f{\f{\f{x}}} \to x}{\mylam{x}{\mit{ff}}}
    \\
    \pair{\contextfill{\acontext}{\f{\f{\f{x}}}} 
    \to \contextfill{\acontext}{x}}{\mylam{x}{a}}
    &&& \text{for all $a \neq \mit{ff}$, %$a\in\algfffx\setminus\{f_2\}$, % and
    $\acontext\in\{\funap{\topsymb}{\cxthole},\g{\cxthole}\}$}
  \end{align*}
  In the last line the context $\f{\contexthole}$ is excluded,
  because the labeled left-hand side of the rule would contain 
  the redex symbol $\svoodoolabel{\tf}{\mit{ff}}$ within the prepended context,
  and thus the step would not be outermost.
  Because of the outermost strategy, the original rule is only applicable in a context 
  $\cxtap{\acxt}{\g{\cxthole}}$ (where $\acxt$ does not contain any redexes)
  or at the top of a term.
  Now for all rules in $\icxtextpairs{1}$ 
  the left and right-hand side have equal interpretations,
  and hence the iterative construction is finished.

  Secondly, the obtained set~$P_1$ is labeled using the family~$\semlab$ of labeling functions.
  The desired context-sensitive TRS~$\cxtext{\algfffx}{\semlab}{\trsfffx}$
  then consists of the rules 
  $\dolabel{\ell}{\alpha} \to \dolabel{r}{\alpha}$ 
  for every $\pair{\ell \to r}{\alpha} \in P_1$, 
  with the replacement map $\samumap$
  defined by $\amumap{\sh} = \setemp$ if $\sh \in \{ \tf^g ,\tf^{\mit{ff}} \}$,
  and $\amumap{\sh} = \{1,\ldots,\arity{\sh}\}$ otherwise,
  for all $\sh\in\labelsig{\asig}$.
  Thus the dynamic context extension of $\trsfffx$ consists of $19$ rules.
  Recall from Example~\ref{ex:fffx:static:cxtext} 
  that the static context extension of $\trsfffx$ had $28$ rules.
\end{example}

We now formalize this transformation. %dynamic context extension.
\begin{definition}\label{def:cxtextpairs}
  Let $\atrs$ be a TRS over $\asig$,
  and $\triple{\aalg}{\semlab}{\asigred}$ a \clabeling\ for $\atrs$.
  We define $\cxtextpairs{\semlab}{\atrs}$ as the least fixed point 
  of the following construction of sets $\icxtextpairs{0},\icxtextpairs{1},\ldots$\,,
  that is, $\cxtextpairs{\semlab}{\atrs} = \icxtextpairs{i}$ 
  as soon as $\icxtextpairs{i+1} = \icxtextpairs{i}$ for some~$i$.
  The initial set $\icxtextpairs{0}$ is defined by:
  \begin{align*}
     \icxtextpairs{0} & = 
       \big\{\, \pair{\ell \to r}{\alpha} \where \ell \to r \in \atrs,\; \alpha \funin \vars{\ell} \to \aalg \,\big\}
  \end{align*}
  and for $i=0,1,\ldots$ the set $\icxtextpairs{i+1}$ is obtained from $\icxtextpairs{i}$ 
  by replacing every pair $\pair{\ell \to r}{\alpha}$ 
  such that $\interpreta{\ell}{\alpha} \ne \interpreta{r}{\alpha}$, 
  or $r \in \avars$,
  by all pairs in $\prepend{\ell \to r}{\alpha}$ where we define:
  \begin{align*}
    \prepend{\ell \to r}{\alpha} 
     = 
     \big\{\, 
       \pair{\contextfill{\acontext}{\ell} \to \contextfill{\acontext}{r}}{\extend{\alpha}{\beta}}
       \mathrel{\big|} \mbox{}
       &\acontext \in \bfunap{\contexts_\ell^\flat}{\asig_\topsymb}{\avars},\;
       \beta \funin \vars{\acontext} \to \aalg,\\
       & \rootsymb{\dolabel{\contextfill{\acontext}{\ell}}{\extend{\alpha}{\beta}}} \not\in\asigred
     \,\big\}
  \end{align*}
Here, for partial functions $f$ and $g$ with disjoint domains,
we write $\extend{f}{g}$ for the function defined by
$\funap{(\extend{f}{g})}{x} = \funap{f}{x}$ if $x\in\dom{f}$,
and $\funap{(\extend{f}{g})}{x} = \funap{g}{x}$ if $x\in\dom{g}$.
\end{definition}
The construction of $\cxtextpairs{\semlab}{\atrs}$ % $\cxtextset{\aalg}{\atrs}$ 
is guaranteed to terminate because of the assumption that 
$\aalg$ is a \cmodel{} for $\atrs$.

\begin{definition}[Dynamic context extension]\label{def:cxtext}
  Let $\atrs$ be a TRS over $\asig$,
  and $\triple{\aalg}{\semlab}{\asigred}$ a \clabeling\ for $\atrs$.
  The \emph{dynamic context extension of $\atrs$ with respect to $\triple{\aalg}{\semlab}{\asigred}$}
  is the \csTRS{} $\pair{\cxtext{\aalg}{\semlab}{\atrs}}{\samumap}$ consisting of:
  \begin{align*}
    \cxtext{\aalg}{\semlab}{\atrs} 
    &= 
    \big\{\,
      \dolabel{\ell}{\alpha} \to \dolabel{r}{\alpha} 
      \where \pair{\ell \to r}{\alpha} \in \cxtextpairs{\semlab}{\atrs} % \cxtextset{\aalg}{\atrs}
    \,\big\}
  \end{align*}
  and the replacement map $\samumap$, % \in \mumaps{\labelsig{\asig_{\topsymb}}}$ 
  defined by $\amumap{\tf} = \setemp$ if $\tf \in \asigred$,
  and $\amumap{\tf} = \{1,\ldots,\arity{\tf}\}$ otherwise,
  for all $\tf\in\labelsig{\asig_{\topsymb}}$.
  Whenever the set $\asigred$, which determines the replacement map,
  is clear from the context, we write $\cxtext{\aalg}{\semlab}{\atrs}$
  as a shorthand for $\pair{\cxtext{\aalg}{\semlab}{\atrs}}{\samumap}$.
\end{definition}

\begin{remark}\label{rem:cxtext:non-collapsing}
  %The condition $r \not\in \avars$ of Definition~\ref{def:cxtext}
  %eliminates collapsing rules.
  In the transformation given in Definition~\ref{def:cxtext}
  collapsing rules are always prepended by at least one flat context. 
  Consequently, $\pair{\cxtext{\aalg}{\semlab}{\atrs}}{\samumap}$
  contains no collapsing rules.
  This is used in the proof of Theorem~\ref{thm:abstract:maxcomplete}
  in order to apply Theorem~\ref{thm:adapt:ohlebusch}.
  %, a result of Ohlebusch~\cite{ohle:02}.
  Without this elimination of collapsing rules, 
  the transformation %() 
  is still sound (Theorem~\ref{thm:cxtext:sound}).
  Note that in the TRS~$\trsfffx$ worked out before, 
  we did not eliminate the collapsing rule.
\end{remark}

Let us work out another example.
\newcommand{\duplrhsnr}{3}%
\newcommand{\sigduplrhs}{\iasig{\duplrhsnr}}%
\newcommand{\trsduplrhs}{\iatrs{\duplrhsnr}}%
\newcommand{\algduplrhs}{\iaalg{\duplrhsnr}}%
\begin{example}\label{ex:dupl_rhs}
  We consider problem~\verb!zantema08/dupl_rhs.trs! 
  from the TPDB~\cite{term:comp:08}:
  %(RULES
  %  f(h(x),c) -> f(i(x),s(x))
  %  h(x) -> f(h(x),c)
  %  i(x) -> h(x)
  %  f(i(x),y) -> x
  %)
  \begin{gather}
  \begin{aligned}
    \fb{\h{x}}{\fun{c}}
    & \to \fb{\funap{\fun{i}}{x}}{\funap{\fun{s}}{x}}
    &&&&&
    \funap{\fun{i}}{x}
    & \to \h{x}
    \\
    \fb{\funap{\fun{i}}{x}}{y}
    & \to x
    &&&&&
    \h{x}
    & \to \fb{\h{x}}{\fun{c}}
  \end{aligned}
  \tag{$\trsduplrhs$}
  \end{gather}
  We denote this TRS by $\trsduplrhs$, and
  take the algebra $\algduplrhs = \pair{\{\bot,c,h,i\}}{\sinterpret}$ 
  with $\sinterpret$ defined by:
  \begin{align*}
    \interpret{\fun{c}}=c
  &&
  \funap{\interpret{\sh}}{x} = h
  &&
  \funap{\interpret{\fun{i}}}{x} = i
  &&
  \bfunap{\interpret{\tf}}{x}{y} = \funap{\interpret{\fun{s}}}{x} = \bot
  \end{align*}
  for all $x,y\in\algduplrhs$.
  Furthermore, we employ minimal labeling; % i.e., 
  only the function symbols that are at the root of a redex occurrence are marked.
  Thus the symbols $\sh$, $\fun{i}$ are always marked:
  $\labelf{\sh}{x} = \labelf{\fun{i}}{x} = \star$.
  We let $\labelf{\tf}{i,x} = \labelf{\tf}{h,c} = \star$,
  and leave $\tf$ unmarked otherwise. 
  Also, the symbols $\fun{s}$ and $\fun{c}$ are never marked.
  
  The dynamic context extension $\cxtext{\algduplrhs}{\semlab}{\trsduplrhs}$ 
  is then formed by the rules: 
  \begin{align*}
    \bfunap{\rmark{\tf}}{\funap{\rmark{\sh}}{x}}{\fun{c}}
    & \to \bfunap{\rmark{\tf}}{\funap{\rmark{\fun{i}}}{x}}{\funap{s}{x}}
    & 
    \funap{\fun{s}}{\funap{\rmark{\fun{i}}}{x}}
    & \to \funap{\fun{s}}{\funap{\rmark{\sh}}{x}}
    \\
    &&
    \fb{y}{\funap{\rmark{\fun{i}}}{x}}
    & \to \fb{y}{\funap{\rmark{\sh}}{x}}
    \\
    \funap{\fun{s}}{\bfunap{\rmark{\tf}}{\funap{\rmark{\fun{i}}}{x}}{y}}
    & \to \funap{\fun{s}}{x} 
    &
    \funap{\topsymb}{\funap{\rmark{\fun{i}}}{x}}
    & \to \funap{\topsymb}{\funap{\rmark{\sh}}{x}}
    \\
    \fb{z}{\bfunap{\rmark{\tf}}{\funap{\rmark{\fun{i}}}{x}}{y}}
    & \to \fb{z}{x}
    \tag{$\cxtext{\algduplrhs}{\semlab}{\trsduplrhs}$}
    \\
    \fb{z}{\bfunap{\rmark{\tf}}{\funap{\rmark{\fun{i}}}{x}}{y}}
    & \to \bfunap{\rmark{\tf}}{z}{x} 
    &
    \funap{\fun{s}}{\funap{\rmark{\sh}}{x}}
    & \to \funap{\fun{s}}{\bfunap{\rmark{\tf}}{\funap{\rmark{\sh}}{x}}{\fun{c}}}
    \\
    \fb{\bfunap{\rmark{\tf}}{\funap{\rmark{\fun{i}}}{x}}{y}}{z}
    & \to \fb{x}{z} 
    &
    \fb{y}{\funap{\rmark{\sh}}{x}}
    & \to \fb{y}{\bfunap{\rmark{\tf}}{\funap{\rmark{\sh}}{x}}{\fun{c}}}
    \\
    \fb{\bfunap{\rmark{\tf}}{\funap{\rmark{\fun{i}}}{x}}{y}}{z}
    & \to \bfunap{\rmark{\tf}}{x}{z}
    &
    \fb{\funap{\rmark{\sh}}{x}}{y}
    & \to \fb{\bfunap{\rmark{\tf}}{\funap{\rmark{\sh}}{x}}{\fun{c}}}{y}
    \\
    \funap{\topsymb}{\bfunap{\rmark{\tf}}{\funap{\rmark{\fun{i}}}{x}}{y}}
    & \to \funap{\topsymb}{x}
    &
    \funap{\topsymb}{\funap{\rmark{\sh}}{x}}
    & \to \funap{\topsymb}{\bfunap{\rmark{\tf}}{\funap{\rmark{\sh}}{x}}{\fun{c}}}
  \end{align*}
  %
  % The rule $\funap{i}{x} \to \h{x}$
  % is not prepended with the context $\acxt = \fb{\cxthole}{y}$,
  % because $\cxtfill{\acxt}{\funap{i}{x}}$ would constitute a redex.
  %
  % The rule $\bfunap{\rmark{\tf}}{\funap{\rmark{i}}{x}}{y} \to x$ 
  % is extended with contexts even for the case that there is 
  % no change of interpretation if the value assigned to $x$ is $\bot$.
  % The reason is that collapsing rules are excluded from
  % $\cxtext{\algduplrhs}{\semlab}{\trsduplrhs}$.
  %
\end{example}

We now work towards the first main theorem,
stating that outermost ground termination of~$\atrs$
is implied by termination of the transformed system~$\cxtext{\aalg}{\semlab}{\atrs}$.
\begin{lemma}\label{lem:yyy}
  Let $\atrs$ be a TRS over $\asig$, 
  and let $\triple{\aalg}{\semlab}{\asigred}$ be a sound \clabeling\ for $\atrs$.
  Moreover, let $s,t \in \term{\asig}{\setemp}$ be ground terms and $p\in\pos{s}$
  such that $s \outred_{\atrs,p} t$. 
  Then for all proper prefixes $q$ of $1p$ we have 
  $\symbat{\dolabelg{\funap{\topsymb}{s}}}{q}\not\in\asigred$.
\end{lemma}
\begin{proof}
  For $q=\posemp$ this follows from $\topsymb^\alab \not\in \asigred$ for any label $\alab$.
  Otherwise we have that $\symbat{\dolabelg{\funap{\topsymb}{s}}}{q} = \symbat{\dolabelg{s}}{q'}$ with
  $q'$ a proper prefix of $p$, and if $\symbat{\dolabelg{s}}{q'}\in\asigred$,
  then, by 
  %assumption~($\ddagger$) 
  definition of sound \clabeling,
  $s$ contains a redex at position $q'$, quod non.
  %\qed
\end{proof}
The following lemma states that any outermost ground rewrite step in $\atrs$
can be transformed into a rewrite step in $\cxtext{\aalg}{\semlab}{\atrs}$.
%
%For ground substitutions $\asubst\funin\avars\to\term{\asig}{\setemp}$
%define $\funap{\overline{\asubst}}{x} = \dolabelg{\subst{\asubst}{x}}$.
%
\begin{lemma}\label{lem:xxx}
  Let $\atrs$ be a TRS over $\asig$, 
  and let $\triple{\aalg}{\semlab}{\asigred}$ be a sound \clabeling\ for $\atrs$.
  Let $s,t \in \term{\asig}{\setemp}$ be ground terms
  such that $s \outred_{\atrs} t$. Then: 
  \[
    \dolabelg{\funap{\topsymb}{s}} 
    \to_{\cxtext{\aalg}{\semlab}{\atrs},\samumap} 
    \dolabelg{\funap{\topsymb}{t}}
  \]
\end{lemma}
\begin{proof}
  Assume $s \outred_{\atrs,p} t$ for some position $p \in \pos{s}$.
  Then there exists a rule $\ell \to r \in \atrs$, 
  a context $\acontext$ with $\symbat{\acontext}{p} = \contexthole$
  and a ground substitution $\asubst$ such that 
  $s = \contextfill{\acontext}{\subst{\asubst}{\ell}}$ and $t = \contextfill{\acontext}{\subst{\asubst}{r}}$.
  We consider the construction of the dynamic context extension from Definition~\ref{def:cxtext},
  and prove by induction that for all $i = 0,1,\ldots$ there exists 
  a context $\acontext_i$ which is a prefix of $\funap{\topsymb}{\acontext}$,
  a ground substitution $\asubst_i$, and terms $\ell_i$, $r_i$ such that
  $\funap{\topsymb}{s} = \contextfill{\acontext_i}{\subst{\asubst_i}{\ell_i}}$,
  $\funap{\topsymb}{t} = \contextfill{\acontext_i}{\subst{\asubst_i}{r_i}}$
  and $\pair{\ell_i \to r_i}{\interpret{\asubst_i}} \in P_i$.
  For the base case we have $\pair{\ell_0 \to r_0}{\interpret{\asubst_0}} \in P_0$ with
  $\ell_0 = \ell$, $r_0 = r$, $\asubst_0 = \asubst$, and $\acontext_0 = \funap{\topsymb}{\acontext}$.
  For the induction step we assume the existence of
  $\acontext_i$, $\asubst_i$, and $\pair{\ell_i \to r_i}{\interpret{\asubst_i}} \in P_i$ 
  with the above properties.
  If $\interpreta{\ell_i}{\interpret{\asubst_i}} = \interpreta{r_i}{\interpret{\asubst_i}}$
  and $r_i \not\in \avars$ then by definition $\pair{\ell_i \to r_i}{\interpret{\asubst_i}} \in P_{i+1}$, 
  and so we are done.
  For the remaining cases 
  $\interpreta{\ell_i}{\interpret{\asubst_i}} \ne \interpreta{r_i}{\interpret{\asubst_i}}$
  and $r_i\in\avars$, we first show that $\acontext_i \ne \contexthole$.
  If $\interpreta{\ell_i}{\interpret{\asubst_i}} \ne \interpreta{r_i}{\interpret{\asubst_i}}$
  and $\acontext_i = \contexthole$, 
  then $\subst{\asubst_i}{\ell_i} = \funap{\topsymb}{s}$ 
  and $\subst{\asubst_i}{r_i} = \funap{\topsymb}{t}$, 
  and hence $\rootsymb{\ell_i} = \rootsymb{r_i} = \topsymb$, contradicting 
  $\interpreta{\ell_i}{\interpret{\asubst_i}} \ne \interpreta{r_i}{\interpret{\asubst_i}}$
  (recall that the interpretation of $\topsymb$ is constant).
  Furthermore, we have $r_i\in\avars$ only if $i=0$, 
  and then $\acontext_i = \funap{\topsymb}{\acontext} \neq \contexthole$.
  Thus we have $\acontext_i = \contextfill{\bcontext}{\subst{\asubst'}{\bcontext'}}$
  for some context $\bcontext$, flat context $\bcontext' \in \contexts_{\ell_i}^\flat$
  and substitution~$\asubst'$.
  We choose $\acontext_{i+1} = \bcontext$,
  $\ell_{i+1} = \contextfill{\bcontext'}{\ell_i}$,
  $r_{i+1} = \contextfill{\bcontext'}{r_i}$, and
  $\asubst_{i+1} = \extend{\asubst_i}{\asubst'}$.
  It remains to be shown that $\pair{\ell_{i+1} \to r_{i+1}}{\interpret{\asubst_{i+1}}} \in P_{i+1}$.
  For this it suffices to prove that 
  $\rootsymb{\dolabel{\ell_{i+1}}{\interpret{\asubst_{i+1}}}} \not\in \asigred$.
  We have $\contextfill{\acontext_{i+1}}{\subst{\asubst_{i+1}}{\ell_{i+1}}} = \funap{\topsymb}{s}$.
  Let $q$ be the position such that $\symbat{\acontext_{i+1}}{q} = \contexthole$.
  Then, by Lemma~\ref{lem:zzz} 
  we obtain 
  \(
    \rootsymb{\dolabel{\ell_{i+1}}{\interpret{\asubst_{i+1}}}} 
    = \rootsymb{\dolabelg{\subst{\asubst_{i+1}}{\ell_{i+1}}}}
    = \symbat{\funap{\topsymb}{\dolabelg{s}}}{q}
  \).
  Note that $q$ is a proper prefix of $1p$. 
  By Lemma~\ref{lem:yyy} we have 
  $\symbat{\dolabelg{\funap{\topsymb}{s}}}{q}\not\in\asigred$.

  Let $i$ be such that $P_{i+1} = P_i$. 
  By the result above we have 
  $\pair{\ell_i \to r_i}{\interpret{\asubst_i}} \in P$  % $\in \cxtextset{\aalg}{\atrs}$,
  with $\interpret{\subst{\asubst_i}{\ell_i}} = \interpret{\subst{\asubst_i}{r_i}}$,
  and 
  $\dolabel{\ell_i}{\interpret{\asubst_i}} \to \dolabel{r_i}{\interpret{\asubst_i}} 
  \in \cxtext{\aalg}{\semlab}{\atrs}$ by definition.
  Let $\bsubst$ and $\csubst$ be defined by 
  $\funap{\bsubst}{\cxthole} = \subst{\asubst_i}{\ell_i}$, 
  $\funap{\csubst}{\cxthole} = \subst{\asubst_i}{r_i}$,
  and $\funap{\bsubst}{x} = \funap{\csubst}{x} = x$ for $x\in\avars$.
  Then we have that 
  $\dolabel{\acxt_i}{\interpret{\bsubst}} = \dolabel{\acxt_i}{\interpret{\csubst}}$
  since $\interpret{\bsubst} = \interpret{\csubst}$.
  Let $E = \dolabel{\acxt_i}{\interpret{\bsubst}}$.
  We get
  $
  \dolabelg{\funap{\topsymb}{s}}
  = \dolabelg{\cxtap{\acxt_i}{\subst{\asubst_i}{\ell_i}}}$
  $= \dolabelg{\subst{\tau}{\acxt_i}}
  = \subst{\dolabelg{\bsubst}}{E}
  = \cxtap{E}{\dolabelg{\subst{\asubst_i}{\ell_i}}}
  = \cxtap{E}{\subst{\dolabelg{\asubst_i}}{\dolabel{\ell_i}{\interpret{\asubst_i}}}}
  $ 
  and
  $
  \dolabelg{\funap{\topsymb}{t}}
  = \ldots
  = \cxtap{E}{\subst{\dolabelg{\asubst_i}}{\dolabel{r_i}{\interpret{\asubst_i}}}}
  $,
  by Lemma~\ref{lem:zzz}.
  By Lemma~\ref{lem:yyy} %it follows that 
  all symbols above position $1p$ in the term $\dolabelg{\funap{\topsymb}{s}}$ are not in $\asigred$
  and hence we have a $\samumap$-rewrite step:
  $\dolabelg{\funap{\topsymb}{s}}
   \to_{\cxtext{\aalg}{\semlab}{\atrs},\samumap} 
   \dolabelg{\funap{\topsymb}{t}}
  $.
  \end{proof}
\begin{theorem}\label{thm:cxtext:sound}
  Let $\atrs$ be a TRS over $\asig$, 
  and $\triple{\aalg}{\semlab}{\asigred}$ a sound \clabeling\ for $\atrs$.
  Then $\atrs$\ is outermost ground terminating if $\cxtext{\aalg}{\semlab}{\atrs}$ is terminating.
\end{theorem}
%
%\begin{proof}
\proof
  Assume that $\atrs$\, admits an infinite outermost rewrite sequence:
  \[
    t_1 \outred_{\atrs} t_2 \outred_{\atrs} t_3 \outred_{\atrs} \ldots
  \]
  Then from Lemma~\ref{lem:xxx} it follows that 
  $\cxtext{\aalg}{\semlab}{\atrs}$ admits an infinite rewrite sequence:
  $$
    \dolabelg{\funap{\topsymb}{t_1}} 
    \to_{\cxtext{\aalg}{\semlab}{\atrs},\samumap}
    \dolabelg{\funap{\topsymb}{t_2}} 
    \to_{\cxtext{\aalg}{\semlab}{\atrs},\samumap}
    \dolabelg{\funap{\topsymb}{t_3}} 
    \to_{\cxtext{\aalg}{\semlab}{\atrs},\samumap}
    \ldots\eqno{\qEd}
  $$
%\end{proof}
%

The following three examples illustrate why our method is sound, but not complete
when applied to non-left-linear (and non-quasi-left-linear) TRSs.
The first example can be handled by our approach employing the \clabeling{} constructed in Section~\ref{sec:minmax}.
The second example fails using the \clabeling{} from Section~\ref{sec:minmax},
but can successfully be treated using a manually constructed \clabeling{}.
For the third example, we show that there exists no \clabeling{}
that can be employed for proving outermost ground termination;
this example is out of reach for the approach proposed in this paper.
\begin{example}\label{ex:nonlin}
\newcommand{\trsnl}{\atrs_4}
\newcommand{\algnl}{\aalg_4}
  We consider the non-left-linear TRS $\trsnl$ with three rules:
  \begin{align*}
    \gb{x}{x} &\to \fb{\fb{x}{x}}{x} &
    \fb{x}{x} &\to \gb{x}{x} &
    \fb{x}{y} &\to y
    \tag{$\trsnl$}
    \label{ex:nl}
  \end{align*}
  over the signature $\asig = \{\tf,\tg,\ta\}$ where $\ta$ is a constant (necessary for the existence of ground terms).
  We choose the algebra $\algnl = \{\bot\}$ with
  $\interpret{\ta} = \bot$,
  $\bfunap{\interpret{\tf}}{\bot}{\bot} = \bot$, and
  $\bfunap{\interpret{\tg}}{\bot}{\bot} = \bot$.
  We label the symbols with the interpretations of their arguments,
  and define $\asigred = \{ \svoodoolabel{\tf}{\bot,\bot} \}$.

  Note that $\asigred$ does not contain $\svoodoolabel{\tg}{\bot,\bot}$.
  % only contani $\svoodoolabel{\tf}{\bot,\bot}$, but not $\svoodoolabel{\tg}{\bot,\bot}$.
  The reason is that using a finite algebra we can (in general) not recognize
  redex positions with respect to non-left-linear rules.
  By excluding $\svoodoolabel{\tg}{\bot,\bot}$ from $\asigred$ we allow rewriting
  below $\tg$ even when $\tg$ is the root of a redex. 
  %corresponds to a redex position.
  This is sound for proving outermost termination
  as it does not restrict the possible rewrite steps, 
  but allows only additional steps.
  The symbol $\svoodoolabel{\tf}{\bot,\bot}$ is part of $\asigred$;
  due to the rule $\fb{x}{y} \to y$ each occurrence of $\tf$ is a redex position.

  The dynamic labeling $\cxtext{\aalg}{\semlab}{\trsnl}$ is then formed by:
  \begin{align*}
    \bfunap{\svoodoolabel{\tg}{\bot,\bot}}{x}{x} &\to \bfunap{\svoodoolabel{\tf}{\bot,\bot}}{\bfunap{\svoodoolabel{\tf}{\bot,\bot}}{x}{x}}{x} 
    \\
    \bfunap{\svoodoolabel{\tf}{\bot,\bot}}{x}{x} &\to \bfunap{\svoodoolabel{\tg}{\bot,\bot}}{x}{x} 
    \tag{$\cxtext{\aalg}{\semlab}{\trsnl}$}
    \\
    \bfunap{\svoodoolabel{\tf}{\bot,\bot}}{x}{y} &\to y
  \end{align*}
  where $\amumap{\svoodoolabel{\tf}{\bot,\bot}} = \setemp$ and $\amumap{\svoodoolabel{\tg}{\bot,\bot}} = \{1,2\}$.
  This system is terminating which can be seen as follows.
  After an application of the first rule:
  \begin{align*}
    \contextfill{\acontext}{\bfunap{\svoodoolabel{\tg}{\bot,\bot}}{t}{t}} 
    &\to_{\cxtext{\aalg}{\semlab}{\trsnl},\,\samumap}
    \contextfill{\acontext}{\bfunap{\svoodoolabel{\tf}{\bot,\bot}}{\bfunap{\svoodoolabel{\tf}{\bot,\bot}}{t}{t}}{t}}
  \end{align*}
  the replacement map $\samumap$ prevents us from reducing the inner $\svoodoolabel{\tf}{\bot,\bot}$.
  Moreover, the second rule cannot be applied to the outer $\svoodoolabel{\tf}{\bot,\bot}$
  since the left and the right subterm are not equal.
  Thus the only rule applicable to the displayed subterm is $\bfunap{\svoodoolabel{\tf}{\bot,\bot}}{x}{y} \to y$
  which reduces the size of the term, and we can conclude termination by induction.

  Hence we conclude outermost ground termination of $\trsnl$ by Theorem~\ref{thm:cxtext:sound}.
  Actually the same \clabeling{} allows also to infer outermost termination, see Lemma~\ref{lem:ground}
  (we simply add a fresh constant $\zer$ and a unary symbol $\ssuc$ with interpretations
  $\interpret{\zer} = \bot$ and $\funap{\interpret{\ssuc}}{\bot} = \bot$).
\end{example}

\begin{example}\label{ex:nonlinb}
  \newcommand{\trsnlb}{\atrs_5}
  \newcommand{\algnlb}{\aalg_5}
  \newcommand{\signlb}{\iasig{5}}
  We consider the non-left-linear TRS $\trsnlb$ over the signature $\signlb = \{\tg,\ta,\tb\}$:
  \begin{align*}
    \ta &\to \gb{\ta}{\ta} &
    \gb{x}{x} &\to \tb
    \tag{$\trsnlb$}
    \label{ex:nlb}
  \end{align*}
  This TRS is outermost terminating.
  However, there exists no \clabeling{} that %correctly 
  recognizes redex positions with respect to the non-left-linear rule $\gb{x}{x} \to \tb$.
  A finite algebra cannot be used to check whether two arbitrary subterms $t_1$ and $t_2$ 
  of $\gb{t_1}{t_2}$ % of arbitrary depth 
  are equal.
  Thus it appears that, in order to have a sound transformation, 
  we cannot include any symbol $\svoodoolabel{\tg}{\alab}$ in the set $\asigred$ of redex symbols.
  But then rewriting below $\tg$ is allowed, and the rule $\ta \to \gb{\ta}{\ta}$ 
  would lead to non-termination of the dynamic labeling $\cxtext{\aalg}{\semlab}{\trsnlb}$.
  
  Nonetheless, in this particular example, the problem can be solved.
  %A finite algebra can in gernal not compare arbitrary terms $t_1$, $t_2$
  %for equality. 
  If some element $e$ of the algebra 
  is the interpretation of precisely one ground term $t$,
  then, of course, $\interpret{t_1} = \interpret{t_2} = e$ implies that $t_1 = t_2$. 
  Let us take the algebra $\algnlb = \{\bot,a\}$ with  
  $\interpret{\ta} = a$, $\interpret{\tb} = \bot$, and
  $\bfunap{\interpret{\tg}}{x}{y} = \bot$ for all $x,y \in \algnlb$.
  We use maximal labeling %label all symbols with the interpretations of their arguments,
  and define $\asigred = \{ \svoodoolabel{\tg}{a,a} \}$.
  That is, we mark redex positions $\gb{t}{t}$ only for the special case $t = \ta$.
  This \clabeling{} is sound since only redex positions are marked, 
  but it is not complete; not all redex positions are marked.
  Nevertheless, this labeling can be used to prove outermost ground termination of $\trsnlb$.
  The dynamic labeling $\cxtext{\aalg}{\semlab}{\trsnlb}$ of $\trsnlb$ consists of:
  \begin{align}
    \bfunap{\svoodoolabel{\tg}{a,\bot}}{\ta}{x} &\to \bfunap{\svoodoolabel{\tg}{\bot,\bot}}{\bfunap{\svoodoolabel{\tg}{a,a}}{\ta}{\ta}}{x} 
    & 
    \bfunap{\svoodoolabel{\tg}{a,a}}{x}{x} &\to \tb 
    \notag
    \\
    \bfunap{\svoodoolabel{\tg}{\bot,a}}{x}{\ta} &\to \bfunap{\svoodoolabel{\tg}{\bot,\bot}}{x}{\bfunap{\svoodoolabel{\tg}{a,a}}{\ta}{\ta}} 
    & \bfunap{\svoodoolabel{\tg}{\bot,\bot}}{x}{x} &\to \tb 
    \tag{$\cxtext{\aalg}{\semlab}{\trsnlb}$}
    \\
    \funap{\svoodoolabel{\topsymb}{\ta}}{\ta} &\to \funap{\svoodoolabel{\topsymb}{\bot}}{\bfunap{\svoodoolabel{\tg}{a,a}}{\ta}{\ta}}
    \notag
  \end{align}
  % Recall that the employed \clabeling{} was not complete.
  The employed \clabeling{} is not complete, and so the \csTRS{}
  $\cxtext{\aalg}{\semlab}{\trsnlb}$
  admits rewrite sequences (starting from correctly labeled terms)
  that do not correspond to outermost rewriting, e.g:
  \begin{align*}
    \funap{\svoodoolabel{\topsymb}{\bot}}{\bfunap{\svoodoolabel{\tg}{\bot,\bot}}{\bfunap{\svoodoolabel{\tg}{a,a}}{\ta}{\ta}}{\bfunap{\svoodoolabel{\tg}{a,a}}{\ta}{\ta}}}
    &\to_{\cxtext{\aalg}{\semlab}{\trsnlb},\samumap}
    \funap{\svoodoolabel{\topsymb}{\bot}}{\bfunap{\svoodoolabel{\tg}{\bot,\bot}}{\tb}{\bfunap{\svoodoolabel{\tg}{a,a}}{\ta}{\ta}}}
    \\
    &\to_{\cxtext{\aalg}{\semlab}{\trsnlb},\samumap}
    \funap{\svoodoolabel{\topsymb}{\bot}}{\bfunap{\svoodoolabel{\tg}{\bot,\bot}}{\tb}{\tb}}
    \\
    &\to_{\cxtext{\aalg}{\semlab}{\trsnlb},\samumap}
    \funap{\svoodoolabel{\topsymb}{\bot}}{\tb}
  \end{align*}
  Despite of this, the \csTRS{} can be shown to be terminating,
  and since the \clabeling{} was sound, we conclude outermost ground termination of $\trsnlb$
  by Theorem~\ref{thm:cxtext:sound}.
\end{example}

\begin{example}\label{ex:nonlinc}
  \newcommand{\trsnlc}{\atrs_6}
  \newcommand{\algnlc}{\aalg_6}
  In Examples~\ref{ex:nonlin} and~\ref{ex:nonlinb} we have seen 
  how our method can be applied
  to prove outermost termination of non-quasi-left-linear TRSs.
  We now consider an example which shows that not every
  non-left-linear TRS can be handled by our method:
  %% Let the TRS $\trsnlc$ consist of the following rules:
  \begin{align*}
    \f{x} &\to \gb{\f{x}}{\f{x}} &
    \gb{x}{x} &\to \tb
    \tag{$\trsnlc$}
    \label{ex:nlc}
  \end{align*}
  %% over the signature $\asig = \{f,g,b\}$.
  This TRS is outermost terminating.
  Now the trick used in Example~\ref{ex:nonlinb} does not work. %is not applicable.
  In order to construct a terminating \csTRS{} %dynamic labeling 
  $\cxtext{\aalg}{\semlab}{\trsnlc}$
  we need to forbid rewriting in all terms of the form
  $\gb{\f{t}}{\f{t}}$.
  This is impossible using a finite algebra.
\end{example}

We need the following adaptation of~\cite[Proposition~5.5.24]{ohle:02} for \csTRS{s};
the proof proceeds along the same lines.
\begin{theorem}\label{thm:adapt:ohlebusch}
  Let $\pair{\atrs}{\samumap}$ be a terminating many-sorted \csTRS.
  If the \csTRS{} obtained from $\pair{\atrs}{\samumap}$ by dropping sorts admits an infinite rewrite sequence,
  then $\pair{\atrs}{\samumap}$ is collapsing and duplicating.
  \qed
\end{theorem}

While for soundness of the transformation (Theorem~\ref{thm:cxtext:sound}) a sound labeling suffices,
for a complete transformation we need the \clabeling{} to be complete, maximal and core:

\begin{theorem}\label{thm:abstract:maxcomplete}
  Let $\atrs$ be a TRS over $\asig$, 
  and $\triple{\aalg}{\semlab}{\asigred}$ a complete, maximal, and core \clabeling\ for $\atrs$.
  Then $\cxtext{\aalg}{\semlab}{\atrs}$ is terminating if $\atrs$\ is outermost ground terminating.
\end{theorem}
\begin{proof}
  Assume that $\cxtext{\aalg}{\semlab}{\atrs}$ is not terminating.
  We turn $\cxtext{\aalg}{\semlab}{\atrs}$ into a sorted TRS.
  The sorts are chosen from the set $\aalg \join \{\topsymb\}$.
  Since the \clabeling{} is maximal, for each $n$-ary symbol
  $\tf^\lambda \in \labelsig{\asig_{\topsymb}}$ 
  we have $\lambda = {\tuple{a_1,\ldots,a_n}}$.
  We let $\svoodoolabel{\tf}{\alab}$ have input sort $\alab$ % $\tuple{a_1,\ldots,a_n}$ 
  and output sort $\funap{\interpret{\tf}}{a_1,\ldots,a_n}$.
  The only exception is the output sort of the symbols $\svoodoolabel{\topsymb}{\alab}$ 
  which we fix to be the sort $\topsymb$.
  Then by Theorem~\ref{thm:adapt:ohlebusch} 
  % an adaptation of~\cite[Proposition~5.5.24]{ohle:02} for \csTRS{s}
  together with non-collapsingness of $\cxtext{\aalg}{\semlab}{\atrs}$ %(see Remark~\ref{rem:cxtext:non-collapsing})
  yields the existence of a well-sorted infinite 
  rewrite sequence $\tau$ in $\cxtext{\atrs}{\semlab}{\atrs}$.
  Since the \clabeling{} is core, by Lemma~\ref{lem:core} there exists a ground term for every sort in $\aalg$.
  Thus by applying a ground substitution to $\tau$ 
  we obtain a well-sorted infinite ground term rewrite sequence $\tau'$.

  Well-sortedness implies correct labeling:
  for each well-sorted term $t \in \term{\labelsig{\asig_{\topsymb}}}{\setemp}$
  there exists a term $t' \in \term{\asig_{\topsymb}}{\setemp}$ such that $t = \dolabelg{t'}$.
  Moreover, %by well-sortedness 
  a symbol $\svoodoolabel{\topsymb}{\alab}$ can only occur at the top of a term.
  Without loss of generality we assume that every term in $\tau'$ 
  has $\svoodoolabel{\topsymb}{\alab}$ (for some $\alab \in \aalg$) as root
  (as rewriting below $\svoodoolabel{\topsymb}{\alab}$ is allowed 
  and context-sensitive rewriting is closed under $\samumap$-replacing contexts).
  Hence it suffices to show that
  for all terms $s,t \in \term{\asig}{\setemp}$ with 
  $\dolabelg{\funap{\topsymb}{s}} \red_{\cxtext{\atrs}{\semlab}{\atrs},\samumap} \dolabelg{\funap{\topsymb}{t}}$
  we have $s \outred_{\atrs} t$.
  By construction, each rule in $\cxtext{\atrs}{\semlab}{\atrs}$ 
  is the result of prepending contexts to, and labeling of, a rule in $\atrs$. 
  Let $\rho \funin s \red_{\atrs} t$ be the step corresponding to 
  $\dolabelg{\funap{\topsymb}{s}} \red_{\cxtext{\atrs}{\semlab}{\atrs},\samumap} \dolabelg{\funap{\topsymb}{t}}$.
  We show that $\rho$ is an outermost step.
  Assume there would be a redex $u$ above the rewrite position.
  Then by completeness of the \clabeling{} we get $\rootsymb{\dolabelg{u}} \in \asigred$.
  But then this symbol must be in $\dolabelg{\funap{\topsymb}{s}}$,
  either above the applied rule from $\cxtext{\atrs}{\semlab}{\atrs}$ 
  or within the prepended context.
  Both cases yield a contradiction:
  the former since $\amumap{\rootsymb{\dolabelg{u}}} = \setemp$
  would prohibit the $\samumap$-step, and
  the latter because we do not prepend symbols from $\asigred$.
\end{proof}

Let us consider the three conditions of Theorem~\ref{thm:abstract:maxcomplete} on \clabeling{s}: complete, maximal and core.
To see that completeness and maximality are necessary, 
we refer to Examples~\ref{ex:nonlinc}, %{ex:nonlin:noncomplete}, 
and~\ref{ex:noncomplete}, respectively.
The following example shows the need to restrict to core algebras:
\begin{example}\label{ex:whycore}
  \newcommand{\trscore}{\iatrs{7}}
  \newcommand{\algcore}{\iaalg{7}}
  Let $\trscore$ be the following term rewriting system:
  \begin{align}
    \f{x} &\to \gb{x}{\f{x}} &
    \gb{\ta}{x} &\to \ta &
    \gb{\f{x}}{y} &\to \ta &
    \gb{\gb{x}{y}}{z} &\to \ta
    \tag{$\trscore$}
  \end{align}
  This TRS is outermost ground terminating:
  First note that without the first rule $\trscore$ is terminating.
  So consider a rewrite step $\f{t} \to \gb{t}{\f{t}}$ for $t\in\ter{\{\tf,\tg,\ta\}}{\setemp}$.
  Then one of the three $\tg$-rules matches $\gb{t}{\f{t}}$ and blocks all inner rules 
  by the outermost strategy.

  We take the \cmodel{} $\algcore = \{0,1\}$ with 
  $\interpret{\ta} = \funap{\interpret{\tf}}{x} = \bfunap{\interpret{\tg}}{x}{y} = 0$,
  for all $x,y\in\algcore$.
  We let $\semlab$ be the maximal labeling 
  and define $\asigred = \{\svoodoolabel{\tf}{0}, \svoodoolabel{\tg}{0,0}\}$.
  Then the dynamic context extension $\cxtext{\aalg}{\semlab}{\trscore}$ contains, amongst others, 
  the following two rules:
  \begin{align*}
  \funap{\svoodoolabel{\tf}{0}}{x} &\to \bfunap{\svoodoolabel{\tg}{0,0}}{x}{\funap{\svoodoolabel{\tf}{0}}{x}} &
  \funap{\svoodoolabel{\tf}{1}}{x} &\to \bfunap{\svoodoolabel{\tg}{1,0}}{x}{\funap{\svoodoolabel{\tf}{1}}{x}}
  \end{align*}
  where $\amumap{\svoodoolabel{\tg}{0,0}} = \setemp$ and $\amumap{\svoodoolabel{\tg}{1,0}} = \{1,2\}$.
  Consequently, the second rule is not terminating, although the original TRS is outermost ground terminating.
  The \clabeling\ $\triple{\algcore}{\semlab}{\asigred}$ is complete for $\atrs$ and maximal,
  but not core.
  Note that there exists no ground term which has the interpretation $1$, 
  and hence the label $1$ should never occur.
\end{example}

Theorems~\ref{thm:cxtext:sound} and~\ref{thm:abstract:maxcomplete} are about outermost ground termination.
This is not a severe restriction, as by adding a fresh constant $\zer$ and a fresh unary symbol $\ssuc$
outermost ground termination implies (and for quasi-left-linear TRSs coincides with) outermost termination:
\begin{lemma}\label{lem:ground}
  A TRS $\atrs$ over $\asig$ is outermost terminating
  if $\atrs$ over $\asig \cup \{\ssuc,\zer\}$ is outermost ground terminating.
  If $\atrs$ is also quasi-left-linear, the converse direction holds as well.
\end{lemma}
\begin{proof}
  Let $\avars$ be countably infinite, and $\phi \funin \avars \to \nat$ a bijection.
  We define a substitution $\asubst$ by $\funap{\asubst}{x} = \funap{\ssuc^{\funap{\phi}{x}}}{\zer}$.
  Then, we have $\subst{\asubst}{s} \outred \subst{\asubst}{t}$
  whenever $s \outred t$ with $s,t \in \ter{\asig}{\avars}$,
  since the symbols $\ssuc$ and $\zer$ do not occur in any pattern of a rule,
  and
  for all $\apos, \bpos \in \pos{s}$ we have 
  $\subtrmat{s}{\apos} = \subtrmat{s}{\bpos} 
  \Leftrightarrow 
  \subtrmat{\subst{\asubst}{s}}{\apos} = \subtrmat{\subst{\asubst}{s}}{\bpos}$.
  This concludes the proof of the first part of the theorem.

  For the converse direction, let $\atrs$ be a quasi-left-linear TRS
  such that $\atrs$ over $\asig \cup \{\ssuc,\zer\}$ is not outermost ground terminating.
  Let $t \in \ter{\asig \cup \{\ssuc,\zer\}}{\setemp}$ be a ground term of minimal size
  admitting an infinite rewrite sequence $t = t_1 \outred t_2 \outred t_3 \outred \ldots$.
  By minimality, infinitely many of these steps must be in the prefix of $t$ not containing 
  $\ssuc$ and $\zer$.
  Let $t' \in \ter{\asig}{\avars}$ be obtained from $t$ by replacing
  all subterms with root symbol $\ssuc$ or $\zer$ in $t$ by a (arbitrary, but fixed) variable $x$.
  Then $t'$  admits an infinite outermost rewrite sequence as well.
  Note that by replacing the subterms no redex in the $\{\ssuc,\zer\}$-free prefix of $t$ is destroyed
  since the symbols $\ssuc$ and $\zer$ do not occur any rule pattern.
  Fresh redexes with respect to non-left-linear rules may be created 
  (but not with respect to left-linear rules).
  By quasi-left-linearity, at each position where a redex is created,
  there is also redex with respect to a left-linear rule.
  Hence, no additional redexes get blocked by outermost strategy.
\end{proof}

The following example shows that extending the signature with a single fresh constant~$\fun{0}$ only
is not enough for the implication:
\textit{$\atrs$ over the extended signature is outermost ground terminating}
$\mbox{}\implies\mbox{}$ \textit{$\atrs$ is outermost terminating}.

\begin{example}
  Consider the following term rewriting system $\atrs$:
  \begin{align*}
    \fb{x}{y} &\to \funap{\ta}{\fb{x}{y}} &
    \funap{\ta}{\fb{\tb}{x}} &\to \tb &
    \funap{\ta}{\fb{x}{\tb}} &\to \tb \\
    \funap{\ta}{\fb{x}{x}} &\to \tb &
    \funap{\ta}{\fb{\funap{\ta}{x}}{y}} &\to \tb &
    \funap{\ta}{\fb{x}{\funap{\ta}{y}}} &\to \tb \\
    &&
    \funap{\ta}{\fb{\fb{x}{y}}{z}} &\to \tb &
    \funap{\ta}{\fb{x}{\fb{y}{z}}} &\to \tb
  \end{align*}
  Because of the first rule $\atrs$ is not outermost terminating:
  \[
    \fb{x}{y} \outred \funap{\ta}{\fb{x}{y}} \outred \funap{\ta}{\funap{\ta}{\fb{x}{y}}} \outred \ldots
  \]
  but the TRS over the extended signature $\asig' = \{\ta,\tf,\tb,\fun{0}\}$ \emph{is} outermost ground terminating:
  Consider a step $\fb{s}{t} \to \funap{\ta}{\fb{s}{t}}$ with $s,t\in\ter{\asig'}{\setemp}$.
  If $s\neq\fun{0}$ one of the rules in the second column applies, 
  and if $t\neq\fun{0}$ then one of the rules in the third column is applicable.
  However if $s = t = \fun{0}$ then the rule $\funap{\ta}{\fb{x}{x}} \to \tb$ matches.
\end{example}

Note that for the second part of Lemma~\ref{lem:ground} we require quasi-left-linearity.
This requirement was erroneously missing from~\cite{endr:hend:09}, 
but is necessary as the following example illustrates.

\begin{example}
  We consider the following term rewriting system $\atrs$: % consisting of the rules:
  \begin{align*}
    \tfunap{\tf}{x}{y}{y} &\to \funap{\ta}{\tfunap{\tf}{x}{x}{y}}
    &
    \funap{\ta}{ \tfunap{\tf}{x}{\funap{\ta}{y}}{\funap{\ta}{z}} } &\to \bot
    \\    
    \tb &\to \funap{\ta}{\tb}
    &
    \funap{\ta}{ \tfunap{\tf}{x}{\tb}{\tb} } &\to \bot
    \\
    \funap{\ta}{ \tb } &\to \bot
    &
    \funap{\ta}{ \tfunap{\tf}{x}{\tfunap{\tf}{y_1}{y_2}{y_3}}{\tfunap{\tf}{z_1}{z_2}{z_3}} } &\to \bot
    \\
    \funap{\ta}{ \tfunap{\tf}{x}{x}{x} } &\to \bot
    &
    \funap{\ta}{ \tfunap{\tf}{x}{\bot}{\bot} } &\to \bot
  \end{align*}
  and explain why this system is outermost terminating.
  Without the rule $\tfunap{\tf}{x}{y}{y} \to \funap{\ta}{\tfunap{\tf}{x}{x}{y}}$ 
  outermost termination of $\atrs$ is obvious;
  Hence, in an infinite rewrite sequence this rule must be applied infinitely often.
  Let us consider a rewrite step
  $\cxtap{\acxt}{\tfunap{\tf}{t}{u}{u}} \outred \cxtap{\acxt}{\funap{\ta}{\tfunap{\tf}{t}{t}{u}}}$.
  If $t \in \avars$, then $u = t$ since no non-variable term rewrites to a variable;
  then $\funap{\ta}{ \tfunap{\tf}{x}{x}{x} } \to \bot$ is applicable
  and has priority (by outermost strategy) over all inner rewrite steps (and we terminate).
  If $t \not\in \avars$, then the second argument $u$ and third argument $t$ in $\funap{\ta}{\tfunap{\tf}{t}{t}{u}}$
  have to rewrite to a common non-variable reduct (in order to make the first rule applicable again).
  However, as soon as the common reduct is reached, one of the rules displayed on the right 
  would be applicable and have priority by outermost rewriting strategy.

  Nevertheless, $\atrs$ over the signature $\asig \cup \{\ssuc,\zer\}$ is not outermost ground terminating:
  \begin{align*}
    \funap{\ta}{\tfunap{\tf}{\funap{\fun{s}}{\tb}}{\funap{\fun{s}}{\funap{\ta}{\tb}}}{\funap{\fun{s}}{\funap{\ta}{\tb}}}}
    &\outred
    \funap{\ta}{\funap{\ta}{\tfunap{\tf}{\funap{\fun{s}}{\tb}}{\funap{\fun{s}}{\tb}}{\funap{\fun{s}}{\funap{\ta}{\tb}}}}} \\
    &\outred
    \funap{\ta}{\funap{\ta}{\tfunap{\tf}{\funap{\fun{s}}{\tb}}{\funap{\fun{s}}{\funap{\ta}{\tb}}}{\funap{\fun{s}}{\funap{\ta}{\tb}}}}} \\
    & \outred \ldots
  \end{align*}
\end{example}

\section{Dynamic Labeling}\label{sec:dynlab}

This paper is about employing context-sensitive rewriting to model outermost rewriting.
We do so by marking redexes, and forbid rewriting below them.
As we have seen, contracting a redex may create another redex higher up in the term tree.
Hence it may be necessary to update some labels during a rewrite step.
In Section~\ref{sec:cxtext} we defined a transformation where this updating 
was accounted for by extending rules with contexts.
Here we give an alternative transformation from TRSs to context-sensitive TRS{s}.
We call this tranformation `dynamic labeling'.
Instead of extending rules with contexts, 
%creating rules beforehand, guaranteeing preservation of correct labeling,
we now employ rewriting to propagate the changed information upward in the term tree,
and set the labels in the surrounding context right, step by step.
Again the \cdepth{} (Definition~\ref{def:cmodel}) serves as a bound:
here on the number of successive ancestor nodes that have to be relabeled.
Each original rewrite step will give rise to a corresponding step and a bounded number ($\leq$ the \cdepth{}) 
of auxiliary steps in the transformed system. 
Thus, although the derivational complexity (the length of rewrite sequences) 
is changed, this is only by a constant factor.
We prove that dynamic labeling is sound for arbitrary TRSs.
Moreover, for left-linear TRSs, the method is complete in a weakened sense, 
see Theorem~\ref{thm:dynlab:complete}.
In Section~\ref{sec:evaluation}, we compare the performance of this method 
to the one of dynamic context extension described in Section~\ref{sec:cxtext}.

We begin with an analysis for evaluating %which relabel steps are needed, i.e., 
which value changes can occur by rewriting and need to be propagated upward.
%That is, we compute pairs of the form 
%$\pair{\interpret{\cxtfill{\acxt}{\subst{\asubst}{\ell}}}}{\interpret{\cxtfill{\acxt}{\subst{\asubst}{r}}}}$ 
%where $\ell\to r \in \atrs$ 
%and where $\interpret{\cxtfill{\acxt}{\subst{\asubst}{\ell}}} \neq \interpret{\cxtfill{\acxt}{\subst{\asubst}{r}}}$.
As we will see, this restricts the number of auxiliary `relabel symbols', and, in particular,
the number of `relabeling rules'.

\begin{definition}\label{def:value-change-pairs}
  Let $\atrs$ be a TRS over $\asig$, 
  and let $\triple{\aalg}{\semlab}{\asigred}$ be a \clabeling\ for $\atrs$.
  For $i = 0,1,2,\ldots$\,, we define the set $\ireachables{i}\subseteq\aalg\times\aalg$ inductively by:
  \begin{align*}
    \ireachables{0}
    =
    \{\,
      \pair{\interpreta{l}{\alpha}}{\interpreta{r}{\alpha}}
      \where
      {\ell \to r} \in \atrs,\,
      & \mbox{}
      \alpha \funin \vars{\ell} \to \aalg,\,
      \interpreta{l}{\alpha} \neq \interpreta{r}{\alpha}
    \,\}
    \\
    \ireachables{i+1} 
    =
    \ireachables{i}
    \join
    \{\,
      \pair{\funap{\interpret{\tf}}{\vec{a},b,\vec{c}}}{\funap{\interpret{\tf}}{\vec{a},b',\vec{c}}}
      \where \mbox{}
      & 
      \tf\in\asig,\,
      % \arity{\tf} = \veclength{\vec{\ta}} + 1 + \veclength{\vec{c}},\;
      \vec{a} \cdot b \cdot \vec{c}\,\in\aalg^{\arity{\tf}},\,
      \pair{b}{b'} \in \ireachables{i}, 
      \\
      &
      \svoodoolabel{\tf}{\funap{\slabelf{\tf}}{\vec{a},b,\vec{c}}} \in { \labelsig{\topsig{\asig}} \setminus \asigred }      
    \,\}
  \end{align*}
  Then we define the set $\reachables{\semlab}{\atrs}$ of \emph{value-change pairs} by:
  \[
    \reachables{\semlab}{\atrs} = \ireachables{i} \setminus \{\pair{a}{a}\}_{a\in\aalg}
  \]
  with $i$\, the least number
  such that $\ireachables{i+1} = \ireachables{i}$.
\end{definition}

%%The set of value-change pairs is then used in the following definition.
The `dynamic labeling' $\dynlab{\aalg}{\semlab}{\atrs}$ of a TRS~$\atrs$\, is partitioned into two sets of rules.
%$\dynlaborg{\aalg}{\semlab}{\atrs}$ and $\dynlabprp{\aalg}{\semlab}{\atrs}$.
The first set is denoted by $\dynlaborg{\aalg}{\semlab}{\atrs}$
and consists of a semantic labeling of the original rules,
where, additionally, a right-hand side is prefixed by a symbol $\srelabel{a}{a'}$ 
whenever application of the rule causes a change of interpretation from $a$ to $a'$.
%% i.e., when $\pair{\ta}{\ta'}\in\reachables{\semlab}{\atrs}$.
The second set, $\dynlabprp{\aalg}{\semlab}{\atrs}$, 
is a set of rules for relabeling the context of the rule application.
A symbol $\srelabel{a}{a'}$, with $\pair{a}{a'}\in\reachables{\semlab}{\atrs}$,
indicates that the value of its subterm has changed from $a$ to $a'$,
and the rules in $\dynlabprp{\aalg}{\semlab}{\atrs}$ 
take care of propagating this change of value upward in the term.

\begin{definition}[Dynamic labeling]\label{def:dynlab}
  Let $\atrs$ be a TRS over $\asig$, 
  and let $\triple{\aalg}{\semlab}{\asigred}$ be a \clabeling\ for $\atrs$.
  The TRS $\dynlab{\aalg}{\semlab}{\atrs}$ over the signature 
  $\labelsig{\asig_{\topsymb}} \join \{ \srelabel{a}{a'} \where { \pair{a}{a'} \in \reachables{\semlab}{\atrs} } \}$
  %$\overline{\asig}$ as follows:
  is defined by 
  $\dynlab{\aalg}{\semlab}{\atrs} = \dynlaborg{\aalg}{\semlab}{\atrs} \join \dynlabprp{\aalg}{\semlab}{\atrs}$. 
  Here the set $\dynlaborg{\aalg}{\semlab}{\atrs}$ of \emph{labeled rules} contains,   
  for each rule ${\ell\to r}\in\atrs$ and assignment $\alpha\funin{\vars{\ell}\to\aalg}$,
  one of the rules:
  \begin{align*}
    \dolabel{\ell}{\alpha} \to 
    \begin{cases}
      \dolabel{r}{\alpha}
      & \text{if $\interpreta{\ell}{\alpha} \eq \interpreta{r}{\alpha}$}
      \\
      \relabel{\interpreta{\ell}{\alpha}}{\interpreta{r}{\alpha}}{\dolabel{r}{\alpha}}
      & \text{otherwise}
    \end{cases}  
  \end{align*}
  Secondly, the set $\dynlabprp{\aalg}{\semlab}{\atrs}$ of 
  \emph{relabeling rules} %\emph{$\ssrelabel$ propagation rules} 
  contains,
  for each $n$-ary $\tf\in\asig$, $\pair{b}{b'}\in\reachables{\semlab}{\atrs}$, 
  and $\tuple{\vec{a},b,\vec{c}} \in \aalg^n$ 
  %$\vec{a}\cdot b \cdot\vec{c}\in{\aalg^n}$ 
  such that 
  $\svoodoolabel{\tf}{\alab} \in { \labelsig{\topsig{\asig}} \setminus \asigred }$ 
  with $\alab = \funap{\slabelf{\tf}}{\vec{a},b,\vec{c}}$,
  one of the rules:
  \begin{align*}
    \funap{\svoodoolabel{\tf}{\alab}}{\vec{x},\relabel{b}{b'}{y},\vec{z}}
    \to
    \begin{cases}
      \funap{\svoodoolabel{\tf}{\alab'}}{\vec{x},y,\vec{z}}
      &\text{if $d \eq d'$} \\
      \relabel{d}{d'}{\funap{\svoodoolabel{\tf}{\alab'}}{\vec{x},y,\vec{z}}}
      &\text{otherwise}
    \end{cases}
  \end{align*}
  where 
  $\alab' = \funap{\slabelf{\tf}}{\vec{a},b',\vec{c}}$,
  $d = \funap{\interpret{\tf}}{\vec{a},b,\vec{c}}$,
  $d' = \funap{\interpret{\tf}}{\vec{a},b',\vec{c}}$,
  $\veclength{\vec{x}} = \veclength{\vec{a}}$, 
  and $\veclength{\vec{z}} = \veclength{\vec{c}}$.

  The \emph{dynamic labeling of $\atrs$ (with respect to the \clabeling~$\triple{\aalg}{\semlab}{\asigred}$)}
  is the context-sensitive TRS~$\pair{\dynlab{\aalg}{\semlab}{\atrs}}{\samumap}$,
  where the replacement map $\samumap$ is defined by 
  $\amumap{\srelabel{a}{a'}} = \setemp$ for all $\pair{a}{a'}\in\reachables{\semlab}{\atrs}$,
  $\amumap{\tf} = \setemp$ if $\tf \in \asigred$,
  and $\amumap{\tf} = \{1,\ldots,\arity{\tf}\}$ otherwise,
  for all $\tf\in\labelsig{\asig_{\topsymb}}$.
  Whenever $\asigred$ is clear from the context, we leave $\samumap$ implicit,
  and overload the notation $\dynlab{\aalg}{\semlab}{\atrs}$ to denote 
  $\pair{\dynlab{\aalg}{\semlab}{\atrs}}{\samumap}$.
\end{definition}

\begin{example}
  We revisit the TRS~$\trsfffx$ from Example~\ref{ex:fffx:cmodel} 
  for which we worked out the static and dynamic context extensions 
  in Examples~\ref{ex:fffx:static:cxtext} and~\ref{ex:fffx:dyn:cxtext}.
  We repeat its definition and the \clabeling{} from Example~\ref{ex:fffx:clabeling}: % employed there:
  $\trsfffx$ is the TRS over $\asig = \{\ta,\tf,\tg\}$ %(where $\ta$ is a constant) 
  consisting of the rules:
  \begin{align}
    \f{\g{x}} &\to \f{\f{\g{x}}} 
    & \f{\f{\f{x}}} &\to x
    \tag{$\trsfffx$}
  \end{align} 
  A \cmodel{} for $\trsfffx$ is formed by the $\asig$-algebra 
  $\algfffx = \{\bot,f,\mit{ff},g\}$
  with interpretation:
  \begin{align}
    \interpret{\fun{c}} = \bot 
    && 
    \funap{\interpret{\tf}}{\bot} = \funap{\interpret{\tf}}{g} = f
    && 
    \funap{\interpret{\tf}}{f} = \funap{\interpret{\tf}}{\mit{ff}} = \mit{ff}
    &&
    \funap{\interpret{\tg}}{x} = g
    \tag{$\algfffx$}
  \end{align}
  for all $x\in\algfffx$.
  Furthermore, $\pair{\algfffx}{\semlab}$ denotes the maximal labeling for $\trsfffx$, 
  and $\sigredex{\asig} = \{ \svoodoolabel{\tf}{g}, \svoodoolabel{\tf}{\mit{ff}} \}$.
  Then $\triple{\algfffx}{\semlab}{\sigredex{\asig}}$
  forms a sound and complete \clabeling\ of $\trsfffx$.
  Also note that $\algfffx$ forms a core algebra;
  for each value $e\in\algfffx$ there is a ground term $t$ such that 
  $\interpret{t} = e$.
  We first compute the set $\reachables{\semlab}{\trsfffx}$ of value-change pairs. 
  For the initial set $\ireachables{0}$, note that the rule 
  $\f{\g{x}} \to \f{\f{\g{x}}}$ 
  changes the interpretation from $f$ to $\mit{ff}$, regardless of the value assigned to $x$.
  % so $\pair{f_1}{f_2} \in \ireachables{0}$ 
  The other rule creates three value-change pairs; one for each of the values $g,\bot,f$ assigned to $x$. 
  If the interpretation of $x$ is $\mit{ff}$ there is no change.
  Hence we get:
  \[
    \ireachables{0} = \{ \pair{f}{\mit{ff}} , \pair{\mit{ff}}{\bot} , \pair{\mit{ff}}{f} , \pair{\mit{ff}}{g} \}
  \]
  % And that is it. 
  All symbols $\srelabel{e}{e'}$ with $\pair{e}{e'} \in \ireachables{0}$ 
  will disappear in one relabeling step, whence $\reachables{\semlab}{\trsfffx} = \ireachables{0}$.
  The dynamic labeling of $\trsfffx$ then is
  $\dynlab{\algfffx}{\semlab}{\trsfffx} = \dynlaborg{\algfffx}{\semlab}{\trsfffx} \join \dynlabprp{\algfffx}{\semlab}{\trsfffx}$
  where $\dynlaborg{\algfffx}{\semlab}{\trsfffx}$ consists of the rules:
  \begin{align*}
    \funap{\svoodoolabel{\tf}{g}}{\funap{\svoodoolabel{\tg}{e}}{x}} 
    & \to \relabel{f}{\mit{ff}}{\funap{\svoodoolabel{\tf}{f}}{\funap{\svoodoolabel{\tf}{g}}{\funap{\svoodoolabel{\tg}{e}}{x}}}}
    & \text{for all $e \in \algfffx$} %  (4 rules) 
    \\
    \funap{\svoodoolabel{\tf}{\mit{ff}}}{\funap{\svoodoolabel{\tf}{\mit{ff}}}{\funap{\svoodoolabel{\tf}{e'}}{x}}} 
    & \to \relabel{e}{e'}{x}
    & \text{for all $\pair{e}{e'} \in \{ \pair{\mit{ff}}{\bot} , \pair{\mit{ff}}{f} , \pair{\mit{ff}}{g} \}$}
    \\
    \funap{\svoodoolabel{\tf}{\mit{ff}}}{\funap{\svoodoolabel{\tf}{\mit{ff}}}{\funap{\svoodoolabel{\tf}{\mit{ff}}}{x}}} &\to x
  \end{align*}
  and where $\dynlabprp{\algfffx}{\semlab}{\trsfffx}$ is formed by:
  \begin{align*}
    \funap{\svoodoolabel{\tg}{e}}{\relabel{e}{e'}{x}} &\to \funap{\svoodoolabel{\tg}{e'}}{x}
    &\text{for all $\pair{e}{e'}\in\reachables{\semlab}{\trsfffx}$ } % (4 rules)
    \\
    \funap{\svoodoolabel{\tf}{f}}{\relabel{f}{\mit{ff}}{x}} &\to \funap{\svoodoolabel{\tf}{\mit{ff}}}{x}
  \end{align*}
  In total the dynamic labeling of $\trsfffx$ has 13 rules. 
  %
  %\noindent
  Had we not restricted the construction of the set of the relabeling rules to 
  the `reachable' symbols $\srelabel{e}{e'}$ 
  (by the requirement $\pair{e}{e'}\in\reachables{\semlab}{\trsfffx}$ in Definition~\ref{def:dynlab}),
  we would have come up with 18 instead of 5 relabeling rules.
\end{example}

\begin{example}
  We reconsider the term rewrite system~$\trsduplrhs$ from Example~\ref{ex:dupl_rhs}:
  \begin{align*}
    \fb{\h{x}}{\fun{c}}
    & \to \fb{\funap{\fun{i}}{x}}{\funap{\fun{s}}{x}}
    &
    \funap{\fun{i}}{x}
    & \to \h{x}
    \\
    \fb{\funap{\fun{i}}{x}}{y}
    & \to x
    &
    \h{x}
    & \to \fb{\h{x}}{\fun{c}}
  \end{align*}
  and the algebra $\algduplrhs = \pair{\{\bot,c,h,i\}}{\sinterpret}$ 
  with $\sinterpret$ defined, for all $x,y\in\algduplrhs$, as follows:
  \begin{align*}
    \interpret{\fun{c}}=c
  &&
  \funap{\interpret{\sh}}{x} = h
  &&
  \funap{\interpret{\fun{i}}}{x} = i
  &&
  \bfunap{\interpret{\tf}}{x}{y} = \funap{\interpret{\fun{s}}}{x} = \bot
  \end{align*}
  Moreover, we employ minimal labeling again; see Example~\ref{ex:dupl_rhs}. 

  The set of change-value pairs is:
  \begin{align*}
    \reachables{\semlab}{\trsduplrhs}
    = \{ \pair{\bot}{c} , \pair{\bot}{h} , \pair{\bot}{i} , \pair{i}{h} , \pair{h}{\bot} \}
  \end{align*}
  The set $\dynlaborg{\algduplrhs}{\semlab}{\trsduplrhs}$ of labeled rules is constructed thus: 
  \begin{align*}
    \bfunap{\rmark{\tf}}{\funap{\rmark{\sh}}{x}}{\fun{c}}
    & \to \bfunap{\rmark{\tf}}{\funap{\rmark{\fun{i}}}{x}}{\funap{\fun{s}}{x}} 
    \\
    \bfunap{\rmark{\tf}}{\funap{\rmark{\fun{i}}}{x}}{y}
    & \to x
    \\
    \bfunap{\rmark{\tf}}{\funap{\rmark{\fun{i}}}{x}}{y}
    & \to \relabel{e}{e'}{x}
    & \text{$\pair{e}{e'} \in \{ \pair{\bot}{c} , \pair{\bot}{h} , \pair{\bot}{i} \}$}
    \\
    \funap{\rmark{\fun{i}}}{x}
    & \to \relabel{i}{h}{\funap{\rmark{\sh}}{x}}
    \\
    \funap{\rmark{\sh}}{x}
    & \to \relabel{h}{\bot}{\bfunap{\rmark{\tf}}{\funap{\rmark{\sh}}{x}}{\fun{c}}}
  \end{align*}%
  There are four rules with left-hand side $\ell = \bfunap{\rmark{\tf}}{\funap{\rmark{\fun{i}}}{x}}{y}$,
  one for each value assigned to~$x$. 
  In case $\funap{\alpha}{x} = \bot$ there is no change of interpretation, 
  for we have that $\interpreta{\ell}{\alpha} = \bot$ for all $\alpha\funin\{x,y\}\to\algduplrhs$
  and so no $\ssrelabel$ symbol is inserted.
  But if, for instance, $\interpret{\funap{\asubst}{x}} = c$ for some substitution $\asubst$,
  then some labels in the context~$\acxt$ of a rewrite step 
  $\cxtfill{\acxt}{\subst{\asubst}{\ell}} \to \cxtfill{\acxt}{\funap{\asubst}{x}}$
  have to be updated, since the value of $\cxthole$ has changed from $\bot$ to $c$,
  whence the insertion of $\srelabel{\bot}{c}$ to the right-hand side $x$.

  The set $\dynlabprp{\algduplrhs}{\semlab}{\trsduplrhs}$ 
  of relabeling rules is formed by:
  \begin{align*}
    \fb{\relabel{e}{e'}{x}}{y}
    & \to \fb{x}{y} 
    & \text{$\pair{e}{e'} \in \{ \pair{\bot}{c} , \pair{\bot}{h} , \pair{h}{\bot} \}$}
    \\
    \fb{\relabel{e}{e'}{x}}{y}
    & \to \bfunap{\rmark{\tf}}{x}{y}
    & \text{$\pair{e}{e'}\in\{\pair{\bot}{h} , \pair{\bot}{i}\}$}
    \\
    \fb{x}{\relabel{\bot}{c}{y}}
    & \to \bfunap{\rmark{\tf}}{x}{y} 
    \\
    \fb{x}{\relabel{\bot}{c}{y}}
    & \to \fb{x}{y} 
    & \text{$\pair{e}{e'} \in \reachables{\semlab}{\trsduplrhs}$}
    \\
    \funap{\fun{s}}{\relabel{e}{e'}{x}}
    & \to \funap{\fun{s}}{x} 
    & \text{$\pair{e}{e'}\in\reachables{\semlab}{\trsduplrhs}$}
    \\
    \funap{\topsymb}{\relabel{e}{e'}{x}}
    & \to \funap{\topsymb}{x} 
    & \text{$\pair{e}{e'}\in\reachables{\semlab}{\trsduplrhs}$}
  \end{align*}
  Some remarks for clarification: 
  First, note that all $\ssrelabel$ symbols disappear upon one relabeling step.
  Secondly, observe the overlap in, for example, the rules with left-hand side 
  $\fb{\relabel{\bot}{h}{x}}{y}$.
  If the value assigned to $y$ is $c$, 
  then a redex is created; this is witnessed by the marked symbol $\rmark{\tf}$ on the right.
  For other values for $y$, this is not the case.
  Also note that there is no rule for
  $t = \fb{\relabel{i}{h}{x}}{y}$. 
  This is because when the left argument of $\tf$ is interpreted as $i$,
  then $t$ forms a redex, and so $\tf$ should be marked. 
  Definition~\ref{def:dynlab} does not allow $\ssrelabel$ symbols 
  to commute with redex symbols.
  Intuitively, a $\ssrelabel$ symbol is a witness of a rewrite step 
  which we do not want to occur inside other redexes, as we want to model outermost terminination. 
  However, more technically, sometimes illegal (i.e., non-outermost) $\ssrelabel$ steps are allowed.
  This is illustrated in Example~\ref{ex:dyncomplete}.
  The point is that by preventing $\ssrelabel$ symbols to commute with redex symbols,
  for local completeness (Theorem~\ref{thm:dynlab:complete})
  it is as if illegal steps never happened.
\end{example}

\begin{remark}
  \newcommand{\ssrelabelto}{\ssrelabel}%
  \newcommand{\srelabelto}{\super{\ssrelabelto}}%
  \newcommand{\relabelto}[1]{\funap{\srelabelto{#1}}}%
  \newcommand{\fgfxnr}{8}%
  \newcommand{\trsfgfx}{\iatrs{\fgfxnr}}%
  \newcommand{\algfgfx}{\iaalg{\fgfxnr}}%
  We elaborate on the role of the element~$a$ in $\srelabel{a}{a'}$.
  Whenever the application of a rule 
  $\contextfill{\acontext}{\subst{\asubst}{\ell}} \to \contextfill{\acontext}{\subst{\asubst}{r}}$ 
  changes the interpretation,
  i.e., $\interpret{\subst{\asubst}{\ell}} \neq \interpret{\subst{\asubst}{r}}$,
  then a symbol $\srelabel{\interpret{\subst{\asubst}{\ell}}}{\interpret{\subst{\asubst}{r}}}$ is inserted.
  A term of the form $\relabel{a}{a'}{t'}$ can be thought of as a witness 
  of a rewrite step $t\to t'$ 
  causing a change of interpretation from $a = \interpret{t}$ to $a' = \interpret{t'}$.
  This change of the value then needs to be propagated upward to update the labels accordingly,
  using the relabeling rules from $\dynlabprp{\aalg}{\semlab}{\atrs}$.
  At first sight, the value $a$ in $\relabel{a}{a'}{t}$ seems redundant for relabeling:
  why would we store the previous value? 
  However, the label $a$ is important in order to restrict the number of applicable rules,
  and to have a bound on the number of relabeling steps. 
  To see this, consider the system~$\trsfgfx$ with single rewrite rule:
  \begin{align}
    \f{\g{\f{x}}} \to \fun{d}
    \tag{$\trsfgfx$}
    \label{ex:fgfx}
  \end{align}
  and the algebra $\algfgfx = \{ \bot, f, \mit{gf} \}$ with
  $\funap{\interpret{\tf}}{x} = f$ for all $x \in \algfgfx$,
  $\funap{\interpret{\tg}}{\mit{f}} = \mit{gf}$, $\funap{\interpret{\tg}}{x} = \bot$ for all $x \ne f$,
  and $\interpret{\fun{d}} = \bot$.
  We employ minimal labeling, that is, 
  only $\funap{\slabelf{\tf}}{\mit{gf}} = \star$\,,
  and all the other symbols are unlabeled.
  
  The dynamic labeling $\dynlab{\algfgfx}{\semlab}{\trsfgfx}$ %of $\trsfgfx$
  %as defined in Definition~\ref{def:dynlab},
  gives rise to two labelings of the original rule:
  \begin{align}
    \funap{\rmark{\tf}}{\g{\f{x}}} &\to \relabel{f}{\bot}{\fun{d}}
    \label{rule:fgfx:dynlab:1}
    \\
    \funap{\rmark{\tf}}{\g{\funap{\rmark{\tf}}{x}}} &\to \relabel{f}{\bot}{\fun{d}}
    \label{rule:fgfx:dynlab:2}
  \end{align}
  And, among the fourteen rules in $\dynlab{\algfgfx}{\semlab}{\trsfgfx}$ for updating labels,
  we find the following two:
  \begin{align}
    \g{\relabel{f}{\bot}{x}} &\to \relabel{\mit{gf}}{\bot}{\g{x}}
    \label{rule:fgfx:dynlab:3}
    \\ 
    \g{\relabel{\mit{gf}}{\bot}{x}} &\to \g{x}
    \label{rule:fgfx:dynlab:4}
  \end{align}

  \noindent
  Now consider the term $t = \funap{\topsymb}{\g{\cdots(\g{\g{\g{\funap{\svoodoolabel{\tf}{\star}}{\g{\f{\fun{d}}}}}}})}}$,
  %with a redex within a context $\funap{\topsymb{\g{\cdots(\g{})}}}$ consisting of  
  and the rewrite sequence:
  \begin{align*}
    t
    \to_{\ref{rule:fgfx:dynlab:1},\samumap} {}
    &
    \funap{\topsymb}{\g{\cdots(\g{\g{\g{\relabel{f}{\bot}{\fun{d}}}}})}}
    \\
    {} \to_{\ref{rule:fgfx:dynlab:3},\samumap} {}
    &
    \funap{\topsymb}{\g{\cdots(\g{\g{\relabel{\mit{gf}}{\bot}{\g{\fun{d}}}}})}}
    \\
    {} \to_{\ref{rule:fgfx:dynlab:4},\samumap} {}
    &
    \funap{\topsymb}{\g{\cdots(\g{\g{\g{\fun{d}}}})}}
  \end{align*}
  After an application of~\eqref{rule:fgfx:dynlab:1}, relabeling takes two steps,
  resulting in a correctly labeled term.

  In the alternative, let us say `forgetful' version of dynamic labeling,
  where the `from' value $a$ in symbols $\srelabel{a}{b}$ is omitted,
  the rules~\eqref{rule:fgfx:dynlab:1}--\eqref{rule:fgfx:dynlab:4} look like this:
  \begin{align}
    \funap{\rmark{\tf}}{\g{\f{x}}} &\to \relabelto{\bot}{\fun{d}}
    \label{rule:fgfx:dynlab:forgetful:1}
    \tag{$\ref{rule:fgfx:dynlab:1}'$}
    \\
    \funap{\rmark{\tf}}{\g{\funap{\rmark{\tf}}{x}}} &\to \relabelto{\bot}{\fun{d}}
    \label{rule:fgfx:dynlab:forgetful:2}
    \tag{$\ref{rule:fgfx:dynlab:2}'$}
    \\ 
    \g{\relabelto{\bot}{x}} &\to \relabelto{\bot}{\g{x}}
    \label{rule:fgfx:dynlab:forgetful:3}
    \tag{$\ref{rule:fgfx:dynlab:3}'$}
    \\
    \g{\relabelto{\bot}{x}} &\to \g{x}
    \label{rule:fgfx:dynlab:forgetful:4}
    \tag{$\ref{rule:fgfx:dynlab:4}'$}
  \end{align}
  %Note the loss of orthogonality by this transformation.
  %Indeed, due 
  Due to the overlap in rules~\eqref{rule:fgfx:dynlab:forgetful:3} and~\eqref{rule:fgfx:dynlab:forgetful:4},
  the resulting \csTRS\ has a rewrite sequence 
  from $t$ where the symbol $\srelabelto{\bot}$ goes up all the way to the top:
  \begin{align*}
    t \to_{\text{\ref{rule:fgfx:dynlab:forgetful:1}},\samumap} {}
    &
    \funap{\topsymb}{\g{\cdots(\g{\g{\g{\relabelto{\bot}{\fun{d}}}}})}}
    \\
    {} \to_{\text{\ref{rule:fgfx:dynlab:forgetful:3}},\samumap} {}
    &
    \funap{\topsymb}{\g{\cdots(\g{\g{\relabelto{\bot}{\g{\fun{d}}}}})}}
    \\
    {} \to_{\text{\ref{rule:fgfx:dynlab:forgetful:3}},\samumap} {}
    &
    \funap{\topsymb}{\g{\cdots(\g{\relabelto{\bot}{\g{\g{\fun{d}}}}})}}
    \\
    {} \to_{\text{\ref{rule:fgfx:dynlab:forgetful:3}},\samumap} {}
    &
    \ldots
  \end{align*}
\end{remark}

From the following lemma it follows that every $\ssrelabel$ symbol
can be rewritten at most $\cdpth{\aalg}{\atrs}$ times (before it vanishes).
By rewriting a `$\ssrelabel$ symbol' we refer to a notion of
residuals that extends the usual definition of orthogonal projection~\cite{terese:03}
with a concept suggested by the definition of $\dynlabprp{\aalg}{\semlab}{\atrs}$\,:
Whenever we have a rule of the form:
\begin{align*}
  \funap{\svoodoolabel{\tf}{\alab}}{\vec{x},\relabel{a}{a'}{\dolabelg{t}},\vec{z}}
  \to \relabel{b}{b'}{\funap{\svoodoolabel{\tf}{\blab}}{\vec{x},\dolabelg{t},\vec{z}}}
\end{align*}
then we call $\srelabel{b}{b'}$ in the right-hand side
a residual of $\srelabel{a}{a'}$ in the left-hand side.

\newcommand{\reachrel}{\leadsto}
\newcommand{\srelabelw}{w}
\newcommand{\relabelw}{\funap{\srelabelw}}
\begin{lemma}\label{lem:relabel:bound}
  Let $\atrs$ be a TRS over $\asig$, 
  and let $\triple{\aalg}{\semlab}{\asigred}$ be a \clabeling\ for $\atrs$.
  We define the relation ${\reachrel} \subseteq {\reachables{\semlab}{\atrs} \times \reachables{\semlab}{\atrs}}$ by:
  \begin{align*}
    \pair{b}{b'} \reachrel \pair{\funap{\interpret{\tf}}{\vec{a},b,\vec{c}}}{\funap{\interpret{\tf}}{\vec{a},b',\vec{c}}}
    && \text{ for every } &\ \tf\in\asig,\, 
       \pair{b}{b'} \in \reachables{\semlab}{\atrs},\,
       \vec{a} \cdot b \cdot \vec{c}\,\in\aalg^{\arity{\tf}},\\
    &&   &\ \svoodoolabel{\tf}{\funap{\slabelf{\tf}}{\vec{a},b,\vec{c}}} \in { \labelsig{\topsig{\asig}} \setminus \asigred }
  \end{align*}
  Then $\reachrel$ is well-founded and every $\reachrel$ path has length $\le \cdpth{\aalg}{\atrs}$.
\end{lemma}
\begin{proof}
  By definition of value-change pairs we have that for every pair $\pair{b}{b'} \in \reachables{\semlab}{\atrs}$
  there exists a rule $\arule \in \atrs$ and assignment $\alpha \funin \vars{\alhs} \to \aalg$
  such that $\pair{\interpreta{\alhs}{\alpha}}{\interpreta{\arhs}{\alpha}} \reachrel^* \pair{b}{b'}$.
  
  Assume, to arrive at a contradiction, there exists a sequence 
  \begin{align*}
   \pair{\interpreta{\alhs}{\alpha}}{\interpreta{\arhs}{\alpha}} 
   = \pair{b_0}{b'_0} \reachrel \pair{b_1}{b'_1} \reachrel \ldots \reachrel \pair{b_m}{b'_m}
  \end{align*}
  with $m > \cdpth{\aalg}{\atrs}$.
  For $i = 0,1,\ldots,m$ we construct thin contexts $\bcxt_i$ 
  and assignments $\alpha_i : \vars{\bcxt_i} \to \aalg$ 
  such that 
  $b_i = \interpreta{\cxtfill{\bcxt_i}{\alhs}}{\alpha_i}$
  and $b'_i = \interpreta{\cxtfill{\bcxt_i}{\arhs}}{\alpha_i}$.
  We begin with $\bcxt_0 = \cxthole$ and $\alpha_0 = \alpha$.
  Then we have $b_0 = \interpreta{\alhs}{\alpha}$ and $b'_0 = \interpreta{\arhs}{\alpha}$.
  For $i = 1,\ldots,m$ there exist
  $\tf_i \in \asig$, and $\vec{a_i} \cdot b_{i-1} \cdot \vec{c_i} \in \aalg^{\arity{\tf_i}}$
  such that
  $b_i = \funap{\interpret{\tf_i}}{\vec{a_i},b_{i-1},\vec{c_i}}$
  and $b'_i = \funap{\interpret{\tf_i}}{\vec{a_i},b'_{i-1},\vec{c_i}}$.
  We pick fresh variables $\vec{x_i}$ and $\vec{z_i}$
  with $\lstlength{\vec{x_i}} = \lstlength{\vec{a_i}}$ and $\lstlength{\vec{z_i}} = \lstlength{\vec{c_i}}$,
  and define $\bcxt_i = \funap{\tf_i}{\vec{x_i},\bcxt_{i-1},\vec{z_i}}$,
  and $\alpha_i$ is $\alpha_{i-1}$ extended by
  mapping variables $\vec{x_i}$ to the corresponding $\vec{a_i}$ and $\vec{z_i}$ to $\vec{c_i}$.
  It follows that $b_i = \interpreta{\cxtfill{\bcxt_i}{\alhs}}{\alpha_i}$
  and $b'_i = \interpreta{\cxtfill{\bcxt_i}{\arhs}}{\alpha_i}$.
  But then
  $\interpreta{\contextfill{D}{\alhs}}{\alpha} = b_m \ne b_m' = \interpreta{\contextfill{D}{\arhs}}{\alpha}$
  which contradicts that $\cdpth{\aalg}{\arule}$ is the \cdepth{} of $\arule$.
\end{proof}

% As a consequence we can assign weights to $\ssrelabel$ symbols as follows:
% \begin{definition}
%   Let $\atrs$ be a TRS over $\asig$, 
%   let $\quadruple{\aalg}{\sinterpret}{\semlab}{\asigred}$ a \clabeling\ for $\atrs$.
%   We define a function $\srelabelw \funin \reachables{\semlab}{\atrs} \to \{1,\ldots,\cdpth{\aalg}{\atrs}\}$
%   by well-founded induction over $\reachrel$ as follows:
%   \begin{itemize}
%     \item let $\relabelw{\pair{a}{a'}} = 1$, 
%           whenever $\pair{a}{a'} \in \reachables{\semlab}{\atrs}$ is a $\reachrel$ normal form, and
%     \item let $\relabelw{\pair{a}{a'}} = 
%                1 + \max\ \{\relabelw{\pair{b}{b'}} \where \pair{a}{a'} \reachrel \pair{b}{b'}\}$, otherwise.
%   \end{itemize}
% \end{definition}
% 
% \begin{corollary}\label{cor:relabel:bound}
%   Let $\atrs$ be a TRS over $\asig$, 
%   and let $\quadruple{\aalg}{\sinterpret}{\semlab}{\asigred}$ be a \clabeling\ for $\atrs$.
%   For all symbols $\srelabel{a}{a'} \in \sigrelabel{\semlab}{\atrs}$
%   it holds $\relabelw{\pair{a}{a'}} \le \cdpth{\aalg}{\atrs}$,
%   and for every rule:
%   \begin{align*}
%     \funap{\svoodoolabel{\tf}{\alab}}{\vec{x},\relabel{a}{a'}{\dolabelg{t}},\vec{z}}
%     \to \relabel{b}{b'}{\funap{\svoodoolabel{\tf}{\blab}}{\vec{x},\dolabelg{t},\vec{z}}}
%   \end{align*}
%   from $\dynlabprp{\aalg}{\semlab}{\atrs}$, we have $\relabelw{\pair{a}{a'}} > \relabelw{\pair{b}{b'}}$.
% \end{corollary}

\begin{corollary}
  Every $\ssrelabel$ symbol disappears at latest 
  after having applied $\cdpth{\aalg}{\atrs}$ many 
  relabeling rules (to this symbol).
\end{corollary}
\begin{proof}
  For every rule in $\dynlabprp{\aalg}{\semlab}{\atrs}$ of the form:
  \begin{align*}
    \funap{\svoodoolabel{\tf}{\alab}}{\vec{x},\relabel{a}{a'}{\dolabelg{t}},\vec{z}}
    \to \relabel{b}{b'}{\funap{\svoodoolabel{\tf}{\blab}}{\vec{x},\dolabelg{t},\vec{z}}}
  \end{align*}
  we have that $\pair{a}{a'} \reachrel \pair{b}{b'}$.
\end{proof}

For the dynamic context extension, the `intended' terms in $\ter{\labelsig{\asig}}{\setemp}$
are those terms that can be obtained by correctly labeling terms in $\ter{\asig}{\setemp}$.
For the purpose of adapting this definition to dynamic labeling,
we enrich the (unlabeled) signature $\asig$ to $\asig_{+}$:
\begin{align*}
  \asig_{+} &= \asig \join \{ \srelabelone{a} \where \pair{a}{a'}\in\reachables{\semlab}{\atrs} \} 
\end{align*}
and extend the \clabeling{} to $\asig_{+}$ by:
\begin{align*}
  \labelf{\srelabelone{b}}{b'} &= b' & \interpret{\srelabelone{b}} = b
\end{align*}
for all $b,b'\in\aalg$ such that $\pair{b}{c}\in \reachables{\semlab}{\atrs}$ for some $c\in\aalg$.
Then labeled symbols are identified by $\svoodoolabel{(\srelabelone{a})}{a'} = \srelabel{a}{a'}$.
% The following lemma states that if the $\ssrelabel$ symbols are considered as symbols of the unlabeled signature
% then the \clabeling{} can be extended in such a way that every `intended' labeled term 
% can be viewed as the labeling of an unlabeled term.
% Here by an `intended term' we mean a term for which each subterm of the form $\relabel{a}{a'}{t}$
% fulfils that $\pair{a}{a'}\in\reachables{\semlab}{\atrs}$ and $a' = \interpret{\funap{\mit{unlab}}{t}}$
% where $\mit{unlab}$ simply removes the labels of a term. 

We obtain the following lemma: 
\begin{lemma}\label{lem:relabel}
  Let $\atrs$ be a TRS over $\asig$, 
  and let $\triple{\aalg}{\semlab}{\asigred}$ be a \clabeling\ for $\atrs$.
  Whenever we have a ground term $s$ of the form:
  \[
    s = \dolabelg{\contextfill{\acxt}{\f{s_1,\ldots,\relabelone{a}{t},\ldots,s_n}}}
  \]
  with $\pair{a}{a'}\in\reachables{\semlab}{\atrs}$, $a' = \interpret{t}$,
  and where the displayed $\ssrelabel$ symbol is at a $\samumap$-replacing position,
  then one of the following steps applies:
  \begin{align}
    s
    & \to_{\dynlabprp{\aalg}{\semlab}{\atrs},\samumap} 
    \dolabelg{\contextfill{\acxt}{\relabelone{b}{\f{s_1,\ldots,t,\ldots,s_n}}}}
    \\
    s
    & \to_{\dynlabprp{\aalg}{\semlab}{\atrs},\samumap} 
    \dolabelg{\contextfill{\acxt}{\f{s_1,\ldots,t,\ldots,s_n}}}
  \end{align}
  where $b = \interpreta{\f{s_1,\ldots,\cxthole,\ldots,s_n}}{\cxthole\mapsto a}$.
\end{lemma}
\begin{proof}
  Let $b' = \interpret{\f{s_1,\ldots,t,\ldots,s_n}}$. 
  Note that $b' = \interpreta{\f{s_1,\ldots,\cxthole,\ldots,s_n}}{\cxthole\mapsto a'}$.
  Then:
  %Labeling has the following effect:
  \begin{align*}
    \dolabelg{\f{s_1,\ldots,\relabelone{a}{t},\ldots,s_{n}}} 
    & = \funap{\svoodoolabel{\tf}{\alab}}{\dolabelg{s_1},\ldots,\relabel{a}{a'}{\dolabelg{t}},\ldots,\dolabelg{s_{n}}} 
    \\
    \dolabelg{\relabelone{b}{\f{s_1,\ldots,t,\ldots,s_n}}} 
    & = \relabel{b}{b'}{\funap{\svoodoolabel{\tf}{\blab}}{\dolabelg{s_1},\ldots,\dolabelg{t},\ldots,\dolabelg{s_n}}}
  \end{align*}
  where $\alab = \labelf{\tf}{\interpret{s_1},\ldots,a,\ldots,\interpret{s_{n}}}$
  and $\blab = \labelf{\tf}{\interpret{s_1},\ldots,a',\ldots,\interpret{s_{n}}}$.

  By Definition~\ref{def:dynlab} the dynamic labeling $\dynlab{\aalg}{\semlab}{\atrs}$ contains a rule of the form:
  \begin{align*}
  \funap{\svoodoolabel{\tf}{\alab}}{\vec{x},\relabel{a}{a'}{\dolabelg{t}},\vec{z}}
  & \to 
  \contextfill{\acxt}{\funap{\svoodoolabel{\tf}{\blab}}{\vec{x},\dolabelg{t},\vec{z}}}
%   \funap{\svoodoolabel{\tf}{\alab}}{\vec{x},\relabel{a}{a'}{\dolabelg{t}},\vec{z}}
%   & \to 
%   \funap{\svoodoolabel{\tf}{\blab}}{\vec{x},\dolabelg{t},\vec{z}}
  \end{align*}
  with $\acxt = \cxthole$ or $\acxt = \relabel{b}{b'}{\cxthole}$.
  Consequently we have a step of the form:
  \begin{align*}
    \dolabelg{\f{s_1,\ldots,\relabelone{a}{t},\ldots,s_n}} &\to_{\dynlab{\aalg}{\semlab}{\atrs},\samumap} \dolabelg{\contextfill{\bcxt}{\f{s_1,\ldots,t,\ldots,s_n}}}
%    \dolabelg{\f{\ldots,\relabelone{a}{t},\ldots}} &\to \dolabelg{\f{\ldots,t,\ldots}}
  \end{align*}
  with $\bcxt = \cxthole$ or $\bcxt = \relabelone{b}{\cxthole}$.

  Now the claim follows since $s = \contextfill{\dolabel{\acxt}{\cxthole \mapsto b}}{\dolabelg{\f{s_1,\ldots,\relabelone{a}{t},\ldots,s_n}}}$.
\end{proof}

\begin{lemma}\label{lem:dynlab:sound}
  Let $\atrs$ be a TRS over $\asig$, 
  and let $\triple{\aalg}{\semlab}{\asigred}$ be a sound \clabeling\ for $\atrs$.
  Let $s,t \in \term{\asig}{\setemp}$ be ground terms
  such that $s \outred_{\atrs} t$. 
  Then, for some $m \leq \cdpth{\aalg}{\atrs}$\,:
  \[
    \dolabelg{\funap{\topsymb}{s}} 
    \relcomp
      {\to_{\dynlaborg{\aalg}{\semlab}{\atrs},\samumap}}
      {\to^{m}_{\dynlabprp{\aalg}{\semlab}{\atrs},\samumap}}
    \dolabelg{\funap{\topsymb}{t}}
    %%\punc.
  \]
\end{lemma}
\begin{proof}
  Assume $s \outred_{\atrs,p} t$ for some position $p \in \pos{s}$.
  Then there exists a rule $\arule\in\atrs$, 
  a context $\acontext$ with $\symbat{\acontext}{p} = \contexthole$
  and a ground substitution $\asubst\funin{\avars\to\term{\asig}{\setemp}}$ such that 
  $s = \contextfill{\acontext}{\subst{\asubst}{\alhs}}$ 
  and $t = \contextfill{\acontext}{\subst{\asubst}{\arhs}}$.
  \newcommand{\labacxt}[1]{\overline{\acontext}_{#1}}%
  \newcommand{\labasubst}{\overline{\asubst}}%
%   We write % abbreviate 
%   $\alab = \labelf{\topsymb}{\interpret{s}}$,
%   $\labacxt = \funap{\svoodoolabel{\topsymb}{\alab}}{\dolabel{\acontext}{\cxthole\mapsto\interpret{\subst{\asubst}{\ell}}}}$,
%   and $\labasubst = \dolabelg{\asubst}$.
  Let $\labacxt{a} = \dolabel{\funap{\topsymb}{\acontext}}{\cxthole \mapsto a}$ 
  and $\labasubst = \dolabelg{\asubst}$, then by Lemma~\ref{lem:zzz} we obtain:
  \begin{align}
    \dolabelg{\funap{\topsymb}{s}}
    = 
    \dolabelg{\funap{\topsymb}{\contextfill{\acontext}{\subst{\asubst}{\alhs}}}}
    =
    \contextfill{\labacxt{\interpret{\subst{\asubst}{\ell}}}}{\dolabelg{\subst{\asubst}{\alhs}}}
    =
    \contextfill{\labacxt{\interpret{\subst{\asubst}{\ell}}}}{\subst{\,\labasubst}{\dolabel{\alhs}{\interpret{\asubst}}}}
    \label{eq:simpl:s}
  \end{align}
  \newcommand{\refeqsimplt}{\ensuremath{\ref{eq:simpl:t}}}%{\ensuremath{\ref{eq:simpl:s}'}}%
  and, likewise:
  \begin{align}
  \dolabelg{\funap{\topsymb}{t}} = \contextfill{\labacxt{\interpret{\subst{\asubst}{\arhs}}}}{\subst{\,\labasubst}{\dolabel{\arhs}{\interpret{\asubst}}}}
  \label{eq:simpl:t}
  %\tag{\ref{eq:simpl:s}$'$}
  \end{align}

  By definition of dynamic labeling, one of the following rules is in $\dynlaborg{\aalg}{\semlab}{\atrs}$: 
  %i.e., the first partition of the dynamic labeling: % of $\atrs$\,:
  \begin{align}
    \dolabel{\alhs}{\interpret{\asubst}} & \to \dolabel{\arhs}{\interpret{\asubst}} 
    & \text{if $\interpret{\subst{\asubst}{\alhs}} \eq \interpret{\subst{\asubst}{\arhs}}$}
    \label{dynlabrule:1}
    \\
    \dolabel{\alhs}{\interpret{\asubst}} 
    & \to \relabel{\interpret{\subst{\asubst}{\alhs}}}{\interpret{\subst{\asubst}{\arhs}}}{\dolabel{\arhs}{\interpret{\asubst}}}
    & \text{if $\interpret{\subst{\asubst}{\alhs}} \neq \interpret{\subst{\asubst}{\arhs}}$}
    \label{dynlabrule:2}
  \end{align}
  (Note that 
  $\interpreta{\alhs}{\interpret{\asubst}} = \interpret{\subst{\asubst}{\alhs}}$, 
  and $\interpreta{\arhs}{\interpret{\asubst}} = \interpret{\subst{\asubst}{\arhs}}$ by Lemma~\ref{lem:xyz}.)
  
  In case $\interpret{\subst{\asubst}{\alhs}} \eq \interpret{\subst{\asubst}{\arhs}}$, 
  no relabeling is needed and we take $m = 0$:
  \begin{align*}
    \dolabelg{\funap{\topsymb}{s}}
    & \stackrel{\eqref{eq:simpl:s}}{=} 
    \contextfill{\labacxt{\interpret{\subst{\asubst}{\alhs}}}}{\subst{\,\labasubst}{\dolabel{\alhs}{\interpret{\asubst}}}}
    \stackrel{\eqref{dynlabrule:1}}{\to_{\dynlaborg{\aalg}{\semlab}{\atrs},\samumap}}
    \contextfill{\labacxt{\interpret{\subst{\asubst}{\alhs}}}}{\subst{\,\labasubst}{\dolabel{\arhs}{\interpret{\asubst}}}}
    \stackrel{\refeqsimplt}{=} 
    \dolabelg{\funap{\topsymb}{t}}
  \end{align*}
  
  If $\interpret{\subst{\asubst}{\alhs}} \neq \interpret{\subst{\asubst}{\arhs}}$, 
  we get:
  \begin{align*}
    \dolabelg{\funap{\topsymb}{s}}
    \stackrel{\eqref{eq:simpl:s}}{=} &\ 
    \contextfill{\labacxt{\interpret{\subst{\asubst}{\alhs}}}}{\subst{\,\labasubst}{\dolabel{\alhs}{\interpret{\asubst}}}}
    \\
    \stackrel{\eqref{dynlabrule:2}}{\to_{\dynlaborg{\aalg}{\semlab}{\atrs},\samumap}} &\ 
    \contextfill{\labacxt{\interpret{\subst{\asubst}{\alhs}}}}{\subst{\,\labasubst}{\relabel{\interpret{\subst{\asubst}{\alhs}}}{\interpret{\subst{\asubst}{\arhs}}}{\dolabel{\arhs}{\interpret{\asubst}}}}}
    = \dolabelg{\funap{\topsymb}{\contextfill{\acxt}{\relabelone{\interpret{\subst{\asubst}{\alhs}}}{\subst{\asubst}{\arhs}}}}}
  \end{align*} 
  By Lemma~\ref{lem:relabel} the $\ssrelabel$ symbol can `walk' upward until it disappears,
  and at the latest it vanishes when it meets $\topsymb$. 
  Hence we have:
  \begin{align*}
    \dolabelg{\funap{\topsymb}{s}}
    \to_{\dynlaborg{\aalg}{\semlab}{\atrs},\samumap}
    \dolabelg{\funap{\topsymb}{\contextfill{\acxt}{\relabelone{\interpret{\subst{\asubst}{\alhs}}}{\subst{\asubst}{\arhs}}}}}
    \to_{\dynlabprp{\aalg}{\semlab}{\atrs},\samumap}^m
    \dolabelg{\funap{\topsymb}{\contextfill{\acxt}{\subst{\asubst}{\arhs}}}}
  \end{align*}
  for some $m \le \cdpth{\aalg}{\atrs}$ by Lemma~\ref{lem:relabel:bound}.
\end{proof}

\begin{theorem}\label{thm:dynlab:sound}
  Let $\atrs$ be a TRS over $\asig$, 
  and $\triple{\aalg}{\semlab}{\asigred}$ a sound \clabeling\ for $\atrs$.
  Then $\atrs$\ is outermost ground terminating if $\dynlab{\aalg}{\semlab}{\atrs}$ is terminating.
\end{theorem}
\proof
  Assume that $\atrs$\, admits an infinite outermost rewrite sequence of ground terms:
  \[
    t_1 \outred_{\atrs} t_2 \outred_{\atrs} t_3 \outred_{\atrs} \ldots
  \]
  Then from Lemma~\ref{lem:dynlab:sound} it follows that 
  $\dynlab{\aalg}{\semlab}{\atrs}$ admits an infinite rewrite sequence: % of ground terms:
  $$
    \dolabelg{\funap{\topsymb}{t_1}} 
    \to_{\dynlab{\aalg}{\semlab}{\atrs},\samumap}^{+}
    \dolabelg{\funap{\topsymb}{t_2}} 
    \to_{\dynlab{\aalg}{\semlab}{\atrs},\samumap}^{+}
    \dolabelg{\funap{\topsymb}{t_3}} 
    \to_{\dynlab{\aalg}{\semlab}{\atrs},\samumap}^{+}
    \ldots\eqno{\qEd}
  $$
%\end{proof}

We note that Theorem~\ref{thm:dynlab:sound} can be strengthened
by weakening the termination of $\dynlab{\aalg}{\semlab}{\atrs}$
to local termination of $\dynlab{\aalg}{\semlab}{\atrs}$ on 
the set of correctly labeled ground terms without $\ssrelabel$ symbols.
Let us denote this set by $\dolabelg{\ter{\asig}{\setemp}}$:
%In what follows, we use $\dolabelg{\ter{\asig}{\setemp}}$ 
%to denote the set of 
%correctly labeled ground terms which do not contain $\ssrelabel$ symbols:
\[
  \dolabelg{\ter{\asig}{\setemp}}
  =
  \{\dolabelg{t} \where t \in \ter{\asig}{\setemp} \}
\]

Theorem~\ref{thm:dynlab:complete} below states
that dynamic labeling is complete with respect to 
local termination on $\dolabelg{\ter{\asig}{\setemp}}$.
More precisely, outermost ground termination of $\atrs$ implies termination of 
$\dynlab{\aalg}{\semlab}{\atrs}$ on $\dolabelg{\ter{\asig}{\setemp}}$.
The following example helps to understand the proof of that theorem;
it illustrates that even when starting from terms in
$\dolabelg{\ter{\asig}{\setemp}}$, not every rewrite step in 
$\dynlab{\aalg}{\semlab}{\atrs}$ corresponds to an outermost step in $\atrs$.

\begin{example}\label{ex:dyncomplete}
  Let $\atrs$ consist of the following rules:
  \begin{align*}
    \fb{\tb}{x} &\to \ta &
    \fb{x}{\tb} &\to \ta \\
    \fb{\tb}{\tb} &\to \fb{\ta}{\ta} &
    \ta &\to \tb
  \end{align*}
  Moreover, let $\aalg = \{\bot,b\}$ with
  $\interpret{\ta} = \bot$, $\interpret{\tb} = b$,
  and $\bfunap{\interpret{\tf}}{x}{y} = \bot$ for all $x,y \in \aalg$.
  Labeling symbols with the value of their arguments, 
  we obtain for $\dynlaborg{\aalg}{\semlab}{\atrs}$:
  \begin{align*}
    \bfunap{\svoodoolabel{\tf}{b,\bot}}{\tb}{x} &\to \ta &
    \bfunap{\svoodoolabel{\tf}{\bot,b}}{x}{\tb} &\to \ta \\
    \bfunap{\svoodoolabel{\tf}{b,b}}{\tb}{x} &\to \ta &
    \bfunap{\svoodoolabel{\tf}{b,b}}{x}{\tb} &\to \ta \\
    \bfunap{\svoodoolabel{\tf}{b,b}}{\tb}{\tb} &\to \bfunap{\svoodoolabel{\tf}{\bot,\bot}}{\ta}{\ta} &
    \ta &\to \relabel{\bot}{b}{\tb}
  \end{align*}
  and for $\dynlabprp{\aalg}{\semlab}{\atrs}$:
  \begin{align*}
    \bfunap{\svoodoolabel{\tf}{\bot,\bot}}{\relabel{\bot}{b}{x}}{y} &\to \bfunap{\svoodoolabel{\tf}{b,\bot}}{x}{y} 
    & 
    \bfunap{\svoodoolabel{\tf}{\bot,\bot}}{x}{\relabel{\bot}{b}{y}} &\to \bfunap{\svoodoolabel{\tf}{\bot,b}}{x}{y}
  \end{align*}
  where $\sigredex{\asig} = \{\ta, \svoodoolabel{\tf}{b,\bot}, \svoodoolabel{\tf}{\bot,b}, \svoodoolabel{\tf}{b,b}\}$.
  Then we obtain the following rewrite sequence in $\dynlab{\aalg}{\semlab}{\atrs}$:
  \begin{align*}
    \dolabelg{\fb{\ta}{\ta}} =&\ \bfunap{\svoodoolabel{\tf}{\bot,\bot}}{\ta}{\ta}\\
    \to_{\dynlaborg{\aalg}{\semlab}{\atrs},\samumap} &\ 
      \bfunap{\svoodoolabel{\tf}{\bot,\bot}}{\relabel{\bot}{b}{\tb}}{\ta}\\
    \to_{\dynlaborg{\aalg}{\semlab}{\atrs},\samumap}&\ 
      \bfunap{\svoodoolabel{\tf}{\bot,\bot}}{\relabel{\bot}{b}{\tb}}{\relabel{\bot}{b}{\tb}}\\
    \to_{\dynlabprp{\aalg}{\semlab}{\atrs},\samumap}&\ 
      \bfunap{\svoodoolabel{\tf}{b,\bot}}{\tb}{\relabel{\bot}{b}{\tb}}
  \end{align*}
  The second step in this rewrite sequence
  does not correspond to an outermost step.
  Nevertheless, Theorem~\ref{thm:dynlab:complete}
  states that such `illegal' steps do not harm completeness of the transformation.
  The reason is that if the relabeling rules create a redex above
  some $\ssrelabel$ symbol, then this $\ssrelabel$ symbol
  is prevented from further propagating its information upward
  (until it becomes $\samumap$-replacing again).
  The crucial point is that above $\ssrelabel$ symbols
  the labels are unchanged, thus as if the step would not have taken place.
  Moreover, it is essential that $\dynlab{\aalg}{\semlab}{\atrs}$ 
  prohibits $\ssrelabel$ to propagate over symbols from $\sigredex{\asig}$.
  For instance, in the above example $\dynlab{\aalg}{\semlab}{\atrs}$
  does not contain a rule of the form:
  \begin{align}
    \bfunap{\svoodoolabel{\tf}{b,\bot}}{x}{\relabel{\bot}{b}{y}} &\to \bfunap{\svoodoolabel{\tf}{b,b}}{x}{y}
    \label{illrelabel}
  \end{align}
  This rule would cause non-termination: 
  \begin{align*}
    \bfunap{\svoodoolabel{\tf}{\bot,\bot}}{\ta}{\ta} 
    & \mured^2 
    \bfunap{\svoodoolabel{\tf}{\bot,\bot}}{\relabel{\bot}{b}{\tb}}{\relabel{\bot}{b}{\tb}} \\
    & \mured
    \bfunap{\svoodoolabel{\tf}{b,\bot}}{\tb}{\relabel{\bot}{b}{\tb}}
    \\
    & \stackrel{\text{\eqref{illrelabel}}}{\mured}
    \bfunap{\svoodoolabel{\tf}{b,b}}{\tb}{\tb}
    \\
    & \mured
    \bfunap{\svoodoolabel{\tf}{\bot,\bot}}{\ta}{\ta}
  \end{align*}
\end{example}

\begin{theorem}\label{thm:dynlab:complete}
  Let $\atrs$ be a left-linear TRS over $\asig$, 
  and $\triple{\aalg}{\semlab}{\asigred}$ a complete, maximal, and core \clabeling\ for $\atrs$.
  Then $\dynlab{\aalg}{\semlab}{\atrs}$ is terminating on the set of terms $\dolabelg{\ter{\asig}{\setemp}}$
  if $\atrs$ is outermost ground terminating.
\end{theorem}
\begin{proof}
  \newcommand{\sto}{\hookrightarrow}%
  \newcommand{\nto}{\stackrel{\makebox(0,0){{\scriptsize$\neg\samumap$}}}{\longrightarrow}}%
  Define $T = \dolabelg{\ter{\asig}{\setemp}}$,
  and ${\sto} = (\relcomp{ \to_{\dynlaborg{\aalg}{\semlab}{\atrs},\samumap} }{ \to^*_{\dynlabprp{\aalg}{\semlab}{\atrs},\samumap} }) \cap (T \times T)$.
  Note that the relation $\sto$ is restricted to terms which contain no $\ssrelabel$ symbols.
  Hence always the maximal number of relabeling rules is applied.
  It is clear that each ${\sto}$ rewrite step corresponds to an outermost rewrite step in the original TRS $\atrs$.
  Therefore it suffices to show that any infinite rewrite sequence
  $t = t_0 \muredr{\dynlab{\aalg}{\semlab}{\atrs}} t_1 \muredr{\dynlab{\aalg}{\semlab}{\atrs}} \ldots$
  gives rise to an infinite rewrite sequence $t = s_0 \sto s_1 \sto\ldots$.
  We prove the claim by a kind of standardization of reductions.
  We first classify the rules from $\dynlab{\aalg}{\semlab}{\atrs}$:
  \begin{align*}
  \dolabel{\ell}{\alpha} &\to \dolabel{r}{\alpha}
  \tag{\ensuremath{c1}}\label{dyn1}\\
  \dolabel{\ell}{\alpha} &\to \relabel{\interpreta{\ell}{\alpha}}{\interpreta{r}{\alpha}}{\dolabel{r}{\alpha}}
        \tag{\ensuremath{c2}}\label{dyn2}\\
  \funap{\svoodoolabel{\tf}{\funap{\slabelf{\tf}}{\vec{a},b,\vec{c}}}}{\vec{x},\relabel{b}{b'}{y},\vec{z}}
  &\to
  \funap{\svoodoolabel{\tf}{\funap{\slabelf{\tf}}{\vec{a},b',\vec{c}}}}{\vec{x},y,\vec{z}}
  \tag{\ensuremath{c3}}\label{dyn3}\\
  \funap{\svoodoolabel{\tf}{\funap{\slabelf{\tf}}{\vec{a},b,\vec{c}}}}{\vec{x},\relabel{b}{b'}{y},\vec{z}}
  &\to
  \relabel{d}{d'}{\funap{\svoodoolabel{\tf}{\funap{\slabelf{\tf}}{\vec{a},b',\vec{c}}}}{\vec{x},y,\vec{z}}}
        \tag{\ensuremath{c4}}\label{dyn4}
  \end{align*}
  For $i = 0,1,\ldots$, we analyse the steps $t_i \to_{\dynlab{\aalg}{\semlab}{\atrs},\samumap} t_{i+1}$
  and construct $s_0 \sto s_1 \sto \ldots \sto s_j$
  in such a way that $s_j \nto_{\ref{dyn2},\ref{dyn4}}^* t_{i+1}$
  where we use $\nto_{\ref{dyn2},\ref{dyn4}}$ to denote standard term rewriting 
  ignoring the replacement map $\samumap$, and using rules from \eqref{dyn2} and \eqref{dyn4} only.
  Observe that then the maximal prefix $C_{i+1}$ of $t_{i+1}$ 
  not containing $\ssrelabel$ symbols is also a prefix of $s_j$
  (since everything changed by \eqref{dyn2} and \eqref{dyn4} is `hidden' inside a $\ssrelabel$ symbol).
  We begin with $t = s_0$, and $i = j = 0$.
  For $i = 0,1,\ldots$, we consider the step $\tau_i : t_i \to_{\dynlab{\aalg}{\semlab}{\atrs},\samumap} t_{i+1}$.

  If $\tau_i$ is a step with respect to a rule from:
  \begin{enumerate}[$-$]

  \item
  %If $\tau_i$ is a step with respect to rules 
  \eqref{dyn2} or \eqref{dyn4}, then
  we append $\tau_i$ to the rewrite sequence 
  $s_j \nto_{\ref{dyn2},\ref{dyn4}}^* t_i$
  yielding 
  $s_j \nto_{\ref{dyn2},\ref{dyn4}}^* t_{i+1}$.
  Note that this leaves the $\sto$-rewrite sequence $s_0 \sto^{\ast} s_j$ untouched.

  \item
    %If $\tau_i$ is a step with respect to 
  \eqref{dyn1},
  then the pattern of $\tau_i$ lies entirely in $C_i$ which is also prefix of $s_j$.
  Then we append $\tau_i$ to $s_0 \sto^{\ast} s_j$ (using left-linearity of $\atrs$) 
  yielding $s_0 \sto^{\ast} s_j \mstackrel{\tau_i}{\sto} s_{j+1}$.
  We have $s_{j+1} \nto_{\ref{dyn2},\ref{dyn4}}^* t_{i+1}$
  by orthogonal projection of the steps $s_j \nto_{\ref{dyn2},\ref{dyn4}}^* t_i$
  over $s_j \mstackrel{\tau_i}{\sto} s_{j+1}$ (all steps in $s_j \nto_{\ref{dyn2},\ref{dyn4}}^* t_i$ are below the prefix $C_i$).
  
  \item
    %If $\tau_i$ is a step with respect to 
  \eqref{dyn3}, then a $\ssrelabel$ symbol `disappears'.
  We can trace this symbol back to a sequence of steps
  $\sigma_i \funin s_j \relcomp{ \nto_{\ref{dyn2}} }{ \nto_{\ref{dyn4}}^+ } s_j'$,
  that is, it must have been created in $s_j$ by a \eqref{dyn2} step,
  followed by a number of \eqref{dyn4} steps.
  We combine $\sigma_i$ and $\tau_i$ to a $\sto$ step,
  yielding $s_0 \sto^* s_j \mstackrel{\relcomp{\sigma_i}{\tau_i}}{\sto} s_{j+1}$.
  Then $s_{j+1} \nto_{\ref{dyn2},\ref{dyn4}}^* t_{i+1}$
  as the remaining steps from
  $s_j \nto_{\ref{dyn2},\ref{dyn4}}^* t_i$ are not harmed by the permutation (performing $\sigma_i$ first).

  \end{enumerate}

  \noindent
  It remains to be shown that the constructed sequence $s_0 \sto s_1 \sto s_2 \sto \ldots$ is infinite.
  This follows from the fact that an infinite number of steps in
  $t_0 \to_{\dynlab{\aalg}{\semlab}{\atrs},\samumap} t_1 \to_{\dynlab{\aalg}{\semlab}{\atrs},\samumap} \ldots$
  must be of type \eqref{dyn1} or \eqref{dyn3}.
  This is a direct consequence of the fact that
  $\to_{\ref{dyn2},\ref{dyn4}}$ is terminating
  (with every step the prefix in which rewriting is allowed gets smaller).
\end{proof}

The following example demonstrates why the completeness result for dynamic labeling (Theorem~\ref{thm:dynlab:complete}) 
is restricted to the set $\dolabelg{\ter{\asig}{\setemp}}$ of correctly labeled terms 
which do not contain $\ssrelabel$ symbols.
The point is that, although the original TRS is outermost terminating
the transformed system may in general be non-terminating 
due to the existence of `non-reachable' terms.
\begin{example}\label{ex:dynlab:noncomplete}
  We consider the following term rewriting system $\atrs$:
  \begin{align*}
    \ta &\to \tb
    &
    \fb{\tb}{y} &\to \tb
    &
    \fb{\fun{c}}{y} &\to \h{\fb{y}{y}} 
    \\
    \h{\fb{x}{\tb}} &\to \tb
    &  
    \h{\fb{x}{\fun{c}}} &\to \tb
  \end{align*}
  We explain why this TRS is outermost ground terminating.
  Without the rule $\rho : \fb{\fun{c}}{y} \to \h{\fb{y}{y}}$, 
  the system would even be terminating.
  Now note that the rule $\rho$ can only be applied once 
  to each occurrence of $\fb{\fun{c}}{\cxthole}$
  since $\h{\fb{t}{t}} \to^* \h{\fb{\fun{c}}{t'}}$ implies that $t = \fun{c}$,
  and then the rule $\h{\fb{x}{\fun{c}}} \to \tb$ has priority by the strategy of outermost rewriting.

  We define a maximal, complete, core \clabeling{} 
  $\triple{\aalg}{\asemlab}{\asigred}$ for $\atrs$
  (isomorphic to the result of the construction given in the next section) % Definition~\ref{def:talgebra})
  where the algebra $\aalg = \{ \mit{bc}, \mit{fbc}, \bot \}$ 
  with the interpretation function defined by:
  \begin{align*}
  \bfunap{\interpret{\tf}}{x}{\mit{bc}} &= \mit{fbc}
  &  
  \interpret{\tb} & = \interpret{\fun{c}} = \mit{bc}
  \\
  \bfunap{\interpret{\tf}}{x}{y} &= \bot
  &
  \interpret{\ta} & = \funap{\interpret{\sh}}{x} = \bot 
  \end{align*}
  for all $x,y \in \aalg$ with $y \ne \mit{bc}$,
  and with $\asigred = \{ \ta , \svoodoolabel{\tf}{\mit{bc}} , \svoodoolabel{\sh}{\mit{fbc}} \}$.

  The dynamic labeling $\dynlab{\aalg}{\semlab}{\atrs}$ of $\atrs$ with respect to this \clabeling{} then includes the rules:
  \begin{align*}
    \ta &\to \relabel{\bot}{\mit{bc}}{\tb} \\
    \bfunap{\svoodoolabel{\tf}{\mit{bc},\bot}}{\fun{c}}{y} &\to \funap{\svoodoolabel{\sh}{\bot}}{\bfunap{\svoodoolabel{\tf}{\bot,\bot}}{y}{y}} \\
    \bfunap{\svoodoolabel{\tf}{\bot,\bot}}{\relabel{\bot}{\mit{bc}}{x}}{y} &\to \bfunap{\svoodoolabel{\tf}{\mit{bc},\bot}}{x}{y}
  \end{align*}
  Now the context-sensitive TRS $\dynlab{\aalg}{\semlab}{\atrs}$ 
  admits the following infinite rewrite sequence:
  \begin{align*}
    \bfunap{\svoodoolabel{\tf}{\mit{bc},\bot}}{\fun{c}}{\relabel{\bot}{\mit{bc}}{\fun{c}}} 
    &\to_{\dynlab{\aalg}{\semlab}{\atrs},\samumap}
    \funap{\svoodoolabel{\sh}{\bot}}{\bfunap{\svoodoolabel{\tf}{\bot,\bot}}{\relabel{\bot}{\mit{bc}}{\fun{c}}}{\relabel{\bot}{\mit{bc}}{\fun{c}}}}\\
    &\to_{\dynlab{\aalg}{\semlab}{\atrs},\samumap}
    \funap{\svoodoolabel{\sh}{\bot}}{\bfunap{\svoodoolabel{\tf}{\mit{bc},\bot}}{\fun{c}}{\relabel{\bot}{\mit{bc}}{\fun{c}}}}\\
    &\to_{\dynlab{\aalg}{\semlab}{\atrs},\samumap}
    \ldots
  \end{align*}
  %
  %It is important to observe that the infinite rewrite sequence is caused 
  Observe that this anomaly is caused by the subterm $\relabel{\bot}{\mit{bc}}{c}$,
  which is not reachable from any term in $\dolabelg{\ter{\asig}{\setemp}}$.
\end{example}

\begin{remark}
  Theorem~\ref{thm:dynlab:complete} 
  states completeness of dynamic labeling with respect to 
  local termination on the set of terms $\dolabelg{\ter{\asig}{\setemp}}$.
  We briefly indicate how the theorem can be generalized 
  to termination on $\ter{\labelsig{\asig}}{\setemp}$
  by altering the definition of $\dynlab{\aalg}{\semlab}{\atrs}$.
  Note that
  $\dolabelg{\ter{\asig}{\setemp}} \subsetneq \ter{\labelsig{\asig}}{\setemp}$.
  In particular,
  the set $\ter{\labelsig{\asig}}{\setemp}$ includes terms that are not correctly labeled.
  The necessary modification of the definition of dynamic labeling %$\dynlab{\aalg}{\semlab}{\atrs}$
  concerns the elimination of collapsing rules $\alhs \to x$.
  This can be achieved by wrapping the right-hand side into $\relabel{a}{a}{\cxthole}$
  even when the interpretations of the left and right-hand side are equal.
  Additionaly, we let the symbols $\srelabel{a}{a}$ disappear after one relabeling step.
  By an application of Theorem~\ref{thm:adapt:ohlebusch} %%\cite[Proposition~5.5.24]{ohle:02}
  it then follows that 
  termination on $\dolabelg{\ter{\asig}{\setemp}}$ 
  coincides with termination on $\ter{\labelsig{\asig}}{\setemp}$.
\end{remark}

\section{Constructing Suitable Algebras}\label{sec:constructing}
\newcommand{\cutterms}{\sub{\sterm}{\bot}}%

We construct \cmodel{s} which are able to
recognize redex positions with respect to left-linear rules.
The construction of \cmodel{s} is similar to the construction 
of a deterministic tree automaton (DTA, \cite{tata:07}) for recognizing left-linear redexes~\cite{como:00}.
A DTA is a $\asig$-algebra $\pair{\aset}{\sinterpret}$
with a distinguished set $\aset_F \subseteq \aset$ of final states.
A term $t$ is accepted by %(is in the language of) 
the automaton whenever $\interpret{t} \in \aset_F$.
A difference with the construction of a DTA is that for the construction of a \cmodel{}
we do not distinguish final and non-final states,
but instead have a family of functions $\sisredex{\tf} \funin \aset^{\arity{\tf}} \to \bool$
for indicating the presence of a redex.
\begin{definition}[Redex-algebra]\label{def:redex:algebra}
  A \emph{redex-algebra $\pair{\aalg}{\sisredex{}}$}
  consists of a $\asig$-algebra $\aalg$
  together with a family $\afam{\sisredex{\tf}}{\tf \in \asig}$
  of functions $\sisredex{\tf} \funin \aalg^{\arity{\tf}} \to \bool$.
  The \emph{language of $\aalg$} is the set:
  \begin{align*}
    \lang{\aalg} 
    = 
    \{ 
      \f{t_1,\ldots,t_n} \in \term{\asig}{\setemp} 
      \where 
      \isredex{\tf}{\interpret{t_1},\ldots,\interpret{t_n}} = \true
    \}
  \end{align*}
  Let $\atrs$ be a TRS. A redex-algebra $\aalg$ is called \emph{sound for $\atrs$} 
  whenever $t\in\lang{\aalg}$ implies that $t$ is a redex,
  %A redex-algebra 
  and $\aalg$\,  is called \emph{complete for $\atrs$} 
  if for all redexes $t\in\ter{\asig}{\setemp}$ we have $t\in\lang{\aalg}$.
\end{definition}
Intuitively, a redex-algebra needs to `remember' only the subterms $t_1,\ldots,t_n$
and not the term $\f{t_1,\ldots,t_n}$ itself.
To see this, consider the one-rule system: 
\[
  \f{\g{x}} \to \ta
\]
A tree automaton which can recognize redex positions for this TRS needs at least three states: 
one for indicating a redex $\f{\g{\ldots}}$,
one for $\g{\ldots}$,
and one garbage state.
For redex-algebras two states suffice:
one state for $\g{\ldots}$ and one for garbage,
and then we use $\isredex{\tf}{\g{\ldots}} = \true$ and $\false$, otherwise.

We now describe a syntactical construction of redex-algebras.
%The basic idea is to build $\asig$-algebras that `remember' proper subterms of left-hand sides.
%On the basis of this interpretation, the $\sisredex{}$ functions decide whether a redex is present.
The algebras we construct are \cmodel{s}:
the \cdepth{} of a rule $\arule$ is the maximal pattern depth of a left-hand side (minus $1$),
since for recognizing the subterms of left-hand sides we may `forget' all information
that lies below the patterns.
We first define some auxiliary functions, introduced with different notations in \cite{{klop:midd:91},{huet:levy:91}}. 

\begin{definition}
  Let $\asig$ be a signature and $\avars$ a set of variables.
  We let $\bot$ be a new symbol, $\bot \not\in \asig$, 
  and we define $\cutterms = \term{\asig \cup \{\bot\}}{\setemp}$.
  The function $\scut \funin \term{\asig}{\avars} \to \cutterms$ 
  is defined such that $\cut{t}$ is the result of 
  replacing all variables in a term $t\in\term{\asig}{\avars}$ by %the symbol 
  $\bot$:
   \begin{align*}
   \cut{x} &= \bot & %\text{ for $x \in \avars$}\\
   \cut{\f{t_1,\ldots,t_n}} &= \f{\cut{t_1},\ldots,\cut{t_n}}
   \end{align*}
  We define the function $\smatch \funin \cutterms \times \cutterms \to \bool$
  such that $\match{s}{t} = \true$ if $s$ can be obtained from $t$ by
  replacing subterms of $t$ by $\bot$, and we let $\match{s}{t} = \false$, otherwise.
  Further, we let $\merge{s}{t}$ be the `most general common instance' of $s$ and~$t$, 
  that is, $\smerge \funin \cutterms \times \cutterms \pto \cutterms$ is the partial function defined by:
  \begin{align*}
  \merge{\bot}{t} 
  & =
  \merge{t}{\bot}
  = t
  \\
  \merge{\f{s_1,\ldots,s_n}}{\f{t_1,\ldots,t_n}} 
  & = \f{\merge{s_1}{t_1},\ldots,\merge{s_n}{t_n}}
  \end{align*}
  Hence $\merge{s}{t}$ is undefined whenever there exists a position $\apos \in \pos{s}$ 
  such that $\symbat{s}{\apos} \in \asig$, $\symbat{t}{\apos} \in \asig$, 
  and $\symbat{s}{\apos} \ne \symbat{t}{\apos}$.
  For a term $s\in\cutterms$ and a set $T\subseteq\cutterms$
  we define the term $\shrink{s}{T}$ as the largest $t \in T$ (with respect to the number of symbols)
  such that $\match{t}{s} = \true$.
  Note that $\shrink{s}{T}$ is well-defined whenever $T$ is closed under $\smerge$ and $\bot \in T$:
  whenever two terms $t_1 \neq t_2$ of equal size match $s$
  then $\merge{t_1}{t_2}$ is larger and matches $s$.
\end{definition}

\newcommand{\sralgconstr}{F}%
\newcommand{\ralgconstr}{\funap{\sralgconstr}}%
\begin{definition}[Construction of redex-algebra]\label{def:talgebra}
  We define a mapping $\sralgconstr$ which
  constructs a redex-algebra for a given TRS~$\atrs$.
  Let $\ralgconstr{\atrs} = \pair{\aalg}{\sisredex{}}$ 
  where $\aalg$ is the smallest set such that:
  \begin{enumerate}[$-$]
  
    \item
      $\bot \in \aalg$\,,
  
    \item
      $t \in \aalg$ for every proper subterm $t$ of $\cut{\ell}$ 
      with $\ell$ a left-hand side of a rule in $\atrs$\,,
    \item
      $\merge{s}{t} \in \aalg$ whenever $s,t \in \aalg$ and $\merge{s}{t}$ is defined.
  
  \end{enumerate}
  The interpretation function $\sinterpret$ of $\aalg$, 
  and the functions $\ssisredex{}$ are defined by:
  \begin{align*}
    \funap{\interpret{\tf}}{t_1,\ldots,t_n}
      & = \shrink{\f{t_1,\ldots,t_n}}{\aalg}
    \\
    \isredex{\tf}{t_1,\ldots,t_n} 
      & =
      \begin{cases}
        \true  & \text{if $\match{\cut{\ell}}{\f{t_1,\ldots,t_n}} = \true$ for some $\ell \to r \in \atrs$} \\
        \false & \text{otherwise}
     \end{cases}
  \end{align*}
  for all function symbols $\tf\in\asig$ with arity $n$, and terms $t_1,\ldots,t_n$.

  The core of the algebra $\ralgconstr{\atrs}$, which we denote by $\ralgconstr{\atrs}_c$,
  is called the \emph{\cralg{} redex-algebra for $\atrs$}. 
  Furthermore, let $\btrs\subseteq\atrs$ be the set consisting of all left-linear rules of $\atrs$.
  Then $\ralgconstr{\btrs}_c$ is called the \emph{\sralg{} redex-algebra for $\atrs$}.
\end{definition}

Of course, if $\atrs$ is a left-linear TRS, 
then a \sralg{} redex-algebra for $\atrs$ also is a \cralg{} redex-algebra for $\atrs$.
Moreover, if $\atrs$ is a quasi-left-linear TRS, then
the minimized \sralg{} and \cralg{} redex-algebras for $\atrs$ are isomorphic.
Minimization of redex-algebras is introduced in the next section.
We now consider two examples which illustrate that $\ralgconstr{\atrs}_c$
indeed can be a proper subalgebra of $\ralgconstr{\atrs}$.

\begin{example}
  We consider the term rewriting system $\atrs$ which consists of the rules:
  \begin{align*}
    \funap{\ta}{x} &\to \fb{\funap{\ta}{x}}{x} &
    \fb{x}{\funap{\ta}{y}} &\to \tb &
    \fb{x}{\fb{y}{z}} &\to \tb &
    \fb{x}{\tb} &\to \tb
  \end{align*}
  This TRS is outermost ground terminating, but it is not outermost terminating.
  %According to Definition~\ref{def:talgebra} 
  We construct the redex-algebra $\ralgconstr{\atrs} = \pair{\aalg}{\sisredex{}}$
  %(which coincides with the \cralg{} redex-algebra for $\atrs$ by left-linearity of $\atrs$).
  where $\aalg = \{ \bot, \funap{\ta}{\bot}, \fb{\bot}{\bot}, \tb \}$
  with $\funap{\interpret{\ta}}{x} = \funap{\ta}{\bot}$,
  $\bfunap{\interpret{\tf}}{x}{y} = \fb{\bot}{\bot}$
  and $\interpret{\tb} = \tb$ for all $x,y \in \aalg$.
  But note that $\bot$ is not part of the core, and 
  hence the left-linear (full) redex-algebra 
  $\ralgconstr{\atrs}_c$
  contains only the elements $\{\funap{\ta}{\bot}, \fb{\bot}{\bot}, \tb\}$.
\end{example}
\begin{example}
  Consider the term rewriting system $\atrs$ consisting of the rules:
  \begin{align*}
    \h{\h{\h{x}}} &\to \ta &
    \h{\h{\ta}} &\to \h{\h{\h{\h{\ta}}}}
  \end{align*}  
  The domain of the redex-algebra $\ralgconstr{\atrs}$ is 
  $\{\h{\h{\bot}}, \h{\bot}, \bot, \h{\ta}, \ta\}$
  with the interpretation of the symbols defined by:
  \begin{align*}
    \interpret{\ta} & = \ta
    & \funap{\interpret{\sh}}{\ta} & = \h{\ta}
    & \funap{\interpret{\sh}}{\h{\ta}} & = \h{\h{\bot}}
    & \funap{\interpret{\sh}}{\h{\h{\bot}}} & = \h{\h{\bot}}
  \end{align*}
  The values $\bot$ and $\h{\bot}$ are not part of the core,
  and hence the domain of the left-linear (full) redex-algebra 
  $\ralgconstr{\atrs}_c$ is $\{\h{\h{\bot}}, \h{\ta}, \ta\}$
  with $\isredex{\sh}{\h{\ta}} = \isredex{\sh}{\h{\h{\bot}}} = \true$, 
  and $\false$ otherwise.
\end{example}

The next example illustrates the use of the function $\smerge$ in the construction of a redex-algebra.
\begin{example}\label{ex:merge}
  We construct the \sralg{} (\cralg) redex-algebra for the TRS:
  \begin{align*}
    \fb{x}{y} &\to \funap{\ta}{\fb{\funap{\fun{c}}{x}}{y}} &
    \funap{\ta}{\fb{\funap{\fun{c}}{\funap{\fun{c}}{x}}}{y}} &\to \fun{e} \\
    \fb{x}{y} &\to \funap{\tb}{\fb{x}{\funap{\fun{c}}{y}}} &
    \funap{\tb}{\fb{x}{\funap{\fun{c}}{\funap{\fun{c}}{y}}}} &\to \fun{e}
  \end{align*}
  The subterms of $\cut{\ell}$ of linear left-hand sides $\ell$ are:
  \[  
    S 
    \defdby 
    \{ \bot,\; \fb{\funap{\fun{c}}{\funap{\fun{c}}{\bot}}}{\bot},\; 
       \fb{\bot}{\funap{\fun{c}}{\funap{\fun{c}}{\bot}}},\;
       \funap{\fun{c}}{\funap{\fun{c}}{\bot}},\; \funap{\fun{c}}{\bot}
    \}
    %\punc,
  \]
  and closure of $S$ under $\smerge$ yields the algebra 
  $\aalg = S \join \{\fb{\funap{\fun{c}}{\funap{\fun{c}}{\bot}}}{{\funap{\fun{c}}{\funap{\fun{c}}{\bot}}}}\}$
  with the interpretation function defined by:
  \begin{align*}
    \funap{\interpret{\ta}}{x} 
    & = %%  \shrink{\funap{\ta}{x}}{\aalg} 
    \funap{\interpret{\tb}}{x} = \interpret{\fun{e}} = \bot  
    & \text{for all  $x \in \aset$}
    %\\
    %\funap{\interpret{\tb}}{x} 
    %& = \bot
    %& \text{for all  $x \in \aset$}
    %\\
    %\interpret{e} 
    %& = \bot
    \\
    \funap{\interpret{\fun{c}}}{\funap{\fun{c}}{\bot}} 
    & = \funap{\fun{c}}{\funap{\fun{c}}{\bot}}
    \\
    \funap{\interpret{\fun{c}}}{\funap{\fun{c}}{\funap{\fun{c}}{\bot}}} 
    & %% = \shrink{\funap{c}{\funap{c}{\funap{c}{\bot}}}}{\aalg} 
      = \funap{\fun{c}}{\funap{\fun{c}}{\bot}}
    \\
    \funap{\interpret{\fun{c}}}{x} 
    & = \funap{\fun{c}}{\bot}
    & \text{for all  $x \not\in \{ \funap{\fun{c}}{\bot} , \funap{\fun{c}}{\funap{\fun{c}}{\bot}} \}$}
    \\
    \bfunap{\interpret{\tf}}{\funap{\fun{c}}{\funap{\fun{c}}{\bot}}}{\funap{\fun{c}}{\funap{\fun{c}}{\bot}}}
    & = \fb{\funap{\fun{c}}{\funap{\fun{c}}{\bot}}}{\funap{\fun{c}}{\funap{\fun{c}}{\bot}}}
    \\
    \bfunap{\interpret{\tf}}{\funap{\fun{c}}{\funap{\fun{c}}{\bot}}}{x}
    & = \fb{\funap{\fun{c}}{\funap{\fun{c}}{\bot}}}{\bot}
    & \text{for all  $x \ne \funap{\fun{c}}{\funap{\fun{c}}{\bot}}$}
    \\
    \bfunap{\interpret{\tf}}{x}{\funap{\fun{c}}{\funap{\fun{c}}{\bot}}}
    & = \fb{\bot}{\funap{\fun{c}}{\funap{\fun{c}}{\bot}}}
    & \text{for all  $x \ne \funap{\fun{c}}{\funap{\fun{c}}{\bot}}$}
    \\
    \bfunap{\interpret{\tf}}{x}{y} 
    & = \bot 
    & \text{otherwise}
  \end{align*}
  For the family of $\sisredex{}$ functions we obtain:
  \begin{align*}
    \isredex{\tf}{x,y} & = \true & \text{for all $x,y \in \aset$}
    \\
    \isredex{\ta}{\fb{\funap{\fun{c}}{\funap{\fun{c}}{\bot}}}{\bot}} & = \true
    \\
    \isredex{\ta}{\fb{\funap{\fun{c}}{\funap{\fun{c}}{\bot}}}{\funap{\fun{c}}{\funap{\fun{c}}{\bot}}}} & = \true
    \\
    \isredex{\tb}{\fb{\bot}{\funap{\fun{c}}{\funap{\fun{c}}{\bot}}}} & = \true 
    \\
    \isredex{\tb}{\fb{\funap{\fun{c}}{\funap{\fun{c}}{\bot}}}{\funap{\fun{c}}{\funap{\fun{c}}{\bot}}}} & = \true
  \end{align*}
  and $\sisredex{}$ returns $\false$ in all remaining cases.
  Finally, note that the core of the constructed algebra is the algebra itself,
  i.e., $\ralgconstr{\atrs}_c = \ralgconstr{\atrs}$,
  because $\fun{e} \in \asig$ with $\interpret{\fun{e}} = \bot$.
\end{example}

The following theorem states that \sralg{} (\cralg{}) redex-algebras %for a TRS $\atrs$
recognize only redex positions (at least all redex positions).
\begin{theorem}\label{thm:talgebra}
  Let $\atrs$ be a TRS over $\asig$. The following properties hold:
  \begin{enumerate}[\em(i)]
    \item\label{talgebra:left} The \sralg{} redex-algebra for $\atrs$ is sound.
    \item\label{talgebra:full} The \cralg{} redex-algebra for $\atrs$ is complete.
    \item\label{talgebra:quasi} 
          Let $\aalg$ be the \sralg{} redex-algebra for $\atrs$.
          For all $t \in \term{\asig}{\setemp}$ we have 
          $t \in \lang{\aalg}$
          if and only if\, $t$ is a redex with respect to a left-linear rule in $\atrs$.
          Hence, if $\atrs$ is quasi-left-linear, then the \sralg{} redex-algebra 
          for $\atrs$ is sound and complete.
  \end{enumerate}
\end{theorem}
\begin{proof}
  % By induction over the term structure.
  We prove (ii) and leave (i) and (iii) to the reader.

  Let $\pair{\aalg}{\sisredex{}} = \ralgconstr{\atrs}_c$ 
  with $\sralgconstr$ the mapping defined in Definition~\ref{def:talgebra}.
  Let $t\in\ter{\asig}{\setemp}$ be a redex with respect to a rule ${\ell \to r} \in \atrs$.
  
  We show $\match{\cut{\subtrmat{\ell}{p}}}{\interpret{\subtrmat{t}{p}}} = \true$ 
  for all positions $p\in\pos{\ell}\setminus\{\posemp\}$
  by induction on the structure of $\ell$.
  If $\subtrmat{\ell}{p}$ is a variable then $\cut{\subtrmat{\ell}{p}} = \bot$, 
  and we have $\match{\bot}{a} =\true$ for all $a\in\aalg$.
  If $\subtrmat{\ell}{p} = \funap{\fun{f}}{s_1,\ldots,s_n}$,
  then $\subtrmat{t}{p} = \funap{\fun{f}}{t_1,\ldots,t_n}$ 
  and by induction hypothesis we have 
  $\match{\cut{s_i}}{\interpret{t_i}} = \true$.
  Hence $\match{\cut{\funap{\fun{f}}{s_1,\ldots,s_n}}}{\funap{\fun{f}}{\interpret{t_1},\ldots,\interpret{t_n}}} = \true$
  by definition of $\scut$.
  Moreover, $\match{a}{t}=\true$ implies $\match{a}{\shrink{t}{\aalg}}=\true$ for all $a\in\aalg$.
  By definition we have
  $\interpret{\funap{\fun{f}}{t_1,\ldots,t_n}} = \shrink{\funap{\fun{f}}{\interpret{t_1},\ldots,\interpret{t_n}}}{\aalg}$,
  and hence
  $\match{\cut{\funap{\fun{f}}{s_1,\ldots,s_n}}}{\interpret{\funap{\fun{f}}{t_1,\ldots,t_n}}} = \true$.
  
  Let $t = \funap{\fun{g}}{u_1,\ldots,u_m}$ and $\ell = \funap{\fun{g}}{w_1,\ldots,w_m}$.
  Then we know $\match{\cut{w_i}}{\interpret{u_i}} = \true$.
  Hence
  $\match{\cut{\ell}}{\funap{\fun{g}}{\interpret{u_1},\ldots,\interpret{u_m}}} = \true$
  and $t\in\lang{\aalg}$. %$\isredex{\fun{g}}{}$
\end{proof}

\section{Minimizing Algebras}\label{sec:minimizing}

In this section we are concerned with the minimization of redex-algebras.
The algorithm is similar to the minimization of deterministic tree automata, see~\cite{tata:07}.
For the set of 291 TRSs of
%, on which 
the outermost termination competition of 2008~\cite{term:comp:08}, % has been run,
the redex-algebras constructed according to Definition~\ref{def:talgebra} 
have an average size of $4.6$ elements.
After an application of the minimization algorithm described here, 
the average size falls to $3.4$, a reduction of $27\%$.
This reduction has a polynomial influence on the number of rules of the transformed system.

\begin{definition}
  Two core redex-algebras $\aalg_1$, $\aalg_2$ 
  are called \emph{equivalent} if \mbox{$\lang{\aalg_1} = \lang{\aalg_2}$}.
\end{definition}

\begin{lemma}\label{lem:equivalent}
  Let $\aalg_1$, $\aalg_2$ be equivalent, core redex-algebras.
  Then $\aalg_1$ is sound or complete if and only if $\aalg_2$ has the respective property.
  \qed
\end{lemma}
For a given core redex-algebra we now construct a minimal equivalent algebra.
The difference to the minimization of tree automata from~\cite{tata:07}
lies in the initial equivalence~$\classes_0$.
For tree automata this initial equivalence 
consists of two partitions, the final and the non-final states.
In our setting two states are initially equivalent
if they cannot be distinguished using the $\sisredex{}$ functions, that is,
$\isredex{\tf}{\vec{x},a,\vec{y}} = \isredex{\tf}{\vec{x},b,\vec{y}}$
for each symbol $\tf \in \asig$ and each assignment of $\vec{x}$ and $\vec{y}$.
This can yield any number of partitions between $1$ and $\setsize{\aset}$.

\begin{definition}[Minimization of redex-algebra]\label{def:minimization}
  Let $\pair{\aalg}{\sisredex{}}$ be a core redex-algebra over $\asig$.
  We define equivalence relations~$\classes_i$ for $i \in \nat$ on the elements of $\aalg$.
  Initially two elements $a,b \in \aalg$ are equivalent, $a \mathrel{\classes_0} b$, 
  if: 
  \[
    \isredex{\tf}{\vec{x},a,\vec{y}} = \isredex{\tf}{\vec{x},b,\vec{y}}
  \]
  for all symbols $\tf \in \asig_n$, $j \in \{1,\ldots,n\}$,
  $\vec{x} \in \aalg^{j-1}$, and $\vec{y} \in \aalg^{n-j}$.
  Then for $i = 0,1,\ldots$ and $a, b \in \aalg$ we define $a \mathrel{\classes_{i+1}} b$
  if $a \mathrel{\classes_i} b$ holds and:
  \[
    \funap{\interpret{\tf}}{\vec{x},a,\vec{y}} 
    \mathrel{\classes_i} 
    \funap{\interpret{\tf}}{\vec{x},b,\vec{y}}
  \]
  for all $n$-ary symbols $\tf \in \asig$, $j \in \{1,\ldots,n\}$, $\vec{x} \in \aalg^{j-1}$, 
  and $\vec{y} \in \aalg^{n-j}$.
  The process halts when $\classes_{i+1} = \classes_i$ for some $i \in \nat$,
  and then we define $\classes = \classes_i$.
  Let $\classof{a}$ denote the equivalence class of $a \in \aalg$ with respect to~$\classes$.
  The \emph{minimized redex-algebra of $\aalg$}, denoted $\algmin{\aalg}$,
  is defined as
  $\algmin{\aalg} = \triple{\classes}{\sinterpret^\classes}{\sisredex{}^\classes}$
  where:
  \begin{align*}
    \funap{\interpret{\tf}^\classes}{\classof{a_1},\ldots,\classof{a_n}} 
    = \classof{\f{a_1,\ldots,a_n}}
    \\
    \funap{\sisredex{\tf}^\classes}{\classof{a_1},\ldots,\classof{a_n}} 
    = \isredex{\tf}{a_1,\ldots,a_n}
  \end{align*}
  for each symbol $\tf \in \asig$ with arity $n$.
\end{definition}

\begin{lemma}\label{lem:minimization}
  Let $\aalg$ be a core redex-algebra, 
  then $\aalg$ is equivalent to $\algmin{\aalg}$.
  \qed
\end{lemma}

\begin{example}
  We consider the TRS $\atrs$ consisting of the following three rules:
  \begin{align*}
    \f{\funap{\fun{i}}{\ta}} &\to \ta &
    \f{\funap{\fun{j}}{\ta}} &\to \ta &
    \f{\ta} &\to \ta
  \end{align*}
  %
  %over $\asig = \{f,i,j,a\}$.
  The \sralg{} (\cralg) redex-algebra for $\atrs$ is
  $\aalg = \{\ta,\funap{\fun{i}}{\ta},\funap{\fun{j}}{\ta},\bot\}$ % together
  with the interpretation
  $\interpret{\ta} = \ta$, 
  $\funap{\interpret{\fun{i}}}{\ta} = \funap{\fun{i}}{\ta}$,
  $\funap{\interpret{\fun{j}}}{\ta} = \funap{\fun{j}}{\ta}$,
  and the interpretation is $\bot$ in all non-listed cases;
  $\isredex{\tf}{x} = \true$ for all $x \ne \bot$,
  and $\false$, otherwise.

  The minimization algorithm starts with $\classes_0 = \{\{\ta,\funap{\fun{i}}{\ta},\funap{\fun{j}}{\ta}\},\{\bot\}\}$
  as initial equivalence, since $\bot$ can be distinguished from the other elements
  ($\isredex{\tf}{\bot} = \false$).
  The first iteration of the algorithm yields
  $\classes_1 = \{\{\ta\},\{\funap{\fun{i}}{\ta},\funap{\fun{j}}{\ta}\},\{\bot\}\}$
  as $\funap{\interpret{\fun{i}}}{\ta} = \funap{\fun{i}}{\ta}$ 
  whereas $\funap{\interpret{\fun{i}}}{\funap{\fun{i}}{\ta}} = \funap{\interpret{\fun{i}}}{\funap{\fun{j}}{\ta}} = \bot$.
  The elements $\funap{\fun{i}}{\ta}$ and $\funap{\fun{j}}{\ta}$ 
  are indistinguishable, and so in the second iteration we obtain $\classes_2 = \classes_1$.
  Thus the elements $\funap{\fun{i}}{\ta}$ and $\funap{\fun{j}}{\ta}$ are identified
  and we obtain an algebra that has one element less than the algebra we started with.
\end{example}

\section{Constructing Minimal and Maximal \clabeling{s}}\label{sec:minmax}

In the previous sections we have constructed and minimized 
redex-algebras for recognizing redex positions.
For completing the transformation we still have to explain 
how to construct \clabeling{s} from the redex-algebras.

In minimal labeling symbols are marked with a $\star$ if they correspond to redex positions
and stay unlabeled otherwise.
This labeling creates a small signature and thereby results 
in a small number of rules of the transformed system.
\begin{definition}\label{def:minlabel}
  Let $\atrs$ be a TRS over $\asig$, and $\aalg$ a redex-algebra.
  The \emph{minimal labeling with respect to $\aalg$}
  is the \clabeling{} $\triple{\aalg}{\semlab}{\asigred}$ 
  defined for each $n$-ary symbol $\tf \in \asig_{\topsymb}$ by:
  \begin{align*}
    \labelf{\tf}{a_1,\ldots,a_{n}} =
    \begin{cases}
       \star    &\text{if $\isredex{\tf}{a_1,\ldots,a_{n}} = \true$} \\
       \lstemp  &\text{otherwise}
    \end{cases}
  \end{align*}
  The set of redex symbols is defined by $\asigred = \{\svoodoolabel{\tf}{\star} \where \tf \in \asig\}$.
\end{definition}
\begin{theorem}\label{thm:minsound}
  Let $\atrs$ be a TRS, and $\aalg$ a sound redex-algebra for $\atrs$.
  The minimal labeling with respect to $\aalg$ is a sound \clabeling{} for $\atrs$.
\end{theorem}
\begin{proof}
  Let $t = \f{t_1,\ldots,t_{n}} \in \term{\asig}{\setemp}$ such that $\rootsymb{\dolabelg{t}} \in \asigred$.
  Then by definition of minimal labeling 
  we have $\isredex{\tf}{\interpret{t_1},\ldots,\interpret{t_n}} = \true$.
  Hence $t\in\lang{\aalg}$ and thus $t$ is a redex by definition of sound redex-algebra.
\end{proof}

The construction and minimization of redex-algebras 
(Definitions~\ref{def:talgebra} and~\ref{def:minimization})
give rise to sound minimal \clabeling{s} (Theorem~\ref{thm:talgebra},
Lemmas~\ref{lem:minimization} and~\ref{lem:equivalent}, and Theorem~\ref{thm:minsound}).
In combination with Theorems~\ref{thm:cxtext:sound} and~\ref{thm:dynlab:sound}
this provides us with sound transformations for proving outermost termination:
\begin{corollary}
  Let $\atrs$ be a TRS,
  and let $\triple{\aalg}{\semlab}{\asigred}$ be the minimal labeling 
  with respect to the minimized \sralg{} redex-algebra for $\atrs$.
  Then $\atrs$ is outermost ground terminating whenever
  the dynamic context extension $\cxtext{\aalg}{\semlab}{\atrs}$ 
  or the dynamic labeling $\dynlab{\aalg}{\semlab}{\atrs}$ is terminating.
\end{corollary}

Minimal labeling is sound and efficient, but it is not complete 
(not even for left-linear TRSs where the \sralg{} redex-algebra is complete):
\begin{example}\label{ex:noncomplete}
  Let $\atrs$ be the term rewriting system 
  consisting of the rules:
  \begin{align*}
    \infinite{x} &\to \cons{x}{\infinite{\suc{x}}} 
    & \cons{\suc{x}}{y} &\to \fun{nil}
    %\tag{$\atrs$}
    %\label{ex:inf}
    %\punc.
  \end{align*}
  Obviously, $\atrs$ is outermost terminating.
  The minimized \sralg{} redex-algebra for~$\atrs$ is:
  \begin{align*}
  \aalg = \{ \ssuc , \bot \}
  &&
  \funap{\interpret{\ssuc}}{x} = \ssuc 
  &&
  \interpret{\fun{nil}} = \funap{\interpret{\sinfinite}}{x} = \bfunap{\interpret{\scons}}{x}{y} = \bot
  \end{align*}
  for all $x,y\in\aalg$.
%  $\aalg = \{\ssuc,\bot\}$ with the interpretation function defined by
%  $\funap{\interpret{\ssuc}}{x} = \ssuc$, 
%  $\funap{\interpret{\sinfinite}}{x} = \bfunap{\interpret{\scons}}{x}{y} = \bot$
%  for all $x,y\in\aalg$,
%  and $\interpret{\mathit{nil}}=\bot$.
  The \cdepth{} of both rules (with respect to $\aalg$) is $0$.
  Using minimal labeling we obtain
  $\labelf{\scons}{\ssuc,x} = \star$ and $\labelf{\sinfinite}{x} = \star$  for all $x \in \aset$,
  and other symbols are left unmarked~($\epsilon$).
  Thus the set of redex symbols is $\asigred = \{ \rmark{\sinfinite} , \rmark{\scons} \}$.
  
  The dynamic context extension $\cxtext{\aalg}{\semlab}{\atrs}$ of $\atrs$ 
  with respect to the \clabeling~$\triple{\aalg}{\semlab}{\asigred}$ consists of the following rules,
  the first two of which 
  arise from the $\sinfinite$-rule,
  with the values $\bot$ and $\ssuc$ assigned to $x$ respectively:
  \begin{align*}
    \funap{\svoodoolabel{\sinfinite}{\star}}{x} 
    &\to \cons{x}{\funap{\svoodoolabel{\sinfinite}{\star}}{\suc{x}}}
    \\
    \funap{\svoodoolabel{\sinfinite}{\star}}{x} 
    &\to \bfunap{\svoodoolabel{\scons}{\star}}{x}{\funap{\svoodoolabel{\sinfinite}{\star}}{\suc{x}}}
    \\
    \bfunap{\svoodoolabel{\scons}{\star}}{\suc{x}}{y} 
    &\to \fun{nil}
  \end{align*}
  The replacement map is defined by 
  $\amumap{\svoodoolabel{\sinfinite}{\star}} = \amumap{\svoodoolabel{\scons}{\star}} = \setemp$. 
  
  Now $\cxtext{\aalg}{\semlab}{\atrs}$ admits an infinite derivation:
  \begin{gather*}
    \funap{\svoodoolabel{\sinfinite}{\star}}{x} 
    \to_{\cxtext{\aalg}{\semlab}{\atrs},\samumap}
    \cons{x}{\funap{\svoodoolabel{\sinfinite}{\star}}{\suc{x}}}
    \to_{\cxtext{\aalg}{\semlab}{\atrs},\samumap}
    \cons{x}{\cons{\suc{x}}{\funap{\svoodoolabel{\sinfinite}{\star}}{\suc{\suc{x}}}}} \to \ldots
  \end{gather*}
  The third term is labeled incorrectly, as the inner occurrence of $\scons$ should be marked. % with $\star$.
  The reason is that in the second step, instead of 
  the first $\sinfinite^\star$-rule, the second
  %~\eqref{rule:cons}, rule~\eqref{rule:consstar} 
  should have been applied;
  however, the left-hand side $\funap{\svoodoolabel{\sinfinite}{\star}}{x}$ contains too little information
  to `decide' what the labeling of the right-hand side should be.
\end{example}

This motivates the use of maximal labeling for which
correct labeling is preserved under rewriting.
Function symbols are labeled with the interpretation of their arguments:
\begin{definition}\label{def:maxlabel}
  Let $\atrs$ be a TRS over $\asig$, and let $\aalg$ be a redex-algebra for $\atrs$.
  The \emph{maximal labeling with respect to $\aalg$}
  is the \clabeling{} $\triple{\aalg}{\semlab}{\asigred}$ 
  defined for each $n$-ary $\tf \in \asig_{\topsymb}$ by:
  \begin{align*}
    \labelf{\tf}{a_1,\ldots,a_n} = \tuple{a_1,\ldots,a_n}
  \end{align*}
  The set of redex symbols is defined by:
  $\asigred = \{ \svoodoolabel{\tf}{\tuple{a_1,\ldots,a_n}} \where \isredex{\tf}{a_1,\ldots,a_n} = \true \}$.
\end{definition}

\begin{theorem}\label{thm:maxcomplete}
  Let $\atrs$ be a TRS, and let $\aalg$ be a redex-algebra for $\atrs$.
  The maximal labeling with respect to $\aalg$ is a maximal \clabeling{} for $\atrs$,
  and it is sound, complete, and core whenever $\aalg$ has the respective property.
\end{theorem}
\begin{proof}
  Maximality of the maximal labeling is immediate from the definition.
  Let $\aalg$ be a complete redex-algebra.
  Let $t = \f{t_1,\ldots,t_{n}} \in \term{\asig}{\setemp}$ be a redex.
  Then by definition of complete redex-algebra
  $\isredex{\tf}{\interpret{t_1},\ldots,\interpret{t_n}} = \true$,
  and it follows that $\rootsymb{\dolabelg{t}} = \svoodoolabel{\tf}{\tuple{\interpret{t_1},\ldots,\interpret{t_n}}} \in \asigred$.
  Hence the \clabeling{} is complete.
  Analogous to the proof of Theorem~\ref{thm:minsound},
  we obtain that maximal labeling is sound whenever the redex-algebra $\aalg$ is sound.
  Note that the remaining claim concerning coreness is immediate by definition.
\end{proof}

The construction and minimization of redex-algebras 
(Definitions~\ref{def:talgebra} and~\ref{def:minimization})
give rise to sound and complete maximal \clabeling{s} (Theorem~\ref{thm:talgebra},
Lemmas~\ref{lem:minimization} and~\ref{lem:equivalent}, and Theorem~\ref{thm:maxcomplete}).
In combination with Theorems~\ref{thm:cxtext:sound}, 
\ref{thm:abstract:maxcomplete} and~\ref{thm:dynlab:sound} 
this provides us with sound transformations for proving
outermost termination for arbitrary TRSs.
For quasi-left-linear TRSs dynamic context extension is both sound and complete.
%
%Definition~\ref{def:talgebra} together with 
%Theorems~\ref{thm:talgebra},~\ref{thm:abstract:maxcomplete},
%~\ref{thm:cxtext:sound},~\ref{thm:dynlab:sound}
%and Lemmas~\ref{lem:minimization},~\ref{lem:equivalent}
%lead to the sound, and complete transformations for proving and 
%disproving outermost ground termination for arbitray TRSs, respectively.
%For quasi-left-linear TRSs dynamic context extension is both sound and complete.
%
\begin{corollary}
  Let $\atrs$ be a TRS,
  and let $\triple{\aalg}{\semlab}{\asigred}$ be the maximal labeling for the minimized \sralg{} redex-algebra for $\atrs$.
  Then $\atrs$ is outermost ground terminating whenever
  the dynamic context extension $\cxtext{\aalg}{\semlab}{\atrs}$ 
  or the dynamic labeling $\dynlab{\aalg}{\semlab}{\atrs}$ is terminating.
  Moreover, if $\atrs$ is quasi-left-linear, 
  then $\atrs$ is outermost ground terminating 
  if and only if the dynamic context extension 
  $\cxtext{\aalg}{\semlab}{\atrs}$ terminates.
\end{corollary}

As a consequence, the full redex-algebra for an arbitrary TRS 
can be used to disprove outermost ground termination:
\begin{corollary}
  Let $\atrs$ be a TRS,
  and let $\triple{\aalg}{\semlab}{\asigred}$ be the maximal labeling 
  for the minimized \cralg{} redex-algebra for $\atrs$.
  Then the dynamic context extension $\cxtext{\aalg}{\semlab}{\atrs}$ is terminating
  whenever $\atrs$ is outermost ground terminating.
\end{corollary}

\begin{example}\label{ex:complete}
  We revisit Example~\ref{ex:noncomplete},
  but this time we give the dynamic context extension with respect to  
  \emph{maximal} labeling:
  \begin{align*}
    \funap{\svoodoolabel{\sinfinite}{\bot}}{x} \to 
    \bfunap{\svoodoolabel{\scons}{\bot,\bot}}{x}{&\funap{\svoodoolabel{\sinfinite}{\ssuc}}{\funap{\svoodoolabel{\ssuc}{\bot}}{x}}} &
    \funap{\svoodoolabel{\sinfinite}{\ssuc}}{x} \to 
    \bfunap{\svoodoolabel{\scons}{\ssuc,\bot}}{x}{&\funap{\svoodoolabel{\sinfinite}{\ssuc}}{\funap{\svoodoolabel{\ssuc}{\ssuc}}{x}}} \\
    \bfunap{\svoodoolabel{\scons}{\ssuc,\bot}}{\funap{\svoodoolabel{\ssuc}{\bot}}{x}}{y} &\to \fun{nil} &
    \bfunap{\svoodoolabel{\scons}{\ssuc,\bot}}{\funap{\svoodoolabel{\ssuc}{\ssuc}}{x}}{y} &\to \fun{nil} \\
    \bfunap{\svoodoolabel{\scons}{\ssuc,\ssuc}}{\funap{\svoodoolabel{\ssuc}{\bot}}{x}}{y} &\to \fun{nil} &
    \bfunap{\svoodoolabel{\scons}{\ssuc,\ssuc}}{\funap{\svoodoolabel{\ssuc}{\ssuc}}{x}}{y} &\to \fun{nil}
  \end{align*}
  with $\amumap{\svoodoolabel{\sinfinite}{\bot}} = \amumap{\svoodoolabel{\sinfinite}{\ssuc}} = 
        \amumap{\svoodoolabel{\scons}{\ssuc,\bot}} = \amumap{\svoodoolabel{\scons}{\ssuc,\ssuc}} = \setemp$.
  This context-sensitive TRS is terminating
  as opposed to the one constructed in Example~\ref{ex:noncomplete}.
  To prove termination we give a strictly decreasing polynomial interpretation over the natural numbers:
  \begin{align*}
    \interpret{\fun{nil}} &= 0
    & \bfunap{\interpret{\svoodoolabel{\scons}{\bot,\bot}}}{x}{y} &= x + y
    & \funap{\interpret{\svoodoolabel{\sinfinite}{\bot}}}{x} &= x + 3
    & \funap{\interpret{\svoodoolabel{\ssuc}{\bot}}}{x} &= x \\
    &
    & \bfunap{\interpret{\svoodoolabel{\scons}{\ssuc,\bot}}}{x}{y} &= 1
    & \funap{\interpret{\svoodoolabel{\sinfinite}{\ssuc}}}{x} &= 2
    & \funap{\interpret{\svoodoolabel{\ssuc}{\ssuc}}}{x} &= x
  \end{align*}
\end{example}

\section{Evaluation}\label{sec:evaluation}

%%Empirical evidence for the efficieny of the transformation by dynamic context extension

With the implementation of the transformation by dynamic context extension,
described in Section~\ref{sec:cxtext},
the termination prover \jambox~\cite{jambox} gained first place % in proving termination 
in the category of outermost rewriting of the termination competition
of 2008~\cite{term:comp:08}, see Table~\ref{table:score}.
\begin{table}[h!!]
\begin{center}
  %\vspace{-2ex}
  { \renewcommand{\arraystretch}{1.3}
    \extracolsep{1cm}
    \setlength{\arraycolsep}{1cm}
  \begin{tabular}{|@{\hspace{1ex}\extracolsep{1ex}}c@{\hspace{1ex}\extracolsep{1ex}}|c@{\hspace{1ex}\extracolsep{1ex}}|c@{\hspace{1ex}\extracolsep{1ex}}|}
    \hline
      & score & average time\\
    \hline
    \jambox~\cite{jambox} & 72 (93.5\%) & 4.1s\\
    \hline
    \trafo~\cite{raff:zant:09}  & 46 (59.7\%) & 8.1s\\
    \hline
    \aprove~\cite{gies:schn:thie:06} & 27 (35.0\%) & 10.8s\\
    \hline
  \end{tabular} }
  \caption{Results of proving outermost termination in the %termination 
  competition of 2008~\cite{term:comp:08}.}
  \label{table:score}
  %\vspace{-6ex}
\end{center}
\end{table}
With an average time of $4.1$ seconds per termination proof, 
\jambox{} was also faster than the other participants,
providing empirical evidence for the efficiency of the transformation of dynamic context extension.
Not listed in Table~\ref{table:score} is \ttt{}~\cite{korp:ster:zank:midd:09}, 
which did not prove outermost termination,
but performed best in \emph{dis}proving outermost termination.

The percentages listed in Table~\ref{table:score}
are % to be understood 
relative to the total number of term rewriting systems 
which were proven to be outermost terminating 
by some participating tool. %in the competition.
The TPDB~2008 contained 291~TRSs in the category of outermost rewriting
of which 77 were proven outermost terminating,
161 not outermost terminating,
and 53 remained unsolved in the competition.
We note that around 50 systems in the database are, strictly speaking, not term rewriting systems,
as they contain variables in the right-hand sides that do not occur in the left-hand sides.

In the termination competition 2008,
\jambox{} used exclusively the approach of dynamic context extension (Section~\ref{sec:cxtext}).
If we additionally use dynamic labeling, as defined in Section~\ref{sec:dynlab},
the score of \jambox{} improves by $4$, 
thus proving $76$ systems to be outermost terminating.

The secret behind the efficiency of \jambox{} is threefold:
First, we construct and minimize the algebras employed for marking redex positions,
see Sections~\ref{sec:constructing} and~\ref{sec:minimizing}.
Secondly, we try two labeling strategies: minimal and maximal, see Section~\ref{sec:minmax}.
Minimal labeling is very efficient and contributes to 75\% of the success of \jambox. 
In order to have a complete transformation we also employ maximal labeling.
Thirdly, dynamic context extension is the combination of labeling and context extension,
where the prefixing of contexts to rules depends on the interpretation of the variables.
%This is formalized in Section~\ref{sec:cxtext}.
All these optimizations minimize the number of rules and their size in the transformed systems,
which is important to keep a manageable search space.

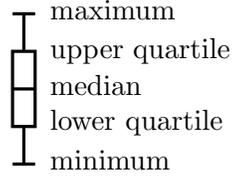
\begin{figure}[h]%{0.24\textwidth}
\vspace{-3.5ex}\hspace*{2ex}
\begin{tikzpicture}[very thick,node distance=2mm]
  \coordinate (u) at (0mm,20mm);
  \coordinate (75) at (0mm,15mm);
  \coordinate (m) at (0mm,10mm);
  \coordinate (25) at (0mm,5mm);
  \coordinate (b) at (0mm,0mm);
  \draw (u) -- (b);
  \draw ($(u)-(1.5mm,0mm)$) -- ($(u)+(1.5mm,0mm)$);
  \draw ($(b)-(1.5mm,0mm)$) -- ($(b)+(1.5mm,0mm)$);
  \draw [fill=white] ($(25)-(1.5mm,0mm)$) rectangle ($(75)+(1.5mm,0mm)$);
  \draw ($(m)-(1.5mm,0mm)$) -- ($(m)+(1.5mm,0mm)$);
  \node [right of=b,anchor=west,yshift=.5mm] {minimum};
  \node [right of=25,anchor=west,yshift=.5mm] {lower quartile};
  \node [right of=m,anchor=west,yshift=.5mm] {median};
  \node [right of=75,anchor=west] {upper quartile};
  \node [right of=u,anchor=west,yshift=.5mm] {maximum};
\end{tikzpicture}
\vspace{-.5ex}
\end{figure}
Next, we compare the performance of 
dynamic context extension and dynamic labeling (Section~\ref{sec:dynlab}).
%the different methods proposed in this paper.
Figure~\ref{fig:size} shows the size of the transformed systems
in relation to the size of the input system, as measured on the TPDB~\cite{term:comp:08}.
For each input size we display the minimum, the lower quartile (25th percentile), the median, 
the upper quartile (75th percentile), and the maximum size of the transformed systems.
From Figure~\ref{fig:size} it can be inferred that 
for larger input systems the dynamic labeling usually is 
a factor 5 or 10 smaller than the dynamic context extension.
For systems with more than 10 rules there are only a few examples available in the database,
which explains why some of the quartiles fall together.

%deteriorated 

%
\begin{figure}[h!]
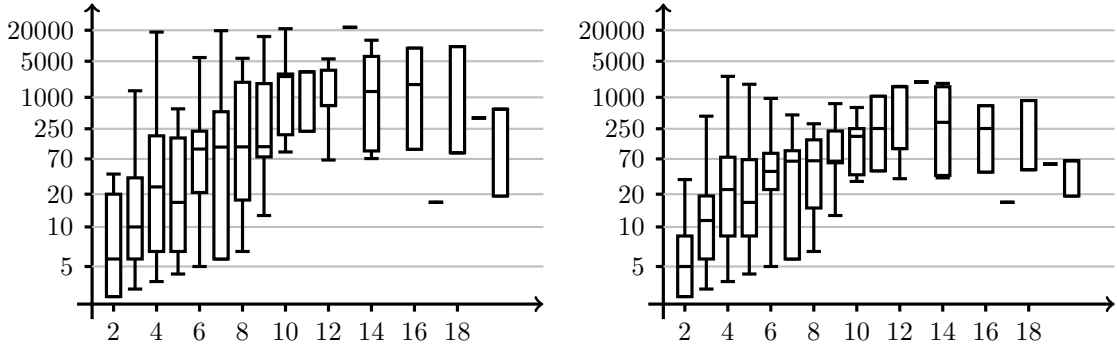

  \begin{center}
  \begin{minipage}{.49\textwidth}
  \input{eval_context_maximal}
  % \caption{Size of transformed systems using dynamic context extension with maximal labeling.}
  \end{minipage}
  \begin{minipage}{.49\textwidth}
  \input{eval_relabel_maximal}
  % \caption{Size of transformed systems using dynamic labeling with maximal labeling.}
  \end{minipage}
  \caption{Size of the transformed systems ($y$-axis) in relation to the size of the input TRS (\mbox{$x$-axis}) 
  using dynamic context extension (left), and dynamic labeling (right), both with maximal \clabeling{s}.}
  \label{fig:size}
  \end{center}
\end{figure}

\begin{table}[h!]
\renewcommand{\arraystretch}{1.25}
  \begin{tabular}{|c|c|c|c|c|c|}
   \hline
   Method & Total Score
   & $\neg \cxtext{\aalg}{\semlab}{\atrs}$, max 
   & $\neg \dynlab{\aalg}{\semlab}{\atrs}$, max 
   & $\neg \cxtext{\aalg}{\semlab}{\atrs}$, min 
   & $\neg \dynlab{\aalg}{\semlab}{\atrs}$, min \\
   \hline
   $\cxtext{\aalg}{\semlab}{\atrs}$, max & 71 & 0 & 5 & 16 & 25 \\
   \hline
   $\dynlab{\aalg}{\semlab}{\atrs}$, max & 69 & 3 & 0 & 16 & 21 \\
   \hline
   $\cxtext{\aalg}{\semlab}{\atrs}$, min & 57 & 2 & 4 & 0 & 9 \\
   \hline
   $\dynlab{\aalg}{\semlab}{\atrs}$, min & 50 & 4 & 2 & 2 & 0 \\
   \hline
  \end{tabular}
  \caption{Comparison of the proposed methods on the TPDB~2008.}
  \label{table:compare}
\end{table}

Table~\ref{table:compare} shows a comparison of our different methods 
(dynamic context extension and dynamic labeling, combined with minimal 
or maximal labeling).
Each row lists the total score of one method with the number of systems it can solve
that cannot be solved by the method corresponding to the column.
For example, the value 3 in Table~\ref{table:compare} means that 
three systems can be solved by dynamic maximal labeling,
but not by dynamic context extension in combination with maximal labeling.
The table shows that dynamic context extension and dynamic labeling 
are roughly equal in strength,
and that maximal labeling gives the best results. 

Table~\ref{table:ratio} illustrates that 
the \cdepth{} of rules is typically small (employing left-linear redex-algebras):
it is $0$ or $1$ in $94.5\%$ of the cases.
Note that $54\%$ of the rules have \cdepth~$0$,
but this does not mean that the same percentage of the TRSs could be handled by a model 
(Definition~\ref{def:model}). Only $14\%$ of the TRSs have \cdepth{} $0$.

\begin{table}[h!]
\renewcommand{\arraystretch}{1.25}
  \begin{tabular}{|c|c|c|c|c|c|c|c|c|c|c|c|}
   \hline
   \cdepth{} & 0 & 1 & 2 & 3 & 4 & 5 & 6 & 7 & 8 & 9 & 10 \\
   \hline
   \#rules & 54\% & 40.5\% & 3.3\% & 1.3\% & 0.5\% & 0.2\% & 0\% & 0\% & 0\% & 0.1\% & 0\%\\
   \hline
  \end{tabular}
  \caption{Ratio of rules having a certain \cdepth{} in the TPDB~2008.}
  \label{table:ratio}
\end{table}

\section{Discussion}\label{sec:discussion}

For arbitrary TRSs the transformation based on dynamic context extension
(including the construction of \clabeling{s}) is sound, 
and for quasi-left-linear TRSs it is sound and complete. %(see Theorem~\ref{thm:maxcomplete}).
The sound redex-algebra we construct recognizes redexes with respect to left-linear rules.
As a consequence, in the \csTRS{} $\cxtext{\aalg}{\semlab}{\atrs}$ rewriting 
is forbidden only inside such redex positions.
This corresponds to a weakening of the outermost rewriting strategy:
contraction of a redex is disallowed
only if it is contained within a redex with respect to left-linear rule.
Let us call this the `left-linear outermost' rewriting strategy.
Dynamic context extension combined with maximal labeling is sound and complete for termination 
with respect to this rewriting strategy for all TRSs.

In a similar way the transformation of~\cite{raff:zant:09}
can be generalized from quasi-left-linear TRSs to arbitrary TRSs.
For soundness the anti-matching rules do not need to exactly match the non-redex terms,
as long as at least all non-redex terms are matched.
Then the symbol $\msf{down}$ can be moved inside redexes with respect to rules which are not left-linear.
This enables only additional rewrite steps but does not harm soundness.
More precisely, using this generalization the transformation of \cite{raff:zant:09} 
becomes complete with respect to left-linear outermost termination.
Thereby the score of \trafo{} in the termination competition of 2008~\cite{term:comp:08}
could possibly have been improved by around~$20\%$, resulting in a score of 57 instead of 47.

We have shown that the transformation of dynamic labeling is complete 
on the set of correctly labeled terms $\dolabelg{\ter{\asig}{\setemp}}$
without the auxiliary $\ssrelabel$ symbols (Theorem~\ref{thm:dynlab:complete}).
The non-completeness with respect to termination on all terms arises from
`illegally placed' $\ssrelabel$ symbols in combination with duplicating rules, see Example~\ref{ex:dynlab:noncomplete}.
The duplicating rules can multiply the illegal symbols and make them reusable over and over again.
To prevent this, one can introduce an extra symbol $\msf{block}$ 
with $\amumap{\msf{block}} = \setemp$
for disallowing $\ssrelabel$ symbols 
beneath the rule application. 
For this purpose, we wrap each duplicated variable in the right-hand side of a labeled rule
into a context $\funap{\msf{block}}{\cxthole}$,
and we extend the dynamic labeling with rules of the form
$\funap{\msf{block}}{\f{x_1,\ldots,x_n}} 
\to \f{\funap{\msf{block}}{x_1},\ldots,\funap{\msf{block}}{x_n}}$
for each symbol $\tf\in\labelsig{\asig_n}$ (excluding $\ssrelabel$ symbols!).
Note that this implies that $\msf{block}$ symbols disappear 
when meeting a constant.
For instance, reconsider the TRS from Example~\ref{ex:dynlab:noncomplete},
which had among others the following rule in its dynamic labeling: 
\begin{align*}
  \bfunap{\svoodoolabel{\tf}{\mit{bc},\bot}}{\fun{c}}{y} &\to \funap{\svoodoolabel{\sh}{\bot}}{\bfunap{\svoodoolabel{\tf}{\bot,\bot}}{y}{y}}
\end{align*}
This rule would be modified to:
\begin{align*}
  \bfunap{\svoodoolabel{\tf}{\mit{bc},\bot}}{\fun{c}}{y} &\to \funap{\svoodoolabel{\sh}{\bot}}{\bfunap{\svoodoolabel{\tf}{\bot,\bot}}{\funap{\msf{block}}{y}}{\funap{\msf{block}}{y}}}
\end{align*}
In this way we `block' each duplicated variable.

Another question is whether there are interesting labelings between minimal and maximal.
In particular, are there more efficient complete labelings?
Here efficiency is measured in the size of the signature 
and the number of rules of the transformed system.
In Example~\ref{ex:complete} it would have been sufficient 
to label $\scons$ with the interpretation of the left argument,
saving two symbols and two rules of the transformed system.

\bibliographystyle{alpha}
\bibliography{main}

\end{document}